\documentclass[11pt,a4paper]{article}
\usepackage{amsfonts}
\usepackage{graphicx}
\usepackage{indentfirst}
\usepackage{amsmath,amssymb,amsthm}
\usepackage{amssymb}
\usepackage{latexsym}
\usepackage{natbib}
\usepackage[margin=1in,a4paper]{geometry}
\usepackage{colortbl}
\usepackage{color}
\usepackage[colorlinks,citecolor=blue,urlcolor=blue]{hyperref}
\usepackage{mathtools}
\usepackage{setspace}

\usepackage{algorithm}
\usepackage{algorithmic}

\usepackage{amsopn}
\DeclareMathOperator{\diag}{diag}

\mathtoolsset{showonlyrefs=true}

\numberwithin{equation}{section}
\newtheorem{proposition}{Proposition}[section]

\newtheorem{lemma}{Lemma}[section]

\newtheorem{assumption}{Assumption}
\newtheorem{theorem}{Theorem}[section]
\theoremstyle{remark}
\newtheorem{remark}{Remark}[section]

\allowdisplaybreaks[4]

\title{A General Approach for Parisian Stopping Times under Markov Processes}

\author{Gongqiu Zhang\thanks{School of Science and Engineering, The Chinese University of Hong Kong, Shenzhen, China. Email: zhanggongqiu@cuhk.edu.cn.} \and Lingfei Li\thanks{Corresponding author. Department of Systems Engineering and Engineering Management, The Chinese University of Hong Kong, Hong Kong SAR. Email: lfli@se.cuhk.edu.hk.}}

\date{\today}

\begin{document}
	
\maketitle

\begin{abstract}
We propose a method based on continuous time Markov chain approximation to compute the distribution of Parisian stopping times and price Parisian options under general one-dimensional Markov processes. We prove the convergence of the method under a general setting and obtain sharp estimate of the convergence rate for diffusion models. Our theoretical analysis reveals how to design the grid of the CTMC to achieve faster convergence. Numerical experiments are conducted to demonstrate the accuracy and efficiency of our method for both diffusion and jump models. To show the versality of our approach, we develop extensions for multi-sided Parisian stopping times, the joint distribution of Parisian stopping times and first passage times, Parisian bonds and for more sophisticated models like regime-switching and stochastic volatility models. 

\end{abstract}




\section{Introduction}

Consider a stochastic process $X$ and define
\begin{align}
	g^-_{L, t} = \sup \{ s \le t: X_s \ge L  \},\ g^+_{L, t} = \sup \{ s \le t: X_s \le L  \}.
\end{align}
for a given level $L$. The ages of excursion of $X$ below and above $L$ are given by $t - g^-_{L, t}$ and $t - g^+_{L, t}$, respectively. Consider the following two stopping times:
\begin{align}
	\tau_{L, D}^+  = \inf\{ t \ge 0: 1_{\{ X_t > L \}} (t - g^+_{L, t}) \ge D \},\ \tau_{L, D}^-  = \inf\{ t \ge 0: 1_{\{ X_t < L \}} (t - g^-_{L, t}) \ge D \},  
\end{align}
which are called Parisian stopping times. $\tau_{L, D}^+$ ($\tau_{L, D}^-$) is the first time that the age of excursion of $X$ above (below) $L$ reaches $D$. In these definitions, we use the convention that $\sup\emptyset = 0$ and $\inf\emptyset = \infty$.

Parisian stopping times appear in important finance and insurance applications. Parisian options are proposed in \cite{chesney1997brownian} as an alternative to barrier options and they have the advantage of being less prone to market manipulation (\cite{labart2010parisian}). These options are activated or cancelled if the process stays long enough below or above a prespecified level. In the insurance literature, the Parisian stopping time for the excursion below the solvency level is used to define the ruin event (e.g., \cite{dassios2008parisian}, \cite{czarna2011ruin}, \cite{loeffen2013parisian}).

There is an extensive literature on the study of Parisian stopping times and their applications in finance and insurance with many papers focused on specific models. For the Brownian motion and the geometric Brownian motion model, see \cite{chesney1997brownian}, \cite{bernard2010new}, \cite{dassios2010perturbed}, \cite{dassios2011double}, \cite{zhu2013pricing}, \cite{dassios2013parisian}, \cite{dassios2016joint}, \cite{le2016analytical}, \cite{dassios2017analytical} and \cite{dassios2018recursive}.
For models with jumps, methods for Parisian option pricing were developed in
\cite{albrecher2012pricing} for Kou's double-exponential jump-diffusion model and in \cite{chesney2018parisian} for the hyper-exponential jump-diffusion model. \cite{czarna2011ruin} and \cite{loeffen2013parisian} derived the ruin probability for spectrally negative L\'evy processes. In addition, \cite{chesney2006american}, \cite{le2016analytical} and \cite{lu2018pricing} studied American-style Parisian option pricing under the Black-Scholes model while \cite{chesney2018parisian} dealt with the hyper-exponential jump-diffusion model. As for the hedging of Parisian options, \cite{kim2016risk} developed a quasi-static strategy and tested it under the double-exponential jump-diffusion model. 

Several numerical methods have also been applied to the Parisian option pricing problem. \cite{baldi2000pricing} developed a Monte Carlo simulation method based on large deviation estimates for general diffusions. For the Black-Scholes model, \cite{avellaneda1999pricing} proposed a trinomial tree approach while \cite{haber1999pricing} developed an explicit finite difference scheme to numerically solve the pricing PDE, which is not unconditionally stable. These two methods can also be applied to general diffusions, but they need to incorporate the current length of excursion as an additional dimension. Furthermore, they may require a very small time step to achieve a target accuracy level, making them inefficient. If the model contains jumps, the lattice method and the finite difference approach will face more challenges for accuracy and efficiency. 

In this paper, we aim to develop a computationally efficient approach for Parisian stopping time problems that is applicable to general Markov models. We consider a one-dimensional (1D) Markov process $X$ living on an interval $I\subseteq\mathbb{R}$ with its infinitesimal generator $\mathcal{G}$ given by 
\begin{align}\label{eq:model-specification}
	\mathcal{G} f(x)=\frac{1}{2} \sigma^{2}(x) f^{\prime \prime}(x)+\mu(x) f^{\prime}(x)+\int_{\mathbb{R}}\left(f(x+z)-f(x)-z 1_{\{|z| \leq 1\}} f^{\prime}(x)\right) \nu(x, d z)
\end{align}
for $f \in C_c^2(I)$ (the space of twice continuously differentiable functions on $I$ with compact support) with $\int_\mathbb{R} (z^2 \wedge 1) \nu(x, dz) < \infty$. The quantities $\mu(x)$, $\sigma(x)$, $\nu(x, dz)$ are known as the drift, volatility and jump measure of $X$, respectively. We would like to obtain the distribution of $\tau_{L, D}^-$ and $\tau_{L, D}^+$, i.e., $\mathbb{P}_x[\tau_{L, D}^\pm \le t]$. The Laplace transform of the probability as a function of $t$ is given by
\begin{align}
	\int_{0}^{\infty} e^{-q t} \mathbb{P}_x[\tau_{L, D}^\pm \le t] dt = \mathbb{E}_x\left[ \int_{\tau_{L, D}^\pm}^{\infty} e^{-qt} dt \right] = \frac{1}{q} \mathbb{E}_x\left[ e^{-q \tau_{L, D}^\pm} \right],\ \Re(q) > 0.
\end{align}

For a general Markov process, there is no analytical solution for this transform. Our idea is approximating $X$ by a continuous time Markov chain (CTMC). We will show later that the Laplace transform of $\tau_{L, D}^\pm$ under a CTMC model is the solution to a linear system which involves the distribution and the Laplace transform of a first passage time of the CTMC, and this solution can be obtained in closed form. To recover the distribution of $\tau_{L, D}^\pm$, we simply invert its Laplace transform numerically using the efficient algorithm of \cite{abate1992fourier}. Using the result for $\tau_{L, D}^\pm$, we can compute the Laplace transfrom of a Parisian option price as a function of maturity in closed form for a CTMC model. We can further solve a range of Parisian problems under a CTMC model such as the distribution of the multi-sided Parisian stopping time, the joint distribution of a Parisian stopping time and a first passage time and the pricing of Parisian bonds. For more complex models like regime-switching models and stochastic volatility models, we can also approximate them by a CTMC and solve the Parisian stopping time problems. Thus, our approach is very general and versatile. Numerical experiments for various popular models further demonstrate the computational efficiency of our method, which can outpeform traditional numerical methods significantly. 

Our paper fits into a burgeoning literature on CTMC approximation for option pricing. \cite{mijatovic2013continuously} used it for pricing European and barrier options under general 1D Markov models. They showed that the barrier option price can be obtained in closed form using matrix exponentials under a CTMC model. Sharp converegnce rate analysis of their algorithm for realistic option payoffs is a challenging problem. Answers are provided for diffusion models first in \cite{li2018error} for uniform grids and then in \cite{zhang2019analysis}) for general non-uniform grids while the problem still remains open for general Markov models with jumps. For diffusions, \cite{zhang2019analysis} also developed principles for designing grids to achieve second order convergence whereas the algorithm only converges at first order in general. The idea of CTMC approximation has also been exploited to develop efficient methods for pricing various types of options for 1D Markov models: \cite{eriksson2015american} for American options, \cite{cai2015general}, \cite{song2018computable}, \cite{cui2018single} for Asian options, and \cite{zhangli2021} and \cite{zhang2021Amerdrawdown} for maximum drawdown options. Furthermore, CTMC approximation can be extended to deal with more complex models 
where the asset price follows a regime-switching process (\cite{CaiKouSongRS}) or exhibits stochastic volatility (\cite{cui2017var,cui2018general,cui2019general}). 

One can view the present paper as a generalization of \cite{mijatovic2013continuously} and \cite{zhang2018analysis} from barrier options to Parisian options. However, this generalization is by no means trivial. First, the Parisian problem is much more complex than the first passage problem as it involves the excursion of the process. Second, there is considerable challenge in sharp convergence rate analysis. Even for barrier options, the sharp convergence rate of CTMC approximation is only known for diffusion models, and the difficulty mainly stems from the lack of smoothness in the option payoffs. In this paper, we focus on analyzing the convergence rate of our algorithm for Parisian problems under diffusion models. We use the spectral analysis framework in \cite{zhang2019analysis}, but there is an important difference which makes the analysis for Parisian problems much harder. Here, we must carefully analyze the dependence of solutions to Sturm-Liouville and elliptic boundary value problems on variable boundary levels whereas in the barrier option problem the boundaries are fixed. Furthermore, we need to establish some new identities that enable us to obtain the sharp rate. Our analysis not only yields the sharp estimate of the convergence rate, but also unveils how the grid of the CTMC should be designed to achieve fast convergence. 

For the sake of simplicity, our discussions in the main text will be focused on the algorithms for single-barrier Parisian stopping times and the down-and-in Parisian option as well as their convergence analysis. Extensions for other cases are provided in the appendix. The rest of this paper is organized as follows. In Section \ref{sec:Parisian-CTMC}, we derive the Laplace transforms of Parisian stopping times and Parisian option prices under a general CTMC model with finite number of states. Section \ref{sec:CTMC} reviews how to construct a CTMC to approximate a general jump-diffusion and establishes the convergence of this approximation for the Parisian problems. Section \ref{sec:conv-rate} analyzes the convergence rate for diffusion models and proposes grid design principles to ensure second order convergence. Section \ref{sec:numerical-ex} shows numerical results for various models and Section \ref{sec:conclusions} concludes the paper. In Appendix \ref{sec:extensions}, we extend our method to deal with multi-sided Parisian stopping times, the joint distribution of Parisian stopping times and first passage times, Parisian bonds and the Parisian problem for regime-switching and stochastic volatility models. Proofs for the convergence rate analysis are collected in Appendix \ref{sec:proofs}.

\section{Parisian Stopping Times Under a CTMC}\label{sec:Parisian-CTMC}
Consider a CTMC denoted by $Y$ which lives in the state space $\mathbb{S}_n=\{y_0,y_1,\cdots,y_n\}$. The transition rate matrix of $Y$ is denoted by ${\pmb G}$ with ${\pmb G}(x, y)$ indicating the transition rate from state $x$ to $y$, $x \ne y$ and ${\pmb G}(x, x) = -\sum_{y \in \mathbb{S}_n \backslash \{ x \}} {\pmb G}(x, y)$. The action of $\pmb{G}$ on a function $g: \mathbb{S}_n \to \mathbb{R}$ is defined as
\begin{align}
	{\pmb G} g(x) = \sum_{y \in \mathbb{S}_n} {\pmb G}(x, y) g(y),\ x \in \mathbb{S}_n.
\end{align}
For any $\xi$, we also introduce the notations
\begin{align} 
&\xi^+ = \min \{ z>\xi: z \in \mathbb{S}_n \}, \ \text{for}\ \xi < y_n, \\
&\xi^- = \max \{ z<\xi: z \in \mathbb{S}_n \}, \ \text{for}\ \xi > y_0,
\end{align}
which are the CTMC states that are the right and left neighbor of $\xi$, respectively. We also define $L^+ = \min\{ z \ge L: z \in \mathbb{S}_n \}$. Note that $L^+$ is different from $\xi^+$ in that $L^+=L$ if $L\in\mathbb{S}_n$, but we always have $\xi^+>\xi$. We slightly abuse the notation here to avoid introducing a new one.  

In this section, we first derive the Laplace transform of Parisian stopping times and then the Laplace transform of the price of a down-and-in Parisian option.

\subsection{Laplace Transform of Parisian Stopping Times}
Define 
\begin{align}
	h_n(q, x; y) := \mathbb{E}_x\left[e^{-q\tau_{L, D}^{-}} 1_{\{ Y_{\tau_{L, D}^{-}}  = y \}} \right],\ x, y \in \mathbb{S}_n,
\end{align} 
which is the Laplace transform of $\tau_{L, D}^{-}$ given $Y_0=x$ with an indicator for the process landing on $y$ at $\tau_{L, D}^{-}$. The joint distribution of $\tau_{L, D}^{-}$ and $Y_{\tau_{L, D}^{-}}$ is characterized by $h(q, x; y)$. We also define 
\begin{align}
	h_n(q, x) := \sum_{y \in \mathbb{S}_n} h_n(q, x; y) = \mathbb{E}_x\left[ e^{-q \tau_{L, D}^{-}} \right],\ x \in \mathbb{S}_n,
\end{align}
which is the Laplace transform of $\tau_{L, D}^{-}$ given $Y_0=x$. To calculate $h_n(q, x; y)$ and $h_n(q, x)$, we introduce the following quantities that involve the first passage times of $Y$:
\begin{align}
	&\tau_L^+:=\inf\{t: Y_t\ge L\},\ \tau_L^-:=\inf\{t: Y_t<L\},\\
	&v_n(D, x; y) := \mathbb{E}_x\left[ 1_{\{ \tau_L^+ \ge D \}} 1_{\{ Y_D = y \}} \right],\\
	&v_n(D, x) := \sum_{y \in \mathbb{S}} v_n(D, x; y) = \mathbb{E}_x\left[ 1_{\{ \tau_L^+ \ge D \}} \right],\\
	&u^+_n(q, D, x; z) := \mathbb{E}_x\left[ e^{-q \tau_L^+} 1_{\{ \tau_L^+ < D \}} 1_{\{ Y_{\tau_L^+} = z   \}}\right],\\
	&u^-_n(q, x; z) := \mathbb{E}_x\left[ e^{-q \tau_L^-} 1_{\{ Y_{\tau_L^-} = z \}} \right].
\end{align}

\begin{remark}
	Throughout the paper, the variable before the semicolon indicates the initial state and the variable after it indicates the state at a stopping time or fixed time. 
\end{remark}

Using the strong Markov property of $Y$, we can derive the following two equations for $h_n(q, x; y)$ and $h_n(q, x)$ using the first passage quantities.

\begin{theorem}\label{thm:h-linear-system}
	For $x \in \mathbb{S}_n$,
	\begin{align}
		h_n(q, x; y)& = 1_{\{ x < L \}} e^{-qD} v_n(D, x; y) + 1_{\{ x < L \}} \sum_{z \ge L}  u^+_n(q, D, x; z) h_n(q, z; y) \label{eq:hqxy}\\
		&\quad + 1_{\{x \ge L\}} \sum_{z < L}  u^-_n(q, x; z) h_n(q, z; y), \label{eq:hqxy-2} \\
		h_n(q, x)& = 1_{\{ x < L \}} e^{-qD} v_n(D, x) + 1_{\{ x < L \}}  \sum_{z \ge L} u^+_n(q, D, x; z) h_n(q, z) \label{eq:hqx-1} \\
		&\quad + 1_{\{x \ge L\}} \sum_{z < L}  u^-_n(q, x; z) h_n(q, z). \label{eq:hqx-2}
	\end{align}	
\end{theorem}

\begin{proof}
First, consider the case $x<L$. We have
\begin{equation}
	h_n(q,x;y) = \mathbb{E}_x\left[e^{-q\tau_{L, D}^{-}} 1_{\{ Y_{\tau_{L, D}^{-}}  = y \}} 1_{\{\tau^+_{L}\ge D\}}\right] +  \mathbb{E}_x\left[e^{-q\tau_{L, D}^{-}} 1_{\{ Y_{\tau_{L, D}^{-}}  = y \}} 1_{\{\tau^+_{L}< D\}}\right].
\end{equation}

The first term represents the case that $Y$ does not cross $L$ from below by time $D$. In this case, $\tau_{L, D}^{-} = D$ and hence $h_n(q, x; y)=e^{-qD}v_n(D, x; y)$, where $v_n(D, x; y)$ is the probability that the first time $Y$ greater than or equal to $L$ exceeds $D$ and $Y_D = y$. 

The second term stands for the case that $Y$ crosses $L$ from below before time $D$. Using the strong Markov property, the process restarts at $\tau_L^{+}$ from $Y_{\tau_L^{+}}$ and the clock $t - g_{L, t}^{-}$ is reset to zero. Thus, the second term becomes 
\begin{align}
	&\sum_{z \ge L}\mathbb{E}_x\left[e^{-q\tau_{L}^{+}}1_{\{\tau^+_{L}< D\}}1_{\{ Y_{\tau_{L}^{+}} = z \}}\mathbb{E}\left[e^{-q\tau_{L, D}^{-}} 1_{\{ Y_{\tau_{L, D}^{-}}  = y \}} \Big| \tau^+_{L}< D, Y_{\tau_L^{+}} = z \right]\right]\\
	&=\sum_{z \ge L}\mathbb{E}_x\left[e^{-q\tau_{L}^{+}}1_{\{\tau^+_{L}< D\}}1_{\{ Y_{\tau_{L}^{+}} = z \}}h_n(q,z;y)\right]=\sum_{z \ge L}  u^+_n(q, D, x; z) h_n(q, z; y).
\end{align}

Next, consider the case $x\ge L$. The clock $t - g_{L, t}^{-}$ would not tick until $Y$ crosses $L$ from above, at which time the process restarts at $Y_{\tau_{L}^{-}}$ by the strong Markov property and $t - g_{L, t}^{-}$ starts ticking. Thus, we have
\begin{align}
	h_n(q,x;y)&=\sum_{z < L}\mathbb{E}_x\left[e^{-q\tau_{L}^{-}}1_{\{ Y_{\tau_{L}^{-}} = z \}}\mathbb{E}\left[e^{-q\tau_{L, D}^{-}} 1_{\{ Y_{\tau_{L, D}^{-}}  = y \}} \Big|  Y_{\tau_L^{-}} = z \right]\right]\\
	&=\sum_{z<L}\mathbb{E}_x\left[e^{-q\tau_{L}^{-}}1_{\{ Y_{\tau_{L}^{-}} = z \}}h_n(q,z;y)\right]=\sum_{z < L}  u^-_n(q, x; z) h_n(q, z; y).
\end{align}
Putting these results together yields the equation for $h_n(q,x;y)$. The equation for $h_n(q,x)$ is obtained by summing the equation for $h_n(q,x;y)$ over $y$. 
\end{proof}

Solving \eqref{eq:hqxy-2} and \eqref{eq:hqx-2} requires the quantities for first passage times, which we show are solutions to some equations.
\begin{proposition}\label{prop:wv-equation}
	We have the following equations for $v_n(D, x; y)$, $u^+_n(q, D, x; z)$ and $u^-_n(q, x; z)$.
	\begin{enumerate}
		\item $v_n(D, x; y)$ satisfies
		\begin{align}\label{eq:vn-dxy}
			\begin{cases}
				\frac{\partial v_n}{\partial D}(D, x; y) = {\pmb G} v_n(D, x; y),\ D > 0,\ x \in \mathbb{S}_n \cap (-\infty, L),\\
				v_n(D, x; y) = 0,\ D > 0,\ x \in \mathbb{S}_n \cap [L, \infty),\\
				v_n(0, x; y) = 1_{\{ x = y \}},\ x \in \mathbb{S}_n.
			\end{cases}
		\end{align}
		
		\item $u_n^+(q, D, x; z)$ satisfies
		\begin{align}\label{eq:un-p}
			\begin{cases}
				\frac{\partial u^+_n}{\partial D}(q, D, x; z) = ({\pmb G} - q{\pmb I}) u^+_n(q, D, x; z),\ D > 0,\ x \in \mathbb{S}_n \cap (-\infty, L),\\
				u^+_n(q, D, x; z) = 1_{\{  x = z \}},\ D > 0,\ x \in \mathbb{S}_n \cap [L, \infty),\\
				u^+_n(q, 0, x; z) = 0,\ x \in \mathbb{S}_n.
			\end{cases}
		\end{align}
		
		
			Here $u^+_n(q, D, x; z)$ can be written as $u^+_{1, n}(q, x; z) - u^+_{2, n}(q, D, x; z)$ with $u^+_{1, n}(q, x; z)$ and $u^+_{2, n}(q, D, x; z)$ as solutions to the following equations:
		\begin{align}
			&\begin{cases}
				({\pmb G} - q{\pmb I}) u_{1, n}^+(q, x; z) = 0,\ x \in \mathbb{S}_n \cap (-\infty, L),\\
				u_{1, n}^+(q, x; z) = 1_{\{ x = z \}},\ x \in \mathbb{S}_n \cap [L, \infty).
			\end{cases}\\
			&\begin{cases}
				\frac{\partial u^+_{2,n}}{\partial D}(q, D, x; z) = ({\pmb G} - q{\pmb I}) u_{2, n}^+(q, D, x; z),\ D > 0,\ x \in \mathbb{S}_n \cap (-\infty, L),\\
				u_{2, n}^+(q, D, x; z) = 0,\ D > 0,\ x \in \mathbb{S}_n \cap [L, \infty),\\
				u_{2, n}^+(q, 0, x; z) = u_{1, n}^+(q, x; z),\ x \in \mathbb{S}_n.
			\end{cases}
		\end{align}
		
		\item $u^-_n(q, x; z)$ satisfies
		\begin{align}
			\begin{cases}\label{eq:un-m}
				({\pmb G} - q{\pmb I}) u^-_n(q, x; z) = 0,\ x \in \mathbb{S}_n \cap [L, \infty),\\
				u^-_n(q, x; z) = 1_{\{ x = z \}},\ x \in \mathbb{S}_n \cap (-\infty, L).
			\end{cases}
		\end{align}
	\end{enumerate}
\end{proposition}

\begin{proof}
We derive the equation for $v_n(D, x; y)$ and the others can be obtained in a similar manner. Let $\mathbb{T}_\delta = \{ i\delta: i = 0, 1, 2, \cdots \}$ and $\tau_L^{\delta,+} = \inf\{ t \in \mathbb{T}_\delta: Y_t \ge L \}$. As $Y$ is a piecewise constant process, we have $\tau_L^{\delta,+} \downarrow \tau_L^{+}$ as $\delta \to \infty$. Using the monotone convergence theorem, we have
\begin{align}\label{eq:vn-delta-to-vn}
	v_{n,\delta}(D, x; y) := \mathbb{E}_x\left[ 1_{\{ \tau_L^{\delta,+} \ge D \}} 1_{\{ Y_D = y \}} \right] \to v_n(D, x; y).
\end{align}
Denote the transition probability of $Y$ by $ p_n(\delta, x, z)$. We have that $p_n(\delta, x, z) = {\pmb G}(x, z) \delta + o(\delta)$ for $z \ne x$ and $ p_n(\delta, x, x) = 1 + {\pmb G}(x, x) \delta $. Then for $x < L$,   
\begin{align}
	v_{n,\delta}(D, x; y) &= \sum_{z \in \mathbb{S}_n} p_n(\delta, x, z) v_{n,\delta}(D-\delta, z; y) \\
	&= ({\pmb I} + {\pmb G}\delta + o(\delta)) v_{n,\delta}(D-\delta, x; y).
\end{align}
Therefore,
\begin{align}\label{eq:vn-delta}
	\frac{v_{n,\delta}(D, x; y) - v_{n,\delta}(D-\delta, x; y)}{\delta} = {\pmb G}v_{n,\delta}(D-\delta, x; y) + o(1). 
\end{align}
Taking $\delta$ to $0$ shows the ODE. The boundary and initial conditions ara obvious. 
\end{proof}

We introduce some matrices and vectors:  ${\pmb I}$ is the $n$-by-$n$ identity matrix. Let ${\pmb I}_L^+ = \operatorname{diag}(1_{\{ x \ge L \}})_{x \in \mathbb{S}_n}$,  ${\pmb I}_L^- = \operatorname{diag}(1_{\{ x < L \}})_{x \in \mathbb{S}_n}$, where diag means creating a diagonal matrix with the given vector as the diagonal. ${\pmb e}$ is the $n$-dimensional vector of ones, and ${\pmb e}_{y} = (1_{\{ x = y \}})_{x \in \mathbb{S}_n}$.

It is not difficult to see that the equations in Proposition \ref{prop:wv-equation} can be solved analytically and the solutions are expressed using matrix operations as follows: for $D>0$,
\begin{align}
	&{\pmb V} = \left( v_n(D, x; y) \right)_{x, y \in \mathbb{S}_n} = \exp\left( {\pmb I}_L^- {\pmb G} D \right)\pmb{I}_L^-, \label{eq:v-matrix}\\
	&{\pmb U}^+_1(q) = \left( u^+_{1, n}(q, x; z) \right)_{x, z \in \mathbb{S}_n} = (q{\pmb I}_L^- - {\pmb I}_L^- {\pmb G} + {\pmb I}_L^+ )^{-1} {\pmb I}_L^+, \label{eq:u1-matrix}\\
	&{\pmb U}^+_2(q) = \left( u^+_{2, n}(q,D, x; z) \right)_{x, z \in \mathbb{S}_n} = e^{-qD}\exp\left( {\pmb I}_L^- {\pmb G}  D \right) {\pmb I}_L^- {\pmb U}^+_1(q), \label{eq:u2-matrix} \\
	&{\pmb U}^-(q) = \left( u^-_n(q, x; z) \right)_{x, z \in \mathbb{S}_n} = (q{\pmb I}_L^+ - {\pmb I}_L^+ {\pmb G} + {\pmb I}_L^- )^{-1} {\pmb I}_L^-. \label{eq:um-matrix}
\end{align}

Let
\begin{align}
	&{\pmb H}(q) = \left( h_n(q, x; y) \right)_{x, y \in \mathbb{S}_n}, \\
	& {\pmb U}(q) = \pmb I^-_L {\pmb U}_1^+(q) - \pmb I^-_L {\pmb U}_2^+(q) + \pmb I_L^+ {\pmb U}^-(q).
\end{align} 
Then, the equation in Theorem \ref{thm:h-linear-system} becomes
\begin{align}
	{\pmb H}(q) = e^{-qD}\pmb{I}_L^-{\pmb V} + {\pmb U}(q) {\pmb H}(q).
\end{align}
The solution is
\begin{align}
	{\pmb H}(q) = e^{-qD} ({\pmb I} - {\pmb U}(q))^{-1}\pmb{I}_L^-{\pmb V}. \label{eq:H-matrix}
\end{align}
Subsequently, ${\pmb h}(q) = (h(q, x))_{x \in \mathbb{S}_n}$ can be calculated as
\begin{align}
	{\pmb h}(q) = {\pmb H}(q) {\pmb e} = e^{-qD} ({\pmb I} - {\pmb U}(q))^{-1} {\pmb V} {\pmb e}, \label{eq:h-matrix}
\end{align}
where ${\pmb e}$ is a $n$-dimensional vector of ones. 

\begin{remark}[Time complexity for a general CTMC]\label{rmk:complexity-lt}
	Calculating ${\pmb V}$, ${\pmb U}_1^+(q)$, ${\pmb U}_2^+(q)$ and ${\pmb U}^-(q)$ involves matrix exponentiation or inversion whose time complexity is generally $\mathcal{O}(n^3)$. In addition, calculating ${\pmb H}(q)$ also involves matrix inversion so its cost is also $\mathcal{O}(n^3)$. Aggregating the costs over all steps yields an overall time complexity of $\mathcal{O}(n^3)$ to calculate ${\pmb H}(q)$ and ${\pmb h}(q)$.
\end{remark}


\begin{remark}[Time complexity for birth-and-death processes]\label{rmk:diffusion-case}
	Significant savings in computations are possible when $Y$ is a birth-and-death process. In this case, it is easy to see that $u_n^+(q, D, x, z) \ne 0$ only when $z = L^+$ and $u^-_n(q, x, z) \ne 0$ only when $z = L^-$. Consequently, \eqref{eq:hqxy-2} reduces to
	\begin{align}
		h_n(q, x; y) &= 1_{\{ x < L \}} e^{-qD} v_n(D, x; y) + 1_{\{ x < L \}} u_n^+(q, D, x; L^+) h_n(q, L^+; y) \\
		&\quad + 1_{\{ x \ge L \}} u^-_n(q, x; L^-) h_n(q, L^-; y).\label{eq:hqxy-diff}
	\end{align}
	Setting $x = L^+, L^-$ respectively in \eqref{eq:hqxy-diff}, we obtain
	\begin{align}
		&h_n(q, L^-; y) = v_n(D, L^-; y) + u^+_n(q, D, L^-; L^+) h_n(q, L^+; y), \label{eq:hLLm-1} \\
		&h_n(q, L^+; y) = u^-_n(q, L^+; L^-) h_n(q, L^-; y). \label{eq:hLLm-2}
	\end{align}
	Solving this linear system for $h_n(q, L^-; y)$ and $h_n(q, L^+; y)$ yields
	\begin{align}
		&h_n(q, L^-; y) = \frac{e^{-qD} v_n(D, L^-; y)}{1 - u^-_n(q, L^+; L^-) u^+_n(q, D, L^-; L^+)}, \label{eq:hLLm-11} \\
		&h_n(q, L^+; y) = \frac{e^{-qD} u^-_n(q, L^+; L^-) v_n(D, L^-; y)}{1 - u^-_n(q, L^+; L^-) u^+_n(q, D, L^-; L^+)}.  \label{eq:hLLm-22} 
	\end{align}
	For a general $x\in \mathbb{S}_n$, $h_n(q, x; y)$ can be found by substituting these two equations back to \eqref{eq:hqxy-diff}.
	
	Recall that ${\pmb e}_{y} = (1_{\{ x = y \}})_{x \in \mathbb{S}_n}$. We have
	\begin{align}
		&{\pmb u}_1^+(q) = \left( u_{1, n}^+(q, D, x; L^+) \right)_{x \in \mathbb{S}_n} = (q{\pmb I}_L^- - {\pmb I}_L^- {\pmb G} + {\pmb I}_L^+ )^{-1} {\pmb e}_{L^+}, \\
		&{\pmb u}_2^+(q) = \left( u_{2, n}^+(q, D, x; L^+) \right)_{x \in \mathbb{S}_n} = \exp\left( {\pmb I}_L^- {\pmb G} D \right) {\pmb I}_L^- {\pmb u}^+_1(q),\\
		&{\pmb u}^+(q) = {\pmb u}_1^+(q) - {\pmb u}_2^+(q),\\
		&{\pmb u}^-(q) = \left( u^-_n(q, x; L^-) \right)_{x \in \mathbb{S}_n} = (q{\pmb I}_L^+ - {\pmb I}_L^+ {\pmb G} + {\pmb I}_L^- )^{-1} {\pmb e}_{L^-}.
	\end{align}
	
	Then, $u^-_n(q, L^+; L^-) = {\pmb e}_{L^+}^\top {\pmb u}^-(q)$, $u^+_n(q, D, L^-; L^+) = {\pmb e}_{L^-}^\top {\pmb u}^+(q)$ and $\left(v_n(D, L^-; y)\right)_{y \in \mathbb{S}_n} = {\pmb e}_{L^-}^\top {\pmb V} \in \mathbb{R}^{1 \times n}$. Using \eqref{eq:hLLm-11} and \eqref{eq:hLLm-22}, we get
	\begin{align}
		&{\pmb h}^-(q) = \left( h_n(q, L^-; y) \right)_{y \in \mathbb{S}_n} = \frac{e^{-qD} {\pmb e}_{L^-}^\top {\pmb V}}{1 - {\pmb e}_{L^+}^\top {\pmb u}^-(q) {\pmb e}_{L^-}^\top {\pmb u}^+(q)} \in \mathbb{R}^{1 \times n},\\
		&{\pmb h}^+(q) = \left( h_n(q, L^+; y) \right)_{y \in \mathbb{S}_n}  = \frac{e^{-qD} {\pmb e}_{L^+}^\top {\pmb u}^-(q) {\pmb e}_{L^-}^\top {\pmb V}}{1 - {\pmb e}_{L^+}^\top {\pmb u}^-(q) {\pmb e}_{L^-}^\top {\pmb u}^+(q)} \in \mathbb{R}^{1 \times n}.
	\end{align}
	The, by \eqref{eq:hqxy-diff} we obtain
	\begin{align}
		{\pmb H}(q) = e^{-qD} {\pmb I}_L^- {\pmb V} + {\pmb I}_L^- {\pmb u}^+(q) {\pmb h}^+(q) + {\pmb I}_L^+  {\pmb u}^-(q) {\pmb h}^-(q).
	\end{align}
	Subsequently, we get
	\begin{align}
		{\pmb h}(q) = {\pmb H}(q) {\pmb e} =  e^{-qD} {\pmb I}_L^- {\pmb V} {\pmb e} + {\pmb I}_L^- {\pmb u}^+(q) {\pmb h}^+(q) {\pmb e} + {\pmb I}_L^+  {\pmb u}^-(q) {\pmb h}^-(q) {\pmb e}.
	\end{align}
	Note that ${\pmb G}$ is a tridiagonal matrix and so is ${\pmb I}_L^- {\pmb G}$. Hence calculating ${\pmb u}^-(q)$ and ${\pmb u}_1^+(q)$ only costs $O(n)$ as opposed to $O(n^3)$ in the general case. Moreover, for any ${\pmb b} \in \mathbb{R}^n$, $\exp\left({\pmb I}_L^- {\pmb G} D\right) {\pmb b}$ can be calculated as,
	\begin{align}\label{eq:exp-mat-discretize}
	\exp\left({\pmb I}_L^- {\pmb G} D\right) {\pmb b} \approx  \left( {\pmb I} -  {\pmb I}_L^- {\pmb G} D/m \right)^{-m} {\pmb b},
	\end{align}  
	where the RHS can be found by solving a sequence of linear systems with the LHS given by ${\pmb I} -  {\pmb I}_L^- {\pmb G} D/m$ and the total cost is $O(mn)$. Hence the cost of computing ${\pmb u}_2^+(q)$ and ${\pmb V} {\pmb e}$ is $O(mn)$. Calculating ${\pmb h}^+(q) {\pmb e}$ and ${\pmb h}^-(q) {\pmb e}$ involves ${\pmb V} {\pmb e}$, ${\pmb u}^-(q)$ and ${\pmb u}^+(q) = {\pmb u}_1^+(q) - {\pmb u}_2^+(q)$ which can be done with complexity of either $O(mn)$ or $O(n)$. Therefore, it costs $O(mn)$ to calculate ${\pmb h}(q)$. For ${\pmb H}(q)$, we need to compute the matrix $\exp\left({\pmb I}_L^- {\pmb G} D\right)$ and it costs $O(mn^2)$ using the approximation in \eqref{eq:exp-mat-discretize}. In our implementation, we use the extrapolation technique in \cite{feng2008pricing}, which achieves good accuracy with a small $m$ ($m$ is much smaller than $n$).
 
\end{remark}


\subsection{Laplace Transform of The Option Price}

We derive the Laplace transform of the price of a Parisian down-and-in option w.r.t. maturity under the CTMC $Y$. The option price is given by 
\begin{align}
	u_n(t, x) = \mathbb{E}_x \left[ 1_{\{\tau_{L, D}^{-} \le t \}}f(Y_t) \right].
\end{align} 
\begin{theorem}\label{thm:opt-price-lt}
	The Laplace transform of $u_n(t, x)$ w.r.t. $t$ has the following representation:
	\begin{align}
		\widetilde{u}_n(q, x) = \int_{0}^{\infty} e^{-q t} u_n(t, x) dt = \sum_{z<L, z\in \mathbb{S}_n} h_n(q, x; z) \widetilde{w}_n(q, z), \Re(q) > 0.
	\end{align}
	where $\widetilde{w}_n(q, z) = \int_{0}^{\infty} e^{-q t} \mathbb{E}_z\left[ f\left( Y_t \right) \right] dt$ is the Laplace transform of the price of a European option with payoff $f$ w.r.t. maturity and it satisfies the linear system
	\begin{align}
		q\widetilde{w}_n(q, z) - {\pmb G}\widetilde{w}_n(q, z) = f(z),\ z \in \mathbb{S}_n. \label{eq:wn-linear-system}
	\end{align}
\end{theorem}

\begin{proof}
	By the strong Markov property of $Y$, we get
	\begin{align}\label{eq:option-LP-CTMC}
		\widetilde{u}_n(q, x) &= \int_{0}^{\infty} e^{-q t} u_n(t, x) dt =\int_{0}^{\infty} e^{-qt} \mathbb{E}_x \left[ 1_{\{{\tau_{L, D}^-} \le t \}}f(Y_t) \right] dt \\
		&= \int_{0}^{\infty} e^{-qt} \mathbb{E}_x \left[\mathbb{E} \left[ 1_{\{{\tau_{L, D}^-} \le t \}}f(Y_t)  | \mathcal{F}_{\tau_{L, D}^-}\right]\right] dt \\
		&= \int_{0}^{\infty}  \sum_{z \in \mathbb{S}_n} e^{-qt} \mathbb{E}_x\left[ 1_{\{ Y_{\tau_{L, D}^-} = z, {\tau_{L, D}^-} \le t \}} \mathbb{E} \left[ f(Y_t) | Y_{\tau_{L, D}^-} = z \right] \right] dt \\
		&= \sum_{z \in \mathbb{S}_n}\mathbb{E}_x \left[ \left(\int_{\tau_{L, D}^-}^{\infty} e^{-q(t - {\tau_{L, D}^-})} \mathbb{E}_z\left[ f(Y_{t - \tau_{L, D}^-}) \right] dt\right) e^{-q{\tau_{L, D}^-}} 1_{\{ Y_{{\tau_{L, D}^-}} = z \}} \right] \\
		&= \sum_{z \in \mathbb{S}_n} \widetilde{w}_n(q, z) \mathbb{E}_x\left[ e^{-q{\tau_{L, D}^-}} 1_{\{ Y_{{\tau_{L, D}^-}} = z \}} \right] \\
		&= \sum_{z \in \mathbb{S}_n} h_n(q, x; z) \widetilde{w}_n(q, z).
	\end{align}
We can deduce $h_n(q, x; z)=0$ for $z\ge L$ from its definition, so the summation above is done only over those $z<L$.
	
	Using the expression in Theorem 3.1 of \cite{mijatovic2013continuously} for the European option price, we can calculate the Laplace transform of the option price as 
	\begin{align}
		\widetilde{{\pmb w}}(q) = \left(\widetilde{w}_n(q, z)\right)_{z \in \mathbb{S}_n} = \int_{0}^{\infty} e^{-q t} e^{{\pmb G} t} {\pmb f} dt = \left(q {\pmb I} - {\pmb G}  \right)^{-1} {\pmb f},
	\end{align}
	where ${\pmb f} = \left( f(z) \right)_{z \in \mathbb{S}_n}$. Thus, $\widetilde{w}_n(q, z)$ solves the linear system as stated. 
\end{proof}

Let $\widetilde{{\pmb u}}(q)=(\widetilde{u}_n(q,x))_{x\in\mathbb{S}_n}$. Theorem \ref{thm:opt-price-lt} shows that
\begin{equation}
	\widetilde{{\pmb u}}(q)={\pmb H}(q)\widetilde{{\pmb w}}(q).
\end{equation}
In general, the cost of computing ${\pmb H}(q)$ and $\widetilde{\pmb w}(q)$ is $\mathcal{O}(n^3)$. So the time complexity for computing $\widetilde{{\pmb u}}(q)$ is $\mathcal{O}(n^3)$. Significant savings are possible for birth-and-death processes. The cost of computing $\widetilde{\pmb w}(q)$ is only $\mathcal{O}(n)$ because $q {\pmb I} - {\pmb G}$ is diagonal. For ${\pmb H}(q) \widetilde{\pmb w}(q)$, it is given by
\begin{align}
	{\pmb H}(q) \widetilde{\pmb w}(q) = e^{-qD} {\pmb I}_L^- {\pmb V} \widetilde{\pmb w}(q) + {\pmb I}_L^- {\pmb u}^+(q) {\pmb h}^+(q) \widetilde{\pmb w}(q) + {\pmb I}_L^+  {\pmb u}^-(q) {\pmb h}^-(q) \widetilde{\pmb w}(q).
\end{align}
Computing ${\pmb V} \widetilde{\pmb w}(q)$, ${\pmb h}^+(q) \widetilde{\pmb w}(q)$ and ${\pmb h}^-(q) \widetilde{\pmb w}(q)$ involves calculating the exponential of a tridiagonal matrix times a vector, which can be done at cost $O(mn)$ using the approximation \eqref{eq:exp-mat-discretize} with $m$ time steps. Thus, the total cost for a birth-and-death process is only $O(mn)$.

\begin{remark}[Parisian ruin probability]
In an insurance risk context, $\tau_{L, D}^{-}$ is known as the Parisian ruin time (see e.g., \cite{czarna2011ruin}, \cite{loeffen2013parisian} and \cite{dassios2008parisian}) for some pre-specified $L$ and $D$. One can calculate $\mathbb{P}_x\left[ \tau_{L, D}^{-} \le t \right]$ with a finite $t$ and $\mathbb{P}_x\left[ \tau_{L, D}^{-} < \infty \right]$ for the ruin probability in a finite horizon and in an infinite horizon, respectively. 

To apply our method, set $f(x) = 1$ and thus $u_n(t, x) = \mathbb{P}_x\left[ \tau_{L, D}^{-} \le t \right]$. In this case, $\widetilde{w}_n(q, z)  = 1/q$, and hence ${\pmb u}(q) = {\pmb H}(q) {\pmb e} /q$. The finite horizon ruin probability can then be obtained by evaluating the Laplace inversion of ${\pmb u}(q) = {\pmb H}(q) {\pmb e} /q$ at $t$ as described in the next subsection. For the infinite horizon ruin probability, we apply the Final Value Theorem of Laplace transform, which implies
	\begin{align}
		\mathbb{P}_x\left[ \tau_{L, D}^{-} < \infty \right] = \lim_{q \downarrow 0} q {\pmb u}(q) = \lim_{q \downarrow 0} {\pmb H}(q) {\pmb e},
	\end{align}
The limit can be well approximated by ${\pmb H}(q) {\pmb e}$ for a sufficiently small $q$.
\end{remark}

\subsection{Numerical Laplace Inversion}

We utilize the numerical Laplace inversion method in \cite{abate1992fourier} with Euler summation to recover the distribution function of a Parisian stopping time and the price of a Parisian option. Suppose a function $g(\cdot)$ has known Laplace transform $\hat{g}(\cdot)$. Then $g(\cdot)$ can be recovered from $\hat{g}(\cdot)$ as
\begin{align}\label{eq:LapInv}
	g(t) \approx \frac{e^{A/2}}{2t} \Re\left( \hat{g}\left( \frac{A}{2t} \right) \right) + \frac{e^{A/2}}{t} \sum_{j = 1}^{k_1 + k_2} (-1)^j \left( \sum_{l = 0 \vee (j - k_1)}^{k_1} \begin{pmatrix}
		k_1 \\ l
	\end{pmatrix} 2^{-k_1} \right) \Re\left( \hat{g}\left( \frac{A + 2j\pi\mathrm{i}}{2t} \right) \right).
\end{align}
In our numerical examples, we set $k_1 = k_2 = 20$ and $A = 15$ and this choice already provides very accurate results. We summarize our algorithm for a CTMC combined with numerical Laplace inversion in Algorithm \ref{algo:option-pricing}. Modifications can be made following Remark \ref{rmk:diffusion-case} to simplify the computation for birth-and-death processes.

\begin{remark}[]\label{rmk:overall-complexity}
	Suppose $n_l$ terms are used in \eqref{eq:LapInv} for the Laplace inversion. Then the overall complexity of calculating the distribution of a Parisian stopping time or the price of a Parisian option is $O(n_l n^3)$ for a general CTMC and $O(n_l m n)$ for a birth-and-death process.
\end{remark}

\begin{algorithm}
	\caption{Calculate $u_n(t, x_0) = \mathbb{E}_x \left[ f(Y_t) 1_{\{ \tau_{L, D}^{-} \le t \}} \right]$ under a CTMC $Y_t$}
	\begin{algorithmic} 
		\REQUIRE ${\pmb G}, {\pmb x}, n, L, D, t, f(\cdot), k_1, k_2, A$
		
		\ENSURE $u(t, x_0) \leftarrow {\pmb u}_{i_0}$
		\STATE ${\pmb q} \leftarrow \pmb 0_{k_1 + k_2 + 1}$, ${\pmb w} \leftarrow \pmb 0_{k_1 + k_2 + 1}$, ${\pmb q}_1 \leftarrow \frac{A}{2t}$, ${\pmb w}_1 \leftarrow \frac{e^{A/2}}{2t}$
		\FOR{$j \in \{1, 2, \cdots, k_1 + k_2\}$}
		\STATE ${\pmb q}_{j + 1} = \frac{A + 2j\pi \mathrm{i} }{2t}$
		\STATE ${\pmb w}_{j + 1} = \frac{e^{A/2}}{t} (-\mathrm{i})^j \sum_{l = 0 \vee  (j - k_2)}^{k_1} \begin{pmatrix}
			k_1 \\ l
		\end{pmatrix} 2^{-k_1} $
		\ENDFOR	
		
		\STATE $i_0 \leftarrow \{ i: {\pmb x}_i = x_0 \}$, ${\pmb I}^- \leftarrow \operatorname{diag}\left(1_{\{ \pmb x < L \}}\right)$, ${\pmb I}^+ \leftarrow \operatorname{diag}\left(1_{\{ \pmb x \ge L \}}\right)$
		\STATE ${\pmb u} \leftarrow {\pmb 0}_n$, ${\pmb f} \leftarrow f(\pmb x)$
		\FOR{$j \in \{ 1, 2, \cdots, k_1 + k_2 + 1 \}$}
		\STATE ${\pmb V} \leftarrow \exp\left( {\pmb I}^- {\pmb G} D \right) {\pmb I}^-$
		\STATE ${\pmb U}_1 \leftarrow \left( {\pmb q}_j {\pmb I}^- - {\pmb I}^- {\pmb G} + {\pmb I}^+ \right)^{-1} {\pmb I}^+$
		\STATE ${\pmb U}_2 \leftarrow \exp\left( {\pmb I}^- {\pmb G} D \right) {\pmb I}^- {\pmb U}_1$
		\STATE ${\pmb U}^- \leftarrow ({\pmb q}_j {\pmb I}^+ - {\pmb I}^+ {\pmb G} + {\pmb I}^-)^{-1} {\pmb I}^-$
		\STATE ${\pmb U} \leftarrow {\pmb I}^- ({\pmb U}_1 - {\pmb U}_2) + {\pmb I}^+ {\pmb U}^-$
		\STATE ${\pmb u} \leftarrow {\pmb u} + {\pmb w}_j e^{-{\pmb q}_j D} ({\pmb I} - {\pmb U})^{-1} \pmb{I}^- {\pmb V} ({\pmb q}_j {\pmb I} - {\pmb G})^{-1} {\pmb f}$
		\ENDFOR
		\STATE $u_n(t, x_0) \leftarrow {\pmb u}_{i_0}$
	\end{algorithmic}
	\label{algo:option-pricing}
\end{algorithm}

\section{Markov Chain Approximation and Its Convergence}\label{sec:CTMC}
In order to make the paper self-contained, we briefly review the construction of CTMC approximation below and refer readers to \cite{mijatovic2013continuously} for more detailed discussions. We then prove its convergence in our problem. 

\subsection{Construction of CTMC Approximation}

We construct a CTMC on the space $\mathbb{S}_n = \{y_0, y_1, y_2, \cdots, y_n \}$ with $y_0 < y_1 < \cdots < y_n$ to approximate the Markov process with its generator given by \eqref{eq:model-specification}. Let $\mathbb{S}_n^o = \mathbb{S} \backslash \{ y_0, y_n \}$, which is the set of interior states. For each $y_j \in \mathbb{S}_n^o$, define $y_j^\pm := y_{j \pm 1}$, $I_{y_j} := ((y_j^- + y_j) / 2, (y_j + y_j^+)/2]$, and set $I_{y_0} = (-\infty, (y_0 + y_1)/2]$, $I_{y_n} = ((y_{n - 1} + y_n)/2, \infty)$ and $A - x = \{ y - x: y \in A \}$ for any $A \subseteq \mathbb{R}$ and $x \in \mathbb{R}$. For any $x \in \mathbb{S}_n^o$ and $y \in \mathbb{S}_n$, let
\begin{align}
	&\Lambda(x, y) := \int_{I_y - x} \nu(x, dz),\\
	&\bar{\sigma}^2(x) := \int_{I_x - x} z^2 1_{\{ |z| \le 1 \}} \nu(x, dz),\ \bar{\mu}(x) := \sum_{y \in \mathbb{S}_n \backslash x} (y - x) \int_{I_y - x} 1_{\{ |z| \le 1 \}} \nu(x, dz). 
\end{align}
Then $\mathcal{G}$ can be approximate as,
\begin{align}
	\mathcal{G} f(x) \approx \frac{1}{2}\widetilde{\sigma}^2(x) f^{\prime \prime}(x)+ \widetilde{\mu}(x) f^{\prime}(x)+\sum_{y \in \mathbb{S}_n \backslash x}\left(f\left(y\right)-f(x)\right) \Lambda\left(x, y\right)
\end{align}
where $\widetilde{\sigma}^2(x) = \sigma^{2}(x)+\bar{\sigma}^{2}(x)$, $\widetilde{\mu}(x) = \mu(x)-\bar{\mu}(x)$. For each $x \in \mathbb{S}_n^o$, let 
\begin{align}
	\delta^+ x = x^+ - x,\ \delta^- x = x - x^-,\ \delta x = (\delta^+ x + \delta^- x)/2.
\end{align}
We further approximate $f'(x)$ and $f''(x)$ using central difference. The resulting transition rate matrix ${\pmb G}$ for $Y_t^{(n)}$ is as follows. For each $x \in \mathbb{S}_n^o$, 
\begin{align}
	&{\pmb G}(x, x^+) = \frac{ \widetilde{\mu}(x) \delta^- x + \widetilde{\sigma}^2(x)}{2 \delta^+ x \delta x} + \Lambda(x, x^+),\\
	&{\pmb G}(x, x^-) = \frac{-\widetilde{\mu}(x) \delta^+ x + \widetilde{\sigma}^2(x)}{2 \delta^- x \delta x} + \Lambda(x, x^-), \\
	&{\pmb G}(x, y) = \Lambda(x, y),\ y \in \mathbb{S}_n \backslash \{ x, x^+, x^- \},\\
	&{\pmb G}(x, x) = -\sum_{z \in \mathbb{S}_n \backslash \{ x \}} {\pmb G}(x, y).
\end{align}
We set $y_0$ and $y_n$ as absorbing states for $Y^{(n)}_t$ and hence ${\pmb G}(y_0, z) = {\pmb G}(y_n, z) = 0$ for all $z \in \mathbb{S}_n$.

\subsection{Convergence}

We study the convergence of CTMC approximation for a general Markov process. Let $X$ be the original process and $\{Y^{(n)}\}$ is a sequence of CTMCs that converges to $X$ weakly on $D(\mathbb{R})$, the Skorokhod space of c\`adl\`ag real-valued functions endowed with the Skorokhod topology. Sufficient conditions for the weak convergence $Y^{(n)}\Rightarrow X$ can be found in e.g., \cite{mijatovic2013continuously}, which we do not repeat here for the sake of brevity. In the following, we prove convergence for the distribution of Parisian stopping times and Parisian option prices.

The key of the proof is to establish the continuity of $\tau_{L, D}^-: D(\mathbb{R}) \to \mathbb{R}$ on a subset of $D(\mathbb{R})$ which has probability one. To characterize such sets, we introduce the following definitions.
\begin{align}
	&\sigma_0^+(\omega) := 0,\ \sigma_k^-(\omega) := \inf\{ t > \sigma_{k-1}^+(\omega): \omega_t < L \},\ \sigma_k^+(\omega) := \inf\{ t > \sigma_k^-(\omega): \omega_t \ge L \},\ k \ge 1.
\end{align}
In this way, $\sigma_k^+(\omega) - \sigma_k^-(\omega)$ tracks the length of each excursion. We also consider the following subsets of $D(\mathbb{R})$.
\begin{align}
	&V := \{ \omega \in D(\mathbb{R}): \lim_{k \to \infty} \sigma^{\pm}_k(\omega) = \infty \},\\
	&W := \{ \omega \in D(\mathbb{R}): \sigma_k^+(\omega) - \sigma_k^-(\omega) \ne D \text{ for all } k \ge 1 \},\\
	&U^+ := \{ \omega \in D(\mathbb{R}): \inf\{ t \ge u: \omega_t \ge L \} = \inf\{ t \ge u: \omega_t > L \} \text{ for all } u \ge 0\},\\
	&U^- := \{ \omega \in D(\mathbb{R}): \inf\{ t \ge u: \omega_t \le L \} = \inf\{ t \ge u: \omega_t < L \} \text{ for all } u \ge 0\}.
\end{align}

\begin{theorem}\label{thm:conv}
	Suppose $X$ is c\`adl\`ag (right continuous with left limits), $\mathbb{P}[X \in V] = \mathbb{P}[X \in W] = \mathbb{P}[X \in U^+] = \mathbb{P}[X \in U^-] = 1$, and $Y^{(n)} \Rightarrow X$. Then, for any bounded function $f(\cdot)$ whose discontinuity points are $\mathcal{D}$ and $t > 0$ such that $\mathbb{P}[X_t \in \mathcal{D}] = 0$ and $\mathbb{P}[\tau_{L, D}^- = t] = 0$,
	\begin{align}
		\mathbb{E}\left[ 1_{\{ \tau_{L, D}^{(n), -} \le t \}} f(Y^{(n)}_t) \right] \to \mathbb{E}\left[ 1_{\{ \tau_{L, D}^{-} \le t \}} f(X_t) \right], \text{ as\ } n \to \infty.
	\end{align}
\end{theorem}
	
\begin{proof}
		$Y^{(n)}\Rightarrow X$ implies that if $g: D(\mathbb{R}) \to \mathbb{R}$ is bounded and continuous on some subset $C$ of $D(\mathbb{R})$ such that $\mathbb{P}[X \in C] = 1$, then it holds that
		\begin{align}
			\mathbb{E}[g(Y^{(n)})] \to \mathbb{E}[g(X)], \text{ as } n \to \infty.
		\end{align}
		With the assumption that $\mathbb{P}[X_t \in \mathcal{D}] = 0$, $\mathbb{P}[\tau_{L, D}^- = t] = 0$ and $\mathbb{P}[X \in V] = \mathbb{P}[X \in W] = \mathbb{P}[X \in U^+] = \mathbb{P}[X \in U^-] = 1$, it suffices to establish the continuity of $\tau_{L, D}^-(\omega)$ on $V \cap W \cap U^+ \cap U^-$. 	
		For $\omega \in V \cap W \cap U^+ \cap U^-$, suppose $\omega^{(n)} \to \omega$ as $n \to \infty$ on $D(\mathbb{R})$. Note that for $\omega\in U^+ \cap U^-$, we have $\sigma_{k+1}^-(\omega)>\sigma_k^+(\omega)>\sigma_k^-(\omega)$ for $k\ge 1$.
		\begin{itemize}
			\item If $\tau_{L, D}^-(\omega) < s$. Then there exists $k 
			\ge 1$ such that $\sigma_k^+(\omega) \wedge s - \sigma^-_k(\omega)  > D$. Let $J(\omega)=\{t: \omega(t)\ne \omega(t-)\}$. Since $\mathbb{R}^+\backslash J(\omega)$ is dense, we can find a $\varepsilon > 0$ that is small enough such that $\sigma_k^+(\omega) \wedge s - \sigma^-_k(\omega) - 2\varepsilon > D$, $\omega$ is continuous at $\sigma_k^+(\omega) \wedge s - \varepsilon$ and $\sigma^-_k(\omega)+\varepsilon$, and $\sup\{ \omega_u: \sigma^-_k(\omega) + \varepsilon \le u \le \sigma_k^+(\omega) \wedge s - \varepsilon  \} < L$. By Proposition 2.4 in  \cite{jacod2013limit} which shows the continuity of the supremum process, $\sup\{ \omega_u^{(n)}: \sigma^-_k(\omega) + \varepsilon \le u \le \sigma_k^+(\omega) \wedge s - \varepsilon  \} < L$. Since $\sigma_k^+(\omega) \wedge s - \sigma^-_k(\omega) - 2\varepsilon > D$, we conclude that $\tau_{L, D}^{ -}(\omega^{(n)}) < s$ for sufficiently large $n$. Therefore, we have $\limsup_{n \to \infty} \tau_{L, D}^{ -}(\omega^{(n)}) \le \tau_{L, D}^-(\omega)$.
			
			\item If $\tau_{L, D}^-(\omega) > s$. Since $\omega \in V$, there is $\overline{k} \ge 1$ such that $\sigma_{\overline{k}}^+(\omega) \ge s $ and $\sigma_{\overline{k}}^-(\omega) < s$. As $\omega \in W$, $\max\{ \sigma_k^+(\omega) \wedge s - \sigma_k^-(\omega) : 1\le k\le \overline{k}\} < D$. Since $\mathbb{R}^+\backslash J(\omega)$ is dense, we can find a $\varepsilon > 0$ that is small enough such that $\sigma_k^+(\omega) \wedge s - \sigma_k^-(\omega) + 2\varepsilon < D$, $\omega$ is continuous at $\sigma_k^+ \wedge s + \varepsilon$ and $\sigma_k^--\varepsilon$ for any $1 \le k \le \overline{k}$, and $\inf\{ \omega_u:  \sigma_k^+(\omega) + \varepsilon \le  u \le \sigma_{k+1}^-(\omega) -\varepsilon \} > L$ for any $1 \le k < \overline{k}$. In the same spirit of Proposition 2.4 in \cite{jacod2013limit}, we can show that the continuity of the infimum process. Hence we have for sufficiently large $n$, $\inf\{ \omega_u^{(n)}:  \sigma_k^+(\omega) + \varepsilon \le  u \le \sigma_{k+1}^-(\omega) -\varepsilon \} > L$ for any $1 \le k < \overline{k}$. It follows that the age of excursion below $L$ of $\omega^{(n)}$ up to time $s$ cannot exceed $\max\{ \sigma_k^+(\omega) \wedge s - \sigma_k^-(\omega) + 2\varepsilon : 1\le k\le \overline{k}\} < D$. This implies $\tau_{L, D}^-(\omega^{(n)}) > s$ for sufficiently large $n$. Therefore, $\liminf_{n \to \infty} \tau_{L, D}^{ -}(\omega^{(n)}) \ge \tau_{L, D}^-(\omega)$.
		\end{itemize}
		Combing the arguments above, we obtain the continuity of $\tau_{L,D}^-(\cdot)$ on $V \cap W \cap U^+ \cap U^-$. This concludes the proof.
	\end{proof}

\section{Convergence Rate Analysis for Diffusion Models}\label{sec:conv-rate}

We assume the target diffusion process $X$ lives on a finite interval $[l,r]$ with absorbing boundaries. For diffusions that are unbounded, it means that we consider their localized version. The absolute value of the localization levels can be chosen sufficiently large so that localization error is negligible compared with error caused by CTMC approximation. The latter is the focus of our analysis and the former is ignored. 

We consider a sequence of CTMCs $\{Y^{(n)}_t, n = 1, 2, \cdots\}$ with grid size $\delta_n = \max_{x \in \mathbb{S}_n^-} \delta^+ x$ ($\mathbb{S}_n^- = \mathbb{S}_n \backslash \{ y_n \}$) and fixed lower and upper levels, i.e., $y_0 = l$, $y_n = r$. Let $\mathbb{S}_n^o=\mathbb{S}_n \backslash \{ y_0, y_n \}$. We impose the following condition on $\mathbb{S}_n$.
\begin{assumption}\label{assump:grid}
	There is a constant $C > 0$ independent $n$ such that for every $\mathbb{S}_n$,
	\begin{align}
		\frac{\max_{x \in \mathbb{S}_n^-} \delta^+ x}{\min_{x \in \mathbb{S}_n^-} \delta^+ x} \le C.
	\end{align}
\end{assumption}
This assumption requires the minimum step size cannot converge to zero too fast compared with the maximum step size, which is naturally satisfied by grids used for applications.  

Let $\pmb{G}_n$ denote the generator matrix of $Y^{(n)}$. 
We specify $l$ and $r$ as absorbing boundaries for the CTMC to match the behavior of the diffusion. Throughout this section, we impose the following conditions on the drift and diffusion coefficient of $X$ and the payoff function of the option.
\begin{assumption}\label{assump:smoothness}
	Suppose $\mu(x) \in C^3([l, r]),\ \sigma(x) \in C^4([l, r]),\ \min_{x \in [l, r]} \sigma(x) > 0$. The payoff function $f(x)$ is Lipschitz continuous on $[l, r]$ except for a finite number of points $\mathcal{D} = \{ \xi_1, \xi_2, \cdots, \xi_d \}$ in $(l, r)$.
\end{assumption} 
The coefficients of many diffusion models used in financial applications are sufficiently smooth, so they satisfy Assumption \ref{assump:smoothness}. It is possible to extend our analysis to handle diffusions with nonsmooth coefficients using the results in \cite{zhang2019analysis}, but for simplicity we only consider the smooth coefficient case in this paper.

For some diffusions, $\sigma(x)$ vanishes at $x = 0$. To fit them into our framework, we can localize the diffusion at $l = \varepsilon > 0$. By choosing a very small $\varepsilon$, the approximation error is negligible. 

We introduce the notation $\mathcal{O}(g(k, n))$ for quantities bounded by $c |g(k, n)|$ for some constant $c> 0$ which is independent of $k$ and $n$.

\subsection{Outline of the Proofs}

Recall that the Laplace transform of the option price is given by
\begin{align}
	\widetilde{u}_n(q, x) = \sum_{z \in \mathbb{S}_n, z < L} h_n(q, x; z) \widetilde{w}_n(q, z), \label{eq:un-opt-recall}
\end{align}
where $\widetilde{w}_n(q, z)$ satisfies the linear system
\begin{align}
	(q{\pmb I} - {\pmb G}_n)\widetilde{w}_n(q, z) = 0,
\end{align}
and $h_n(q, x; y)$ is given by
\begin{align}\label{eq:hqxy-diff-recall}
	h_n(q, x; y) &= 1_{\{ x < L \}} e^{-qD} v_n(D, x; y) + 1_{\{ x < L \}} u_n^+(q, D, x; L^+) h_n(q, L^+; y) \\
	&\quad + 1_{\{ x \ge L \}} u^-_n(q, x; L^-) h_n(q, L^-; y),
\end{align}
where
\begin{align}
	&h_n(q, L^-; y) = \frac{e^{-qD} v_n(D, L^-; y)}{1 - u^-_n(q, L^+; L^-) u^+_n(q, D, L^-; L^+)}, \label{eq:hLLm-11-recall} \\
	&h_n(q, L^+; y) = \frac{e^{-qD} u^-_n(q, L^+; L^-) v_n(D, L^-; y)}{1 - u^-_n(q, L^+; L^-) u^+_n(q, D, L^-; L^+)}.  \label{eq:hLLm-22-recall} 
\end{align}

In the following, we will analyze the convergence of these quantities to their limits, which are the corresponding quantities for the diffusion model as ensured by Theorem \ref{thm:conv}, which proves convergecence of CTMC approximation. 
\begin{enumerate}
	\item Analysis of $v_n(D, x; y)$ (see \eqref{eq:vn-x-conv} and \eqref{eq:vn-L-conv}): The quantity is essentially the up-and-out barrier option price with $L^+$ as the effective upper barrier and payoff function $1_{\{ x = y \}}$. By \cite{zhang2018analysis}, $v_n(D, x; y)$ is represented by a discrete eigenfunction expansion based on the eigenvalues and eigenvectors from a matrix eigenvalue problem. Its continuous counterpart also admits an eigenfunction expansion representation where the eigenvalues and eigenfunctions are solutions to a Sturm-Liouville eigenvalue problem on the interval $(l, L^+)$. To analyze the approximation error, we can first analyze the errors for the eigenvalues and eigenfunctions and then exploit the eigenfunction expansions. However, there is one catch. In general, $L^+$ is not equal to $L$ unless $L$ is put on the grid. Therefore, the eigenpairs are dependent on the grid and we must carefully analyze the sensitivities of eigenpairs with respect to the boundary (see the second part of Lemma \ref{lmm:priors-eigenfunctions}). Hence we have additional error terms like $\mathcal{O}(L^+ - L)$ in \eqref{eq:vn-x-conv} and \eqref{eq:vn-L-conv} caused by $L$ not on the grid. 
	
	\item Analysis of $u_n^+(q, D, x; L^+)$ (see \eqref{eq:unp-x-conv} and \eqref{eq:unp-L-conv}): It can be split into two parts $u_{1, n}^+(q, x; L^+)$ and $u_{2, n}^+(q, D, x; L^+)$ as in \eqref{eq:un-plus-split}. $u_{1, n}^+(q, x; L^+)$ satisfies a linear system which is essentially a finite difference approximation to a two-point boundary value problem with effective boundaries $l$ and $L^+$. The error caused by applying finite difference approximation can be analyzed by standard approaches. However, we need to derive the dependence of the solution to the two-point boundary value problem on the right boundary as $L^+$ is not necessarily on the grid. $u_{2, n}^+(q, D, x; L^+)$ can be analyzed similarly as $v_n(D, x; y)$ once we obtain the error of $u_{1, n}^+(q, x; L^+)$. 
	
	\item Analysis of $u_n^-(q, x; L^-)$ (see \eqref{eq:conv-umn-x} and \eqref{eq:conv-umn-L}): $u_n^-(q, x; L^-)$ satisfies a linear system which results from finite difference approximation to a two-point boundary value problem with boundaries $L^-$ and $r$. Like $u_{1, n}^+(q, x; L^+)$, the error caused by finite difference approximation can be analyzed by standard approaches. In this case, $L^-$ is dependent on the grid and hence we need to derive the sensitivity of the solution to the two-point boundary value problem w.r.t. the left boundary. In particular, we show that the first and second order derivatives w.r.t. the left boundary form a nonlinear differential equation \eqref{eq:wm-b-boundary}.
	
	\item Analysis of $h_n(q, x; y)$ (see \eqref{eq:hn-xy-error}): By setting $x= L^-$ in $v_n(D, x; y)$, $u_n^+(q, D, x; L^+)$ and $x = L^+$ in $u_n^-(q, x; L^-)$ and expanding them around $x=L$, we have the error estimates for $v_n(D, L^-; y)$, $u_n^+(q, D, L^-; L^+)$ and $u_n^-(q, L^+; L^-)$ as shown in \eqref{eq:vn-L-conv}, \eqref{eq:unp-L-conv} and \eqref{eq:conv-umn-L}. Substituting these estimates to \eqref{eq:hLLm-11-recall} and \eqref{eq:hLLm-22-recall} yields the error estimates of $h_n(q, L^-; y)$ and $h_n(q, L^+; y)$. Afterwards, the error of $h_n(q, x; y)$ can be obtained by applying all the error estimates obtained so far to \eqref{eq:hqxy-diff-recall}. We would like to emphasize that the boundary dependent estimates of eigenvalues, eigenfunctions and solutions to the two-point boundary value problems are the key to cancel many lower order error terms. The final error is given by $\mathcal{O}(L^+-L) + \mathcal{O}(\delta_n^2)$, which indicates that second order convergence of $h_n(q, x,y)$ can be ensured if $L$ is put on the grid.
	
	\item Analysis of the option price $\widetilde{u}_n(q, x)$ (see \eqref{eq:un-error}): $\widetilde{w}_n(q, z)$ satisfies a linear system which is a finite different approximation to a two-point boundary value problem with a nonsmooth source term. We show that the error of $\widetilde{w}_n(q, z)$ is $\mathcal{O}(\delta_n^\gamma)$ (see \eqref{eq:wn-conv}), where $\gamma = 1$ in general and $\gamma = 2$ if all the discontinuities are placed exactly midway between two neighboring grid points. Using this estimate and the one for $h_n(q, x; y)$ in \eqref{eq:un-opt-recall} and then applying the error estimate of the trapezoid integration rule yield the error estimate of the option price.
	
\end{enumerate}

\subsection{Analysis of $v_n(D, x; y)$}

By \cite{zhang2018analysis}, $v_n(D, x; y)$ admits an eigenfunction expansion for $x \in \mathbb{S}_n \cap [l, L^+],\ y \in \mathbb{S}_n \cap (l, L]$:
\begin{align}
	v_n(D, x; y) = m_n(y)\delta y \sum_{k = 1}^{n_e} e^{-\lambda_{n, k}^+ D} \varphi_{n, k}^+(x) \varphi_{n, k}^+(y),
\end{align}
where $n_e$ is the number of points in $\mathbb{S}_n^o \cap (l, L)$,
\begin{align}
	&m_n(x) = \prod_{l<y \leq x, y \in \mathbb{S}_{n}} \frac{\mu\left(y^{-}\right) \delta^{-} y^{-}+\sigma^2\left(y^{-}\right)}{-\mu(y) \delta^{+} y+\sigma^2(y)},\ x = y_2, \cdots, y_{n - 1},\\
	&m_n(y_1) = m(y_1) \exp\left( \frac{\mu(y_1)}{\sigma^2(y_1)} (\delta^+ y_1 - \delta^- y_1) \right),
\end{align}
and $(\lambda^{+}_{n, k}, \varphi^+_{n, k}(x))$, $1 \le k \le n_e$ are the pairs of solutions to the eigenvalue problem,
\begin{align}
	\begin{cases}
		{\pmb G}_n \psi(x) = -\lambda\psi(x) ,\ x \in \mathbb{S}_n \cap (l, L),\\
		\psi(l) = \psi(L^+) = 0.
	\end{cases}
\end{align}
These eigenfunctions are normalized and orthogonal under the measure $m_n(x)\delta x$. That is, 
\begin{equation} 
	\sum_{x \in \mathbb{S}_n \cap (l, L)} \varphi^+_{n, k}(x) \varphi^+_{n, l}(x) m_n(x) \delta x = \delta_{k,l}.
\end{equation}
where $\delta_{kl}$ is the Kronecker delta. 

As shown in \cite{zhang2018analysis}, $\pmb{G}_n$ admits a self-adjoint representation
	\begin{align}
		\pmb{G}_n g(x) = \frac{1}{m_n(x)} \frac{\delta^- x}{\delta x} \nabla^-\left( \frac{1}{s_n(x)} \nabla^+ g(x) \right),\ x \in \mathbb{S}_n^o,
	\end{align}
with $1/s_n(x) = m_n(x) (b(x) \delta^-x + a(x))/2$ for $x \in \mathbb{S}_n^o$, $1/s_n(x) = m_n(x^+) (-b(x^+) \delta^+x^+ + a(x^+))/2$ for $x = l$, and
\begin{align}
	\nabla^- g(x) = \frac{g(x) - g(x^-)}{\delta^- x},\ \nabla^+ g(x) = \frac{g(x^+) - g(x)}{\delta^+ x}.
\end{align}

For $x \in \mathbb{S}_n \cap [l, L^+]$, $v_n(D, x; l)$ satisfies 
\begin{align}
	\begin{cases}
		\frac{\partial v_n}{\partial D}(D, x; l) = {\pmb G}_n v_n(D, x; l),\ D> 0, x \in \mathbb{S}_n \cap (l, L),\\
		v_n(D, l; l) = 1,\ D > 0,\\
		v_n(D, L^+; l) = 0,\ D > 0,\\
		v_n(0, x; l) = 1_{\{ x = l \}}.
	\end{cases}
\end{align}
It can be written as the difference of two parts: $v_n(D, x; l) = v_{1, n}(x; l) - v_{2, n}(D, x; l)$ with $v_{1, n}(x; l)$ satisfying
\begin{align}
	\begin{cases}
		{\pmb G}_n v_{1, n}(x; l) = 0,\ x \in \mathbb{S}_n \cap (l, L),\\
		v_{1, n}(l; l) = 1,\ v_{2, n}(L^+; l) = 0, 
	\end{cases}
\end{align}
and $v_{2, n}(D, x; l)$ satisfying
\begin{align}
	\begin{cases}
		\frac{\partial v_{2,n}}{\partial D}(D, x; l) = {\pmb G}_n v_{2,n}(D, x; l),\ D> 0, x \in \mathbb{S}_n \cap (l, L),\\
		v_{2,n}(D, l; l) = 0,\ D > 0,\\
		v_{2,n}(D, L^+; l) = 0,\ D > 0,\\
		v_{2,n}(0, x; l) = v_{1, n}(x; l) - 1_{\{ x = l \}}.
	\end{cases}
\end{align}
$v_{2, n}(D, x; l)$ admits the eigenfunction expansion,
\begin{align}
	v_{2, n}(D, x; l) = \sum_{k = 1}^{n_e} d_{n, k}e^{-\lambda_{n, k}^+ D} \varphi_{n, k}^+(x),\ d_{n, k} = \sum_{y \in \mathbb{S}_n \cap (l, L)} v_{1, n}(y; l) m_n(y) \delta y.
\end{align}

To study the convergence of $v_n(D, x; y)$, we consider the boundary-dependent eigenvalue problem
\begin{align}
	\begin{cases}
		\mathcal{G} \psi(x)  = -\lambda \psi(x),\ x \in (l, b),\\
		\psi(l) = \psi(b) = 0,  
	\end{cases}
\end{align}
where $b>l$, and $\mathcal{G}$ is a second-order differential operator given by
\begin{align}
	\mathcal{G} f(x)=\frac{1}{2} \sigma^{2}(x) f^{\prime \prime}(x)+\mu(x) f^{\prime}(x).
\end{align}
The sequence of solutions are $(\lambda_k^+(b), \varphi_k^+(x, b))$, $k = 1, 2, \cdots$. Consider the speed density of the diffusion given by
\begin{equation}\label{eq:speed-density}
	m(x) = \frac{2}{\sigma^2(x)} e^{\int_{l}^{x} \frac{2\mu(y)}{\sigma^2(y)} dy}.
\end{equation}
The eigenfunctions are normalized and orthogonal under the speed measure $m(x)dx$. That is,
\begin{equation}
	\int_{l}^{b}\varphi^+_k(x, b) \varphi^+_l(x, b) m(x) dx = \delta_{k,l}.
\end{equation}
We also consider the solution $v_1(x, b)$ to the equation
\begin{align}
	\begin{cases}
	 	\mathcal{G} v_1(x, b) = 0,\ x \in [l, b],\\
	 	v_1(l, b) = 1,\ v_1(b, b) = 0,
	\end{cases}
\end{align}
and
\begin{align}
	v_2(D, x, b) = \sum_{k = 1}^{\infty} d_k(b) e^{-\lambda_k^+ D} \varphi_k^+(x),\ d_k(b) = \int_{l}^{b} v_1(y, b) m(y) dy.
\end{align}

\begin{lemma}\label{lmm:priors-eigenfunctions}
	Under Assumption \ref{assump:grid} and Assumption \ref{assump:smoothness}, the following results hold.
	\begin{enumerate}
		\item For the speed density, there holds that for $x \in \mathbb{S}_n^o$,
		\begin{align}
			m_n(x)= m(x) + m(x) \frac{b(x)}{a(x)} \left( \delta^+ x - \delta^- x \right) + \mathcal{O}(\delta_n^2). \label{eq:mn-m}
		\end{align}
		There exist constants $c_1, c_2 > 0$ independent of $h_n$ and $x$ such that
		\begin{align}
			c_1 \le m_n(x) \le c_2,\ c_1 \le m(x) \le c_2. \label{eq:m-bound}
		\end{align}
		
		\item $\partial_b^i \lambda_k(b)$ with $0 \le i \le 1$, $\partial_x^i \partial_b^j \varphi^+_k(x, b)$ with $0 \le i \le 3$, $0 \le j \le 2$, are well-defined and continuous for $x \in [l, L]$ and $b$ in any closed interval $I \subset [L, r)$. We also have
		\begin{align}
			& \lambda_{n, k}^+ = \lambda_k(L) + k^{3}\mathcal{O}(L^+ - L) + \mathcal{O}(k^{5} \delta_n^2), \label{eq:eigenvalue-error}\\
			& \varphi_{n, k}^+(x) = \varphi^+_k(x, L) + k \mathcal{O}(L^+ - L) + \mathcal{O}(k^4 \delta_n^2),\\
			& \varphi^+_{n, k}(L^-) = -\partial_x \varphi^+_k(L, L) \delta^+ L^- + \frac{1}{2} \partial_{xx} \varphi^+_k(L, L) (\delta^+ L^-)^2 \\
			&\quad\quad \quad \quad \quad \quad  + \mathcal{O}(k^{4} \delta^+ L^-)(L^+ - L) + \mathcal{O}(k^{6} \delta_n^3). \label{eq:eigenfunction-error-boundary}
		\end{align}
		Moreover, there exist constants $c_3, c_4, c_5 > 0$ independent of $k$, $h_n$ and $x$ such that
		\begin{align}
			&\left| \varphi_{n, k}^+(x) \right| \le c_3 k,\ \left| \varphi^+_k(x, L^+) \right| \le c_3,\\
			& c_4 k^2 \le \lambda_{n, k}^+ \le c_5 k^2,\ c_4 k^2 \le  \lambda_k(L^+) \le c_5 k^2.
		\end{align}
	
		\item For $v_{1, n}(x; l)$ and $v_{2, n}(D, x; l)$, we have
		\begin{align}
			& v_{1, n}(x; l) = v_1(x, L) + \mathcal{O}(L^+ - L) + \mathcal{O}(\delta_n^2),\\
			&v_{1, n}(L^-; l)= - \partial_x v_1(L, L) \delta^+ L^- + \frac{1}{2} \partial_{xx} v_1(L, L) (\delta^+ L^-)^2  + \mathcal{O}(\delta^+L^-)(L^+ - L) + \mathcal{O}(\delta_n^3),\\
			& v_{2, n}(D, x; l) = v_2(D, x, L) + \mathcal{O}(L^+  - L) + \mathcal{O}(\delta_n^2),\\
			&v_{2, n}(D, L^-; l)= - \partial_x v_2(D, L, L) \delta^+ L^- + \frac{1}{2} \partial_{xx} v_2(D, L, L) (\delta^+ L^-)^2  + \mathcal{O}(\delta^+L^-)(L^+ - L) + \mathcal{O}(\delta_n^3).
		\end{align}
		
	\end{enumerate}	
\end{lemma}

\begin{proposition}\label{prop:conv-vn}
	Under Assumption \ref{assump:grid} and Assumption \ref{assump:smoothness}, we obtain for $x,y \in \mathbb{S}_n \cap (l, L)$,
	\begin{align}
		& v_n(D, x; y) / (m_n(y) \delta y) = \bar{v}(D, x; y) + \mathcal{O}(L^+ - L) + \mathcal{O}(\delta_n^2),\label{eq:vn-x-conv}\\
		& v_n(D, L^-; y) / (m_n(y) \delta y) = - \partial_x \bar{v}(D, L; y) \delta^+ L^- + \frac{1}{2} \partial_{xx} \bar{v}(D, L; y) (\delta^+ L^-)^2 \\
		&\quad\quad \quad \quad \quad \quad \quad\quad \quad\quad \quad + \mathcal{O}(\delta^+L^-)(L^+ - L) + \mathcal{O}(\delta_n^3), \label{eq:vn-L-conv}
	\end{align}
	and
	\begin{align}
		&v_n(D, x; l) = v(D, x; l) + \mathcal{O}(L^+ - L) + \mathcal{O}(\delta_n^2),\\
		& v_n(D, L^-; l)= - \partial_x v(D, L; l) \delta^+ L^- + \frac{1}{2} \partial_{xx} v(D, L; l) (\delta^+ L^-)^2 + \mathcal{O}(\delta^+L^-)(L^+ - L) + \mathcal{O}(\delta_n^3),
	\end{align}
	where $v(D, x; y) = m(y) \sum_{k = 1}^{\infty} e^{-\lambda_k^+(L) D} \varphi_k^+(x; L) \varphi_k^+(y; L)$ and $\bar{v}(D, x; y) = v(D, x; y) / m(y)$ for $y \in (l, L]$ and $v(D, x; l) = v_1(x, L) - v_2(D, x, L)$. 
	
	Furthermore, $\bar{v}(D, x; y)$ and $v(D, x; l)$ satisfy the following equations at $x=L$:
	\begin{align}
		&\mu(L) \partial_x \bar{v}(D, L; y) + \frac{1}{2} \sigma^2(L) \partial_{xx} \bar{v}(D, L; y) = 0, \ y \in (l, L), \label{eq:v-bar-boundary}\\
		&\mu(L) \partial_x v(D, L; l) + \frac{1}{2} \sigma^2(L) \partial_{xx} v(D, L; l) = 0. 
	\end{align}
\end{proposition}

\subsection{Analysis of $u_n^+(q, D, x; L^+)$}

Recall that 
\begin{align}
	u_n^+(q, D, x; L^+) = u_{1, n}^+(q, x; L^+) - u_{2, n}^+(q, D, x; L^+) \label{eq:un-plus-split}
\end{align}
with $u_{1, n}^+(q, x; L^+)$ satisfying
\begin{align}
	\begin{cases}
		{\pmb G}_n u_{1, n}^+(q, x; L^+) - qu_{1, n}^+(q, x; L^+) = 0,\ x \in \mathbb{S}_n \cap (l, L),\\
		u_{1, n}^+(q, l; L^+) = 0,\ u_{1, n}^+(q, L^+; L^+) = 1.
	\end{cases}
\end{align}
In addition, $u_{2, n}^+(q, D, x; L^+)$ admits an  eigenfunction expansion given by
\begin{align}
	&u_{2, n}^+(q, D, x; L^+) = e^{-qD}\sum_{k = 1}^{n_e} c_{n, k}(q) e^{-\lambda_{n, k}^+D} \varphi_{n, k}^+(x),\\
	& c_{n, k}(q) = \sum_{y \in \mathbb{S}_n^o \cap (-\infty, L)} \varphi_{n, k}^+(y) u_{1, n}^+(q, y; L^+) m_n(y) \delta y.
\end{align}

To study the convergence of $u_n^+(q, D, x; L^+)$, we consider the solution $u^+(q, D, x, b)$ to the following boundary-dependent PDE,
\begin{align}
	\begin{cases}
		\frac{\partial u}{\partial D}(D, x) = \mathcal{G} u(D, x) - q u(D, x),\ D > 0, x \in (l, b),\\
		u(D, l) = 0,\ u(D, b) = 1,\ D > 0,\\
		u(0, x) = 0,\ x \in [l, b].
	\end{cases}
\end{align}
$u^+(q, D, x, b)$ can be written as $u^+(q, D, x, b) = u^+_1(q, x, b) - u^+_2(q, D, x, b)$, where $u^+_1(q, x, b)$ is the solution to
\begin{align}
	\begin{cases}
		\mathcal{G} u(x) - q u(x) = 0,\ x \in [l, b],\\
		u(l) = 0,\ u(b) = 1,
	\end{cases}
\end{align}
and $u^+_2(q, D, x, b)$ is the solution to
\begin{align}
	\begin{cases}
		\frac{\partial u}{\partial D}(D, x) = \mathcal{G} u(D, x) - q u(D, x),\ D > 0, x \in (l, b),\\
		u(D, l) = u(D, b) = 0,\ D > 0,\\
		u(0, x) = u^+_1(q, x, b),\ x \in [l, b].
	\end{cases}
\end{align}
We can represent $u^+_2(q, D, x, b)$ using an eigenfunction expansion as follows:
\begin{align}
	u^+_2(q, D, x, b) = e^{-qD}\sum_{k = 1}^{\infty} c_k(q, b) e^{-\lambda_k^+(b) D} \varphi_k^+(x, b),\ c_k(q, b) = \int_{l}^{b} \varphi_k^+(y, b) u^+_1(q, y, b) m(y) dy.
\end{align}

\begin{lemma}\label{lmm:prior-u1p}
	Suppose Assumption \ref{assump:grid} and Assumption \ref{assump:smoothness} hold. Then, for any $q > 0$, $u_1^+(q, x, b)$ is $C^3$ in $b$ for $b$ in any closed interval enclosed by $[L, r]$ and we have
	\begin{align}
		& u_{1, n}^+(q, x; L^+) =  u_1^+(q, x, L) + \mathcal{O}(L^+ - L) + \mathcal{O}(\delta_n^2),\\
		& u_{1, n}^+(q, L^-; L^+) = 1 - \partial_x u_1^+(q, L, L)\delta^+ L^- + \frac{1}{2} \partial_{xx} u_1^+(q, L, L)(\delta^+ L^-)^2 + \mathcal{O}(\delta_n^3).
	\end{align}
\end{lemma}

\begin{proposition}\label{prop:conv-unp}
	Under Assumption \ref{assump:grid} and Assumption \ref{assump:smoothness}, we have that for any $q > 0$,
	\begin{align}
		&u_n^+(q, D, x, L^+) = u^+(q, D, x, L) + \mathcal{O}(L^+ - L) + \mathcal{O}(\delta_n^2), \label{eq:unp-x-conv} \\
		&u_n^+(q, D, L^-, L^+) = 1 - \partial_x u^+(q, D, L, L) \delta^+ L^- + \frac{1}{2} \partial_{xx} u^+(q, D, L, L) (\delta^+ L^-)^2 \\
		&\quad\quad \quad \quad \quad \quad \quad \quad \quad + \mathcal{O}(\delta^+L^-)(L^+ - L) + \mathcal{O}(\delta_n^3). \label{eq:unp-L-conv}
	\end{align}
	Moreover, the following equation holds:
	\begin{align}
		&\mu(L) \partial_x u^+(q, D, L, L) + \frac{1}{2} \sigma^2(L) \partial_{xx} u^+(q, D, L, L) - q = 0. \label{eq:up-boundary}
	\end{align}
\end{proposition}

\subsection{Analysis of $u_n^-(q, x; L^-)$}

$u_n^-(q, x; L^-)$ satisfies the equation
\begin{align}
	\begin{cases}
		{\pmb G}_n u_n^-(q, x; L^-) - qu_n^+(q, x; L^-) = 0,\ x \in \mathbb{S}_n^o \cap (L, \infty),\\
		u_n^-(q, L^-; L^-) = 1,\ u_n^+(q, r; L^-) = 0.
	\end{cases}
\end{align}
We consider the following boundary-dependent differential equation with solution $u^-(q, x; b)$:
\begin{align}
	\begin{cases}
		\mathcal{G} u^-(q, x, b) - qu^-(q, x, b) = 0,\ x \in (b, r),\\
		u^-(q, b, b) = 1,\ u^-(q, r, b) = 0.
	\end{cases}
\end{align}

\begin{lemma}\label{prop:conv-unn}
	Suppose Assumption \ref{assump:grid} and Assumption \ref{assump:smoothness} hold. Then, for any $q > 0$, $u^-(q, x, b)$ is $C^3$ in $b$ for $b$ in any closed interval enclosed by $[l, L]$ and we have
	\begin{align}
		&\mu(L) \partial_b u^-(q, L, L) - \sigma^2(L)  (\partial_b u^-(q, L, L))^2 + \frac{1}{2} \sigma^2(L)  \partial_{bb} u^-(q, L, L) + q = 0. \label{eq:wm-b-boundary}
	\end{align}
	Furthermore, $u^-_n(q, x; L^-) = \widetilde{u}^-_n(q, x; L^+) u^-_n(q, L^+; L^-)$, where $\widetilde{u}_n^-(q, x; L^+) = \mathbb{E}_x\left[ e^{q\tau_{L^{++}}^-} 1_{\{ Y_{\tau_{L^{++}}^-} = L^+ \}} \right]$ and $L^{++}$ is the grid point on the right of $L^+$.
	\begin{align}
		& \widetilde{u}_n^-(q, x; L^+) =  u^-(q, x, L) + \mathcal{O}(L^+ - L^-) + \mathcal{O}(\delta_n^2), \label{eq:conv-umn-x} \\
		& u_n^-(q, L^+; L^-) = 1 - \partial_b u^-(q, L, L) \delta^+L^- + \frac{1}{2} \partial_{bb} u^-(q, L, L) (\delta^+L^-)^2 \\
		&\quad\quad \quad \quad \quad \quad \quad  + \mathcal{O}(\delta^+ L^-)(L^+ - L) + \mathcal{O}(\delta_n^3). \label{eq:conv-umn-L}
	\end{align}	
\end{lemma}

\subsection{Convergence Rates of $h_n(q, x; y)$ and Option Prices}

\begin{theorem}\label{thm:conv-h}
	Suppose Assumption \ref{assump:grid} and Assumption \ref{assump:smoothness} hold. Let
	\begin{align}
		h(q, x ; y)=& 1_{\{x<L\}} e^{-q D} v(D, x ; y)+1_{\{x<L\}} u^{+}\left(q, D, x ; L\right) h\left(q, L ; y\right) \\ &+1_{\{x \geq L\}} u^{-}\left(q, x ; L\right) h\left(q, L ; y\right),
	\end{align}
	with $h(q, L; y) = -e^{-qD} \partial_x v(D, x; y) / \left(\partial_x u^+(q, D, L, L) + \partial_b u^-(q, L, L)\right)$. Then for any $q > 0$, $h(q, x; y)$ is twice continuously differentiable in $y$, and
	\begin{align}
		&\lim_{y \downarrow l} h(q, x; y) =  h(q, x; L)  = 0 \text{ for all } x \in [l, r],\\
		&\sup_{x, y \in (l, L)}\left| \partial_y^k h(q, x; y) \right| < \infty,\ k = 0, 1, 2.		
	\end{align}
	Moreover, we have
	\begin{align}
		h_n(q, x; y) / (m_n(y) \delta y) = h(q, x; y) / m(y) + \mathcal{O}(L^+ - L) + \mathcal{O}(h_n^2). \label{eq:hn-xy-error}
	\end{align}
\end{theorem}

To obtain the convergence rate of option prices, we first analyze $\widetilde{w}_n(q, x)$. Recall that $\mathcal{D}$ is the set of discontinuities of the payoff function $f$. 

\begin{proposition}\label{prop:conv-wn}
	Suppose Assumption \ref{assump:grid} and Assumption \ref{assump:smoothness} hold. Then for any $q > 0$, the following equation admit a unique solution $\widetilde{w}(q, \cdot) \in C^1([l, r]) \cap C^2([l, r] \backslash \mathcal{D})$:,
	\begin{align}\label{eq:w-tilde-equation}
		\begin{cases}
			\mu(x) w'(x) + \frac{1}{2} \sigma^2(x) w''(x) - qw(x) = f(x),\ x \in [l, r] \backslash \mathcal{D},\\
			w(x+) = w(x-),\ w'(x+) = w'(x-),\ x \in \mathcal{D},\\  
			w(l) = f(l) / q,\ w(r) = f(r) / q.
		\end{cases}
	\end{align}
	In particular, $\widetilde{w}(q, \cdot) \in C^2([l, r] )$ if $f(\cdot)$ is continuous on $[l, r]$. Moreover,
	\begin{align}
		\widetilde{w}_n(q, x) - \widetilde{w}(q, x) = \mathcal{O}(h_n^\gamma), \label{eq:wn-conv}
	\end{align}
	where $\gamma = 1$ in general and $\gamma = 2$ if either $f(\cdot)$ is continuous on $(l, r)$ or each point in $\mathcal{D}$ is placed midway between two adjacent grid points.
\end{proposition}

Let 
\begin{align}
	\widetilde{u}(q, x) = \int_{l}^L h(q, x; z) \widetilde{w}(q, z)dz + h(q, x; l) \widetilde{w}(q, l). \label{eq:u-opt-recall}
\end{align}
The error of $\tilde{u}_n(q, x)$ for approximating $\widetilde{u}(q, x)$ can be decomposed into four parts as follows:
\begin{align}
	&\widetilde{u}_n(q, x) - \widetilde{u}(q, x) \\
	&= \sum_{z \in \mathbb{S}_n \cap (l, L)} h_n(q, x; z) \widetilde{w}_n(q, z) + h_n(q, x; l) \widetilde{w}_n(q, l) - \int_{l}^{L} h(q, x; z) \widetilde{w}(q, z) dz - h(q, x; l) \widetilde{w}(q, l)\\
	&= \sum_{z \in \mathbb{S}_n \cap (l, L)} \left(h_n(q, x; z) /(m_n(z) \delta z) - h(q, x; z)/m(z)\right) \widetilde{w}_n(q, z) m_n(z) \delta z \\
	&\quad + h_n(q, x; l) \widetilde{w}_n(q, l) - h(q, x; l) \widetilde{w}(q, l) \\
	&\quad + \sum_{z \in \mathbb{S}_n \cap (l, L)} h(q, x; z) \left( \widetilde{w}_n(q, z) - \widetilde{w}(q, z)  \right) \frac{m_n(z)}{m(z)} \delta z \\
	&\quad + \frac{1}{m(z)}\left( \sum_{z \in \mathbb{S}_n \cap (l, L)} h(q, x; z) \widetilde{w}(q, z) m_n(z) \delta z - \int_{l}^{L} h(q, x; z) \widetilde{w}(q, z) m(z) dz\right).
\end{align}
The first to the third parts of errors can be analyzed using the estimates for the approximation errors for $h(q, x; z)$ and $\widetilde{w}(q, z)$. The last part is the error of a numerical integration scheme which can be analyzed with the estimate of $m_n$ in Lemma \ref{lmm:priors-eigenfunctions}. Utilizing the previous estimates, we obtain the following theorem.

\begin{theorem}\label{thm:conv-u}
	Suppose Assumption \ref{assump:grid} and Assumption \ref{assump:smoothness} hold. Then we have the following estimate for any $q > 0$:
	\begin{align}
		\widetilde{u}_n(q, x) = \widetilde{u}(q, x) +  \mathcal{O}(L^+ - L) + \mathcal{O}(\delta_n^\gamma),  \label{eq:un-error}
	\end{align}
	where $\gamma = 1$ in general and $\gamma = 2$ if each point in $\mathcal{D}$ is placed midway between two adjacent grid points.
\end{theorem}

\begin{remark}
	Setting $f(x) = 1$, we obtain the error estiamte for $h(q,x)$ as $h_n(q, x) = h(q, x) + \mathcal{O}(L^+ - L) +  \mathcal{O}(\delta_n^2)$.
\end{remark}

\subsection{Grid Design}

From Theorem \ref{thm:conv-u}, we need to place $L$ on the grid (in this case $L^+ = L$) and strikes in the midway to ensure second order convergence. This can be achieved by the piecewise uniform (PU) grid construction proposed in \cite{zhang2019analysis}. For example, suppose the payoff has a single strike at $K < L$, then the PU grid can be constructed as follows.
\begin{align}
	\mathbb{S}_{PU} &= \{l + (i + 1/2)h_1: 0 \le i \le n_1 - 1 \} \cup \{ K + h_1/2  + ih_2: 0 \le i \le n_2  \} \\
	&\quad \cup \{ L + ih_3: 0 \le i \le n_3 \}, \label{eq:PU-grid}
\end{align}
where $h_1 = (K - l) / n_1$, $h_2 = (L - K - h_1/2) / n_2$ and $h_3 = (r - L) / n_3$. The case of $L < K$ is similar. 

Although this grid design is inspired by the convergence rate results for diffusion models, we can also apply it to pure-jump and jump-diffusion models to remove convergence oscillations and make extrapolation applicable to accelerate convergence, as shown by numerical examples. 

\begin{remark}
If $K = L$, we do not have a grid that can meet both requirements that $L$ should be on the grid and $K$ should be placed midway. However, we can still achieve second order convergence in the following way. We first construct a grid that places $K$ midway to calculate $\widetilde{w}(q, x)$ on this grid. Then we use a second grid with $L$ on it to calculate $\widetilde{u}_n(q, x)$. Since these two grids do not coincide, we need the values of $\widetilde{w}(q, x)$ at points off the first grid and they can be obtained by interpolating the values of $\widetilde{w}(q, x)$ on the first grid using a high-order interpolation scheme. 
\end{remark}

\section{Numerical Examples}\label{sec:numerical-ex}

To illustrate the performance of our algorithm, we compute the distribution of Parisian stopping times and price Parisian options under several popular models. Throughout this section, we use the following parameter setting: the risk-free rate $r_f = 0.05$, the dividend yield $d = 0$, the option's maturity $T = 1$, the reference level for the length of excursion $D=1/12$, the initial stock price $S_0 = 90$, the barrier level $L = 90$, and the strike price $K = 95$, except for the standard Brownian motion model. We consider the following representative models:
\begin{itemize}
	\item The standard Brownian (BM).
	\item The Black-Scholes (BS) model: $\sigma= 0.3$.
	\item Regime-switching BS model with $2$ regimes: volatility is $0.3$ in the first regime and $0.5$ in the second. The rate of transition is 0.75 from regime 1 to 2 and 0.25 from 2 to 1.
	\item Kou's model (\cite{kou2004option}): $\sigma = 0.30, \lambda = 3.0, \eta^+ = \eta^- = 0.1, p^+ = p^- = 0.5$. This is a popular  jump-diffusion model with finite jump activity. 
	\item Variance Gamma (VG) model (\cite{madan1998variance}): $\sigma = 0.1213$, $\nu = 0.1686$, $\theta = -0.1436$. This is a commonly used pure jump model with infinite jump activity. 
\end{itemize}

Figure \ref{fig:bm-cdf-12} shows the convergence of $\mathbb{P}[\tau_{L, D}^- \le t]$ under the standard BM model by plotting errors against the number of CTMC states. Apparently we can observe second order convergence in both plots (see the blue lines) for $t = 1.5$ and $t = 3.0$, which verifies the theoretical convergence rate. Furthermore, the plots demonstrate that Richardson's extrapolation can substantially reduce the error as seen from the red lines.

Figure \ref{fig:opt-pricing} presents the convergence behavior of CTMC approximation for a sing-barrier down-and-in Parisian option under the BS model,  Kou's model and the VG model, respectively. In the left panel, we compare two different grids: a picewise uniform grid with $L$ on the grid and $K$ placed midway between two adjacent grid points (see \eqref{eq:PU-grid}) and a uniform grid. For all three models, when the uniform grid is used, convergence is only first order. Moreover, there are big oscillations and this means Richardson's extrapolation cannot be applied. In contrast, using the grid we design, convergence becomes second order and smooth in all the models. The results for the BS model justify our theoretical results. Although we don't have proof for models with jumps, these numerical examples suggest that our grid design would also work for jump models. The plots in the right panel showcase the effectiveness of Richardson's extrapolation. We also provide the 
numerical values for the original and extrapolated results in Table \ref{tab:bs-down-in}, \ref{tab:kou-down-in} and \ref{tab:cgmy-down-in}. In Section \ref{app:RS-SV}, we extend our algorithm to deal with regime-switching models. As an example, we consider the regime-switching BS model and the good performance of our method can be seen from the results in Table \ref{tab:rs-bs-down-in}.


We compare our method with the explicit finite difference scheme developed in \cite{wilmott2013paul}. The results are shown in Figure \ref{fig:comparison-fdm}, which shows that our algorithm takes much less time (about 25\% of the finite difference method) to achieve the same error level. The relatively slow performance of the finite difference approach should be expected as the PDE satisfied by the Parisian down-in option involves an additional dimension that records the length of excursion below the barrier and time discretization is required. 

\begin{figure}
	\centering
	\includegraphics[width = 0.45 \textwidth]{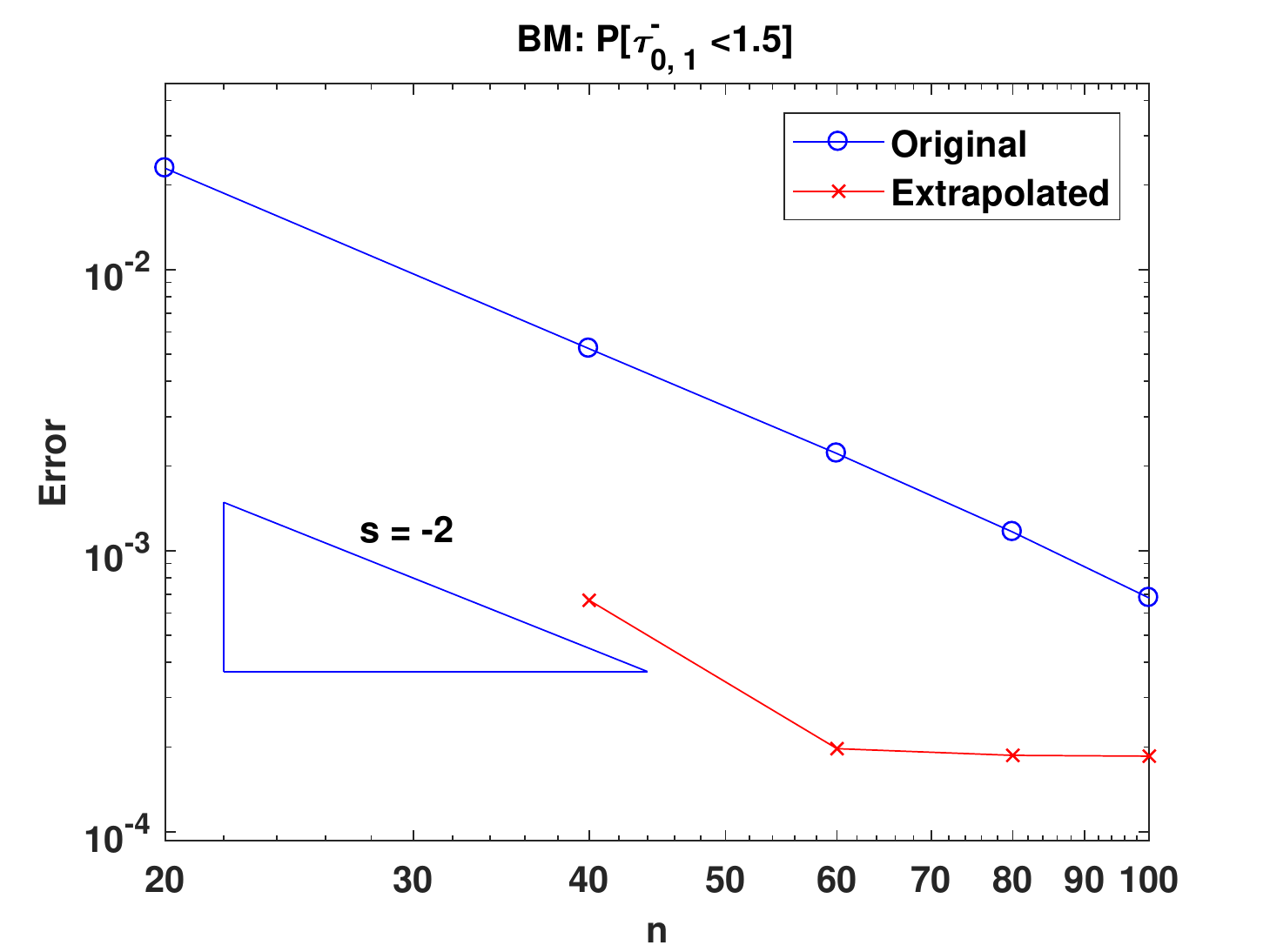}
	\includegraphics[width = 0.45 \textwidth]{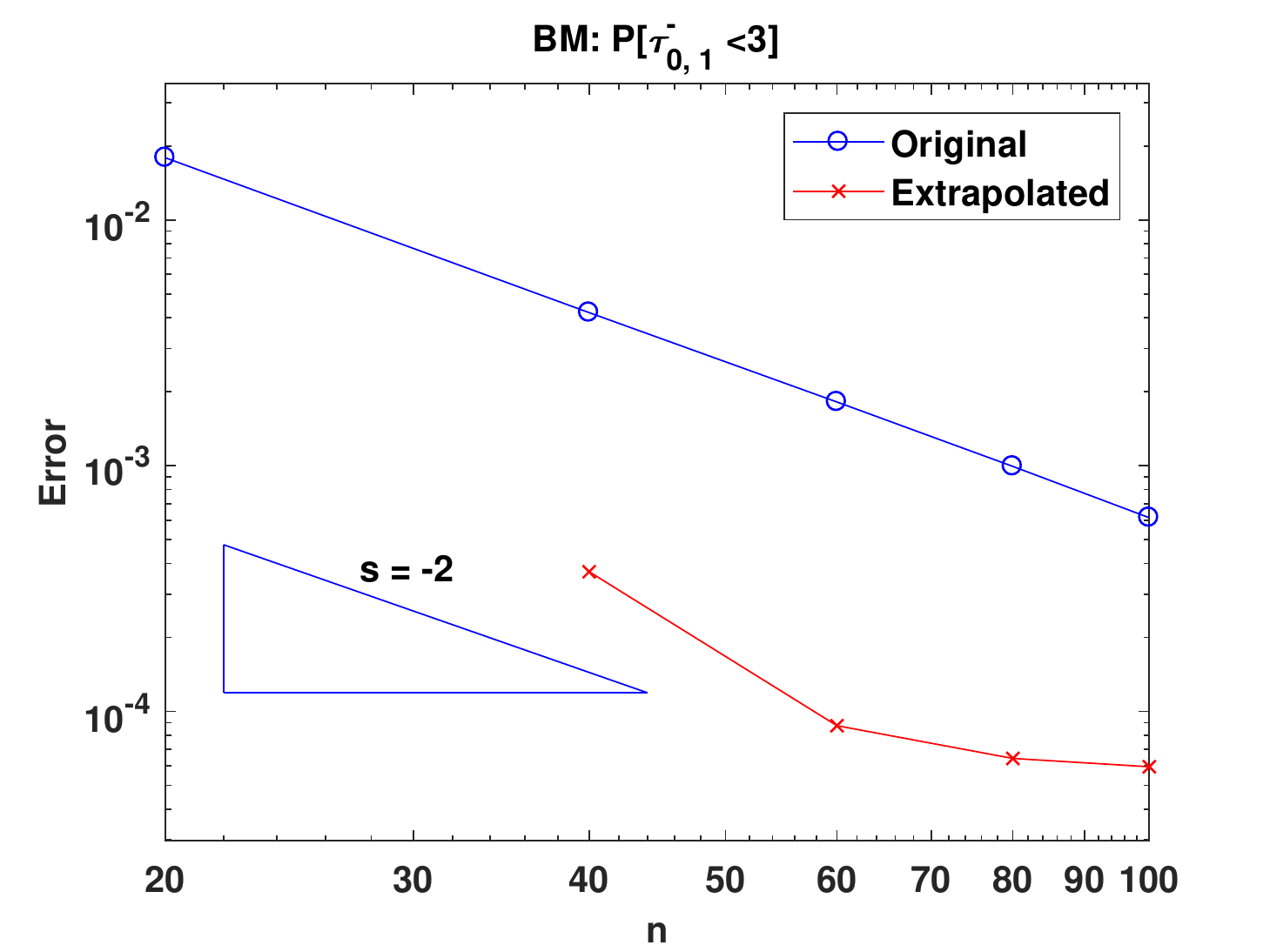}
	\caption{The convergence of $\mathbb{P}[\tau_{L, D}^- \le t]$ of a standard Brownian motion with $L = 0, D = 1$ and $t = 1.5, 3$. The benchmark for calculating the error is given by the solution in \cite{dassios2013parisian}.}\label{fig:bm-cdf-12}  
\end{figure}

\begin{figure}
	\centering
	\includegraphics[width = 0.45 \textwidth]{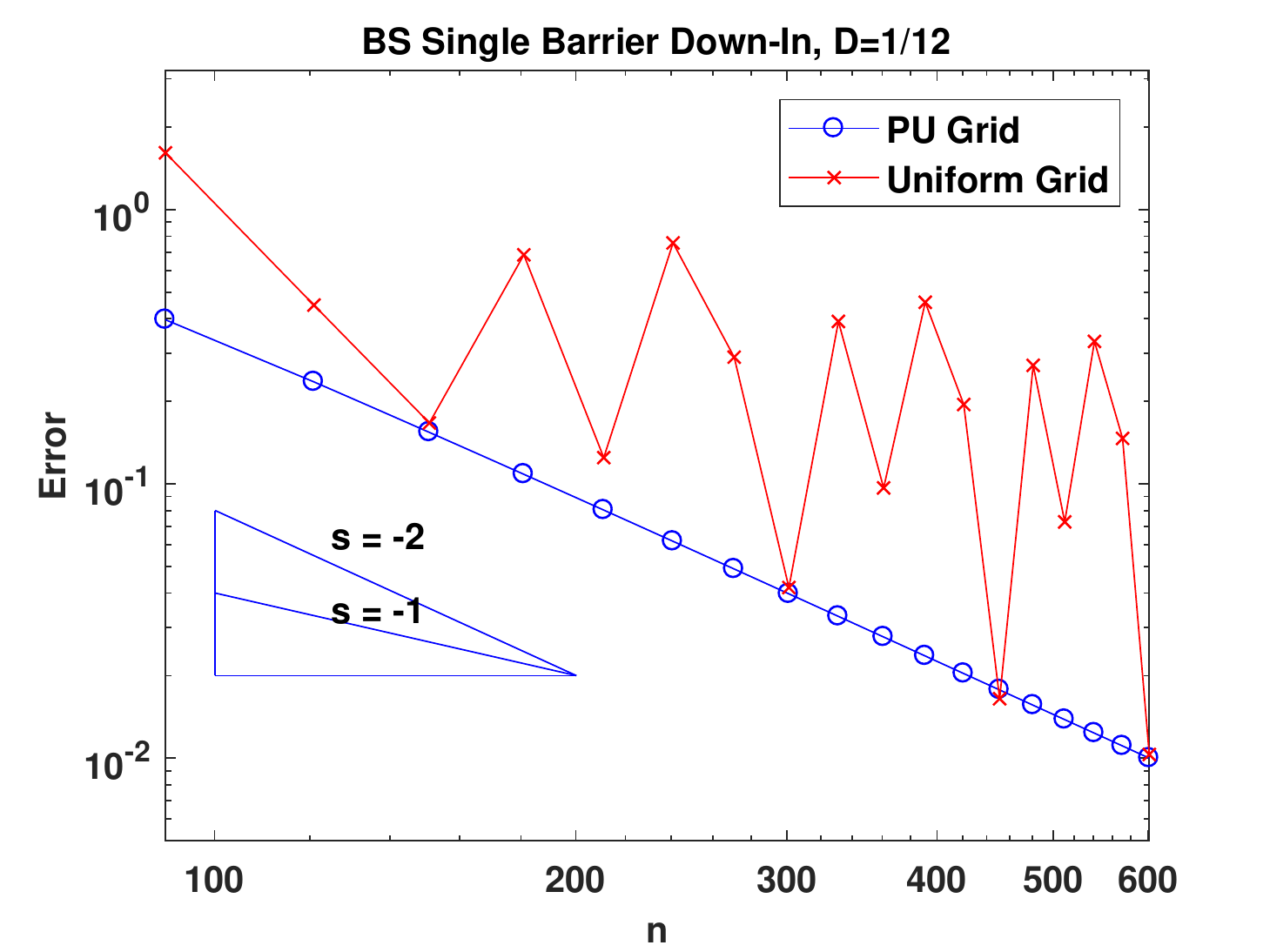}
	\includegraphics[width = 0.45 \textwidth]{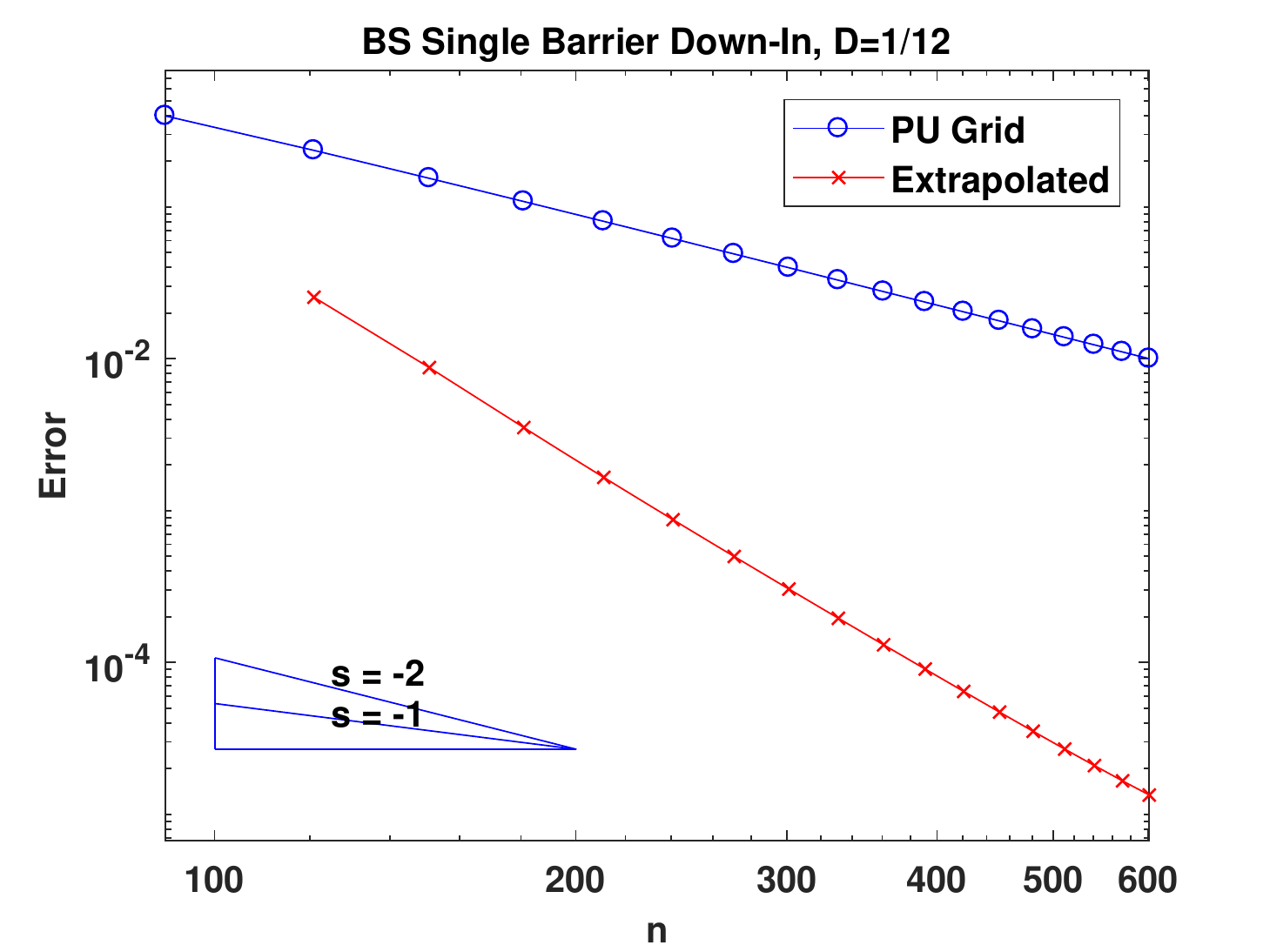}
	\includegraphics[width = 0.45 \textwidth]{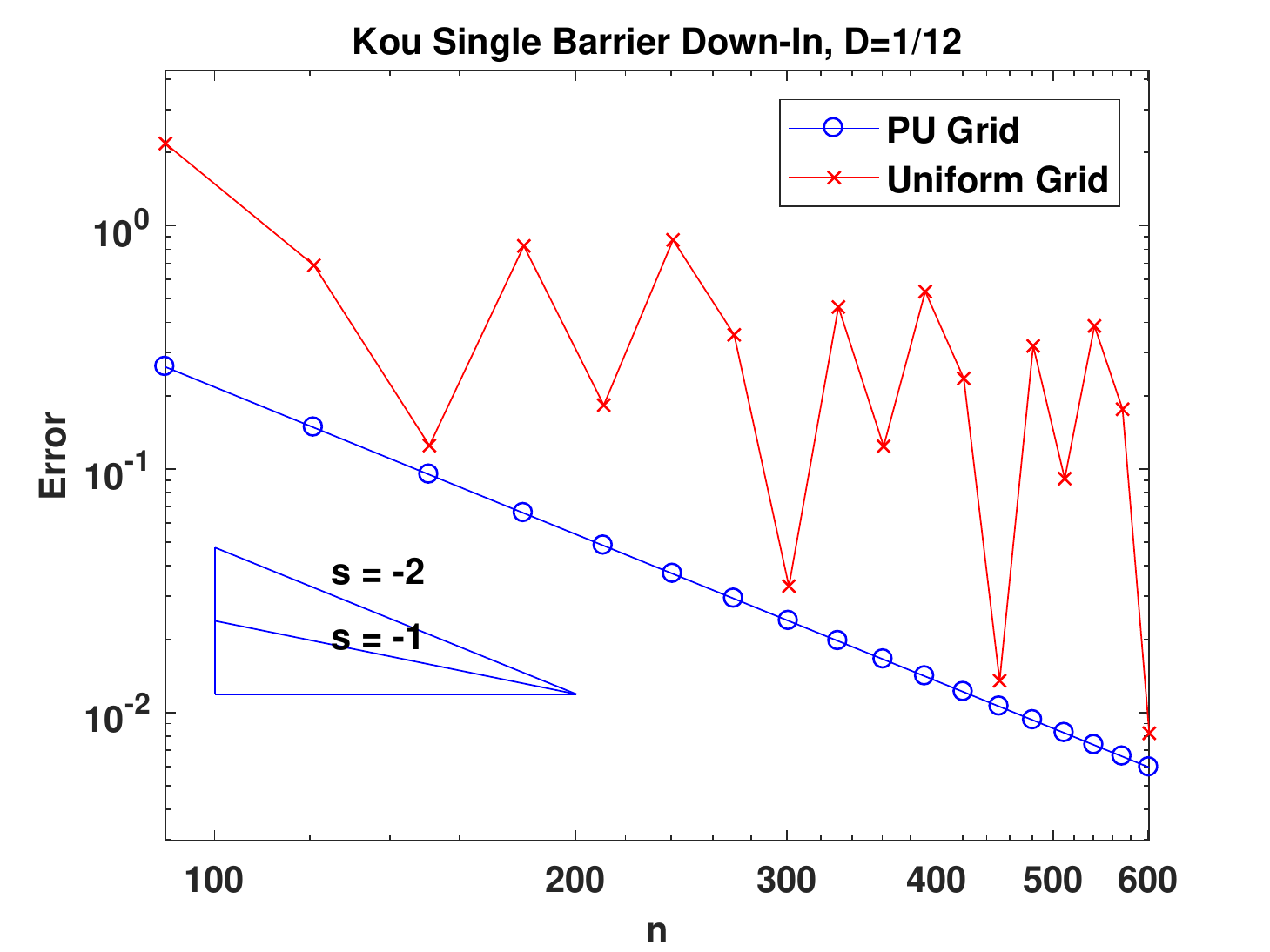}
	\includegraphics[width = 0.45 \textwidth]{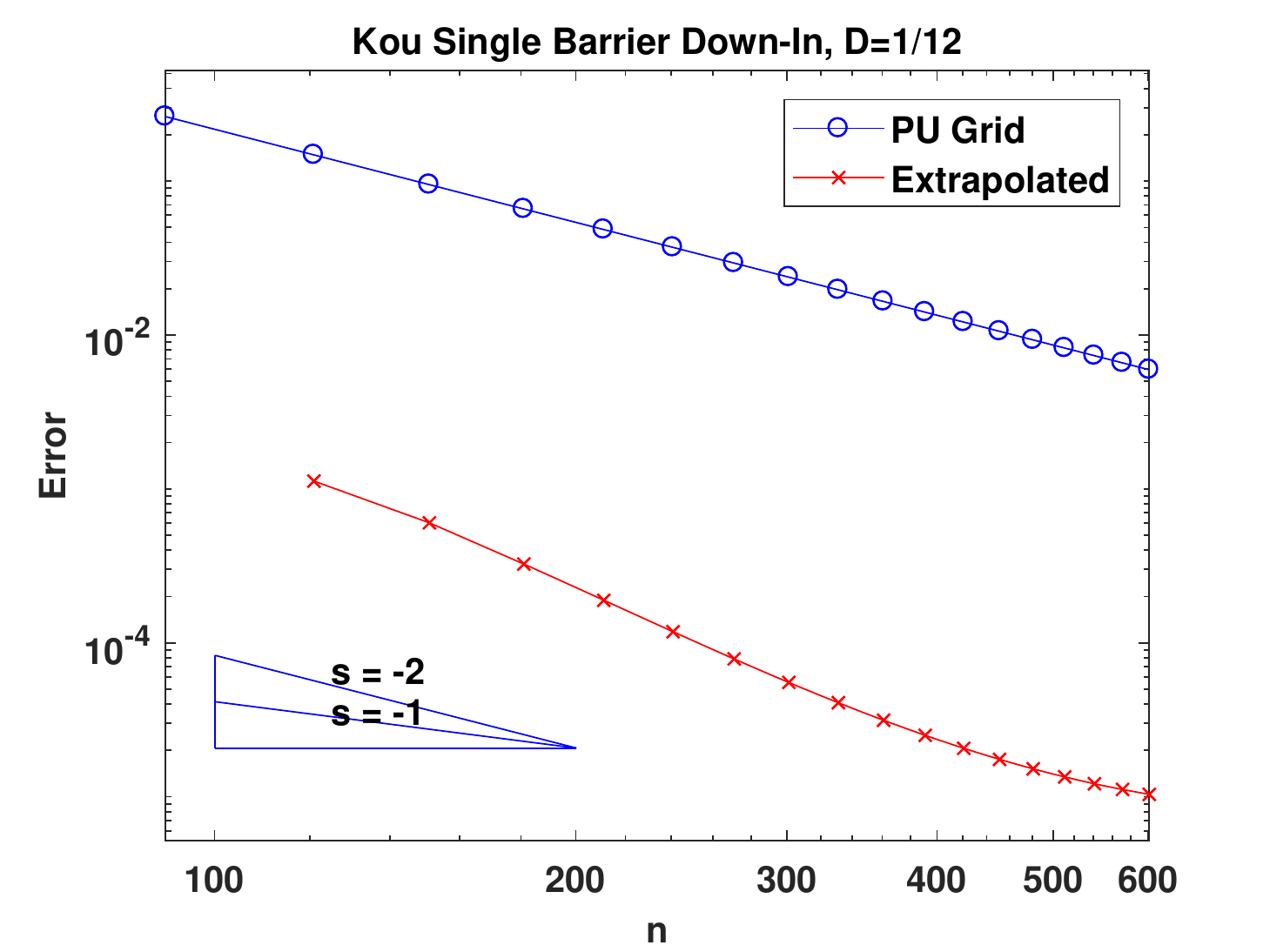}
	\includegraphics[width = 0.45 \textwidth]{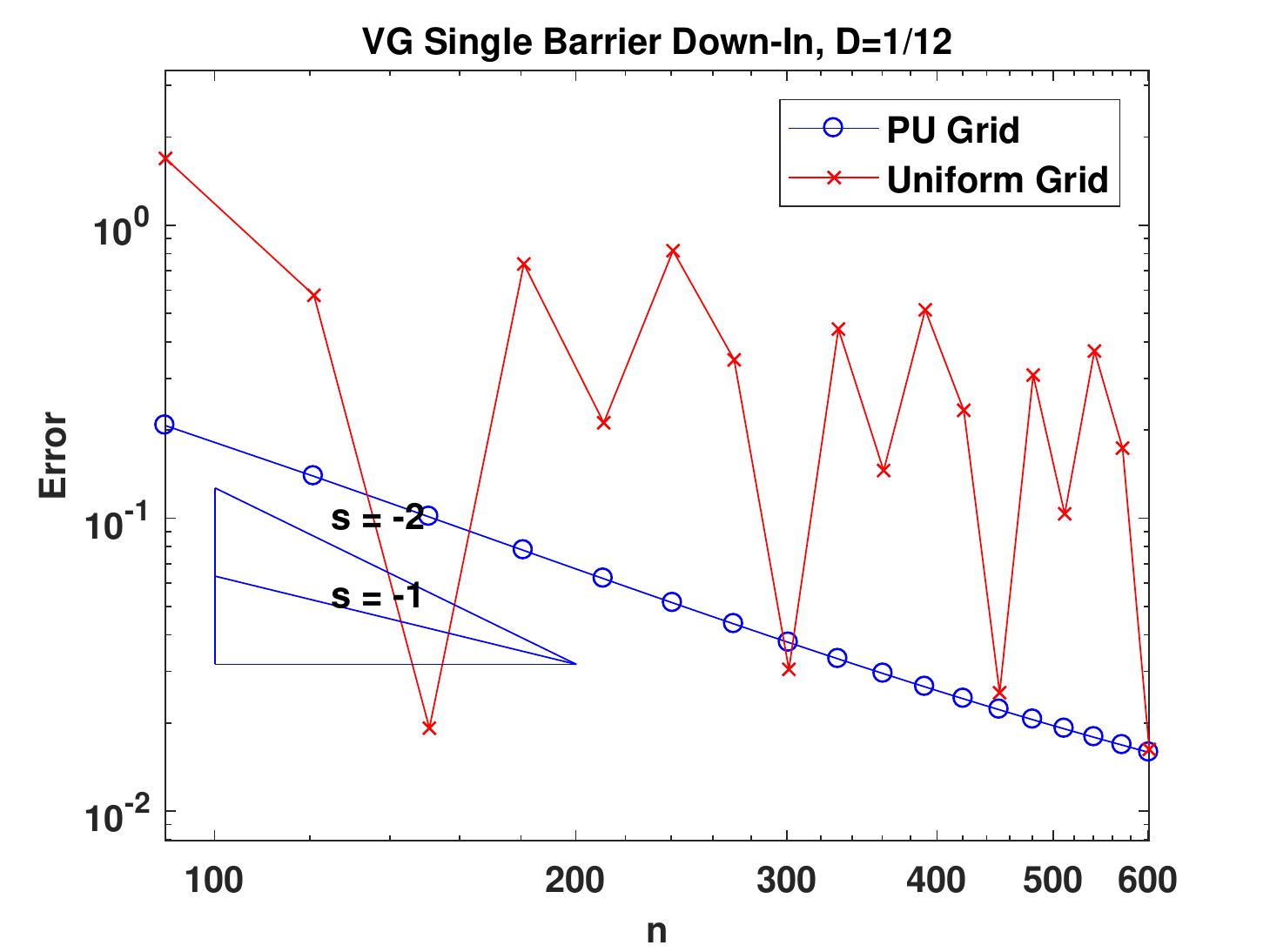}
	\includegraphics[width = 0.45 \textwidth]{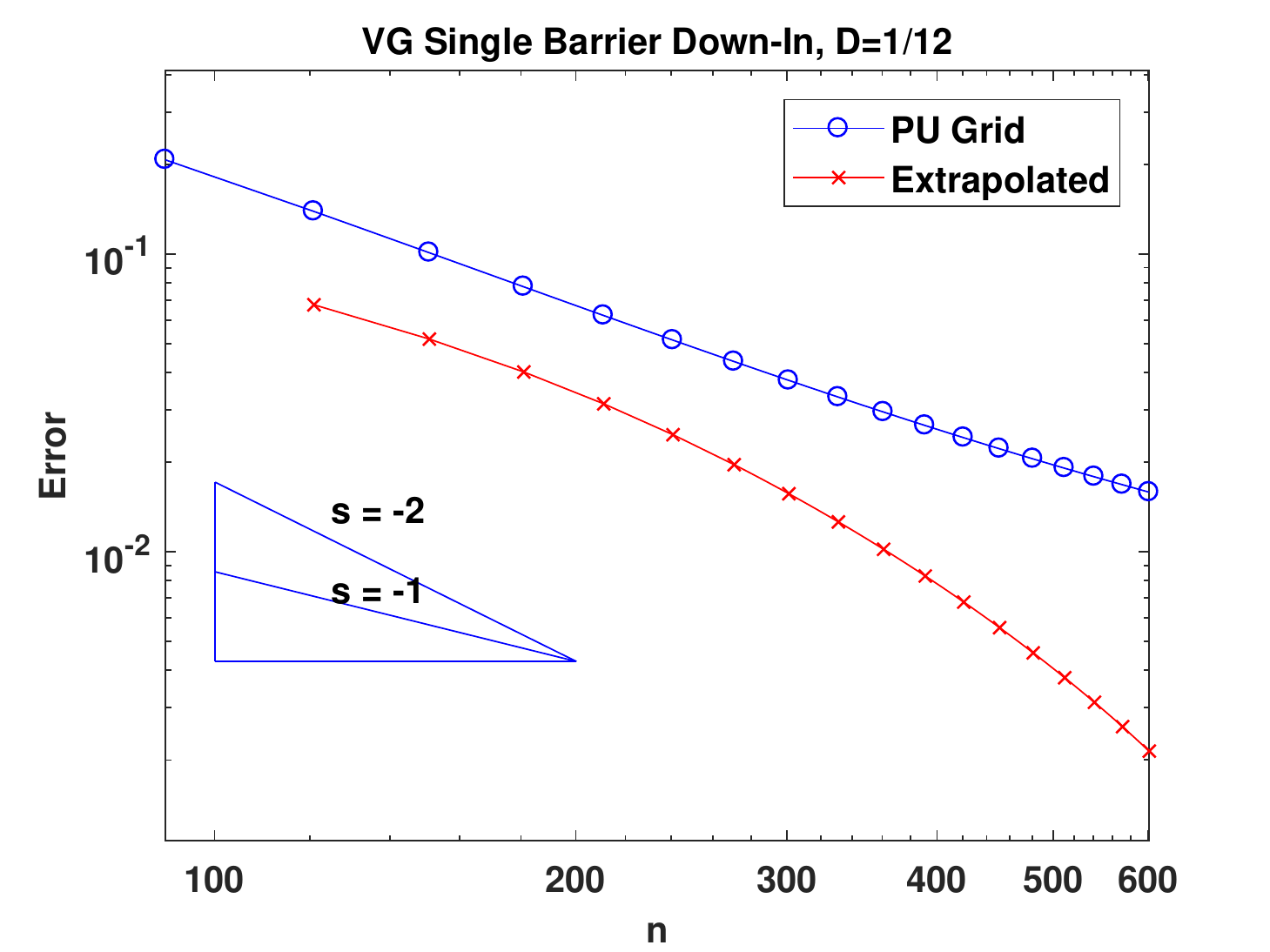}
	\caption{Convergence of CTMC approximation for a single-barrier down-and-in Parisian option under the BS model, Kou's model and the VG model. The benchmarks are obtained from CTMC approximation with very fine grids, and they match values from Monte Carlo simulation with a large number of replications. 
	}\label{fig:opt-pricing}
\end{figure}

\begin{figure}
	\centering
	\includegraphics[width=0.6\textwidth]{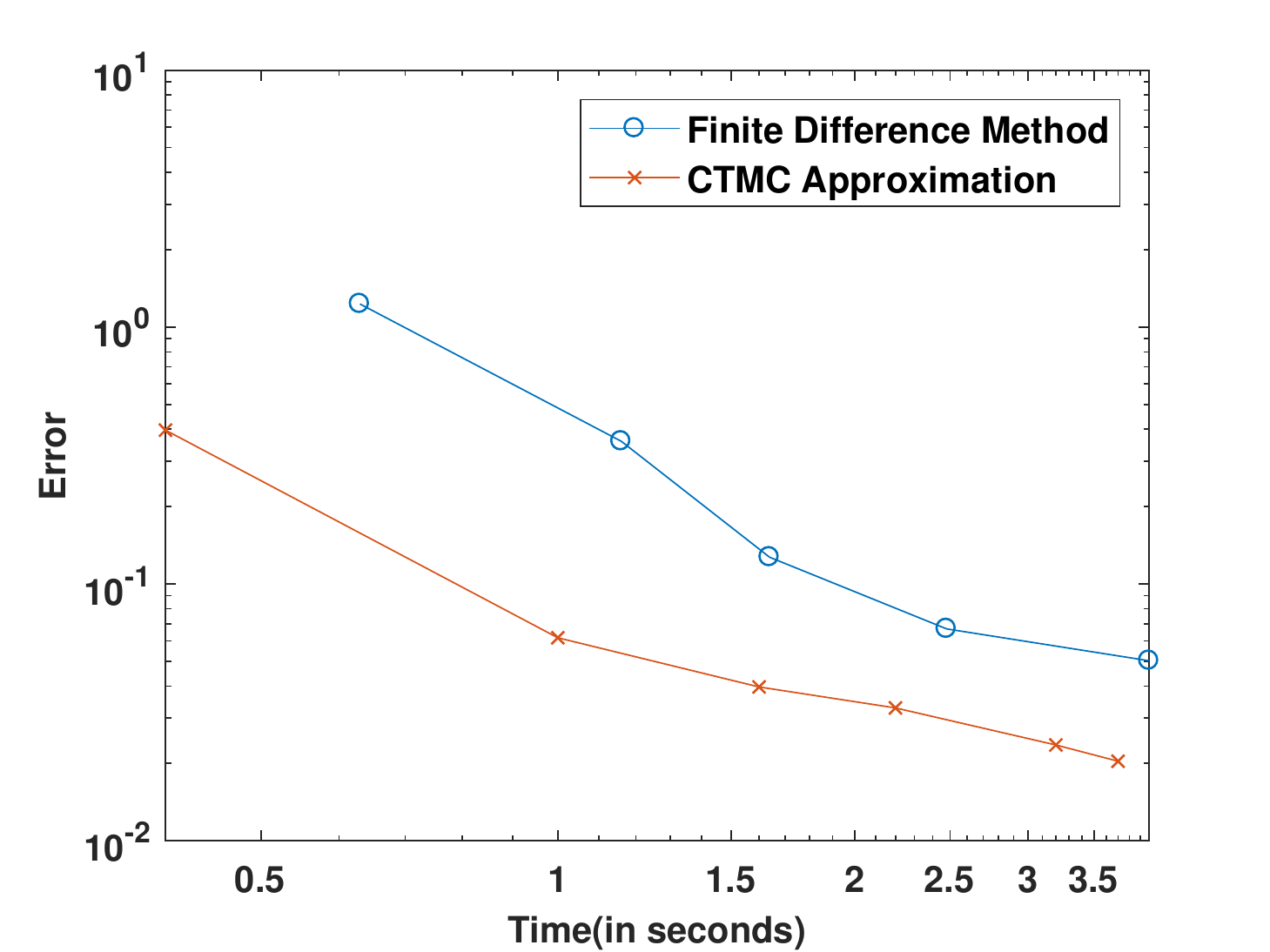}
	\caption{A comparison of CTMC approximation with finite difference method.}
	\label{fig:comparison-fdm}
\end{figure}

\begin{table}[htbp]
	\centering
	\begin{tabular}{ccccccccc}
		\hline
		$n$     & Exact & CTMC  & Error & Rel. Err. & Time/s & Extra. & Error & Rel. Err. \\
		\hline
		91    & 1.97866 & 1.58183 & 3.97E-01 & 20.06\% & 0.1   &       &       &  \\
		121   & 1.97866 & 1.74318 & 2.35E-01 & 11.90\% & 0.2   & 1.95326 & 2.54E-02 & 1.28\% \\
		151   & 1.97866 & 1.82432 & 1.54E-01 & 7.80\% & 0.2   & 1.96990 & 8.76E-03 & 0.44\% \\
		181   & 1.97866 & 1.87017 & 1.08E-01 & 5.48\% & 0.3   & 1.97513 & 3.53E-03 & 0.18\% \\
		211   & 1.97866 & 1.89839 & 8.03E-02 & 4.06\% & 0.3   & 1.97701 & 1.65E-03 & 0.08\% \\
		\hline
	\end{tabular}
	\caption{Errors for the original and extrapolated results for the down-and-in Parisian option under the BS model.}\label{tab:bs-down-in}
\end{table}

\begin{table}[htbp]
	\centering
	\begin{tabular}{ccccccccc}
		\hline
		$n$     & Exact & CTMC  & Error & Rel. Err. & Time/s & Extra. & Error & Rel. Err. \\
		\hline
		91    & 4.55552 & 4.29302 & 2.63E-01 & 5.76\% & 0.4   &       &       &  \\
		121   & 4.55552 & 4.40754 & 1.48E-01 & 3.25\% & 0.5   & 4.55665 & 1.13E-03 & 0.02\% \\
		151   & 4.55552 & 4.46071 & 9.48E-02 & 2.08\% & 0.6   & 4.55611 & 5.90E-04 & 0.01\% \\
		181   & 4.55552 & 4.48963 & 6.59E-02 & 1.45\% & 0.8   & 4.55584 & 3.15E-04 & 0.01\% \\
		211   & 4.55552 & 4.50709 & 4.84E-02 & 1.06\% & 1.1   & 4.55573 & 2.10E-04 & 0.00\% \\
		\hline
	\end{tabular}
	\caption{Errors for the original and extrapolated results for the down-and-in Parisian option under Kou's model.}\label{tab:kou-down-in}
\end{table}

\begin{table}[htbp]
	\centering
	\begin{tabular}{ccccccccc}
		\hline
		$n$     & Exact & CTMC  & Error & Rel. Err. & Time/s & Extra. & Error & Rel. Err. \\
		\hline
		391   & 1.05872 & 1.03818 & 2.05E-02 & 1.94\% & 11.5  &       &       &  \\
		421   & 1.05872 & 1.03961 & 1.91E-02 & 1.81\% & 13.8  & 1.05825 & 4.72E-04 & 0.04\% \\
		451   & 1.05872 & 1.04084 & 1.79E-02 & 1.69\% & 15.9  & 1.05810 & 6.19E-04 & 0.06\% \\
		481   & 1.05872 & 1.04192 & 1.68E-02 & 1.59\% & 17.3  & 1.05816 & 5.64E-04 & 0.05\% \\
		511   & 1.05872 & 1.04286 & 1.59E-02 & 1.50\% & 19.6  & 1.05793 & 7.89E-04 & 0.07\% \\
		\hline
	\end{tabular}%
	\caption{Errors for the original and extrapolated results for the down-and-in Parisian option under the VG model.}
	\label{tab:cgmy-down-in}%
\end{table}%

\begin{table}[htbp]
	\centering
	\begin{tabular}{ccccccccc}
		\hline
		$n$     & Exact & CTMC  & Error & Rel. Err. & Time/s & Extra. & Error & Rel. Err. \\
		\hline
		91    & 4.30229 & 3.98449 & 3.18E-01 & 7.39\% & 2.0     &       &       &  \\
		121   & 4.30229 & 4.12057 & 1.82E-01 & 4.22\% & 3.4   & 4.29775 & 4.54E-03 & 0.11\% \\
		151   & 4.30229 & 4.18514 & 1.17E-01 & 2.72\% & 6.5   & 4.30099 & 1.30E-03 & 0.03\% \\
		181   & 4.30229 & 4.22062 & 8.17E-02 & 1.90\% & 9.9   & 4.30184 & 4.51E-04 & 0.01\% \\
		211   & 4.30229 & 4.24215 & 6.01E-02 & 1.40\% & 13.8  & 4.30213 & 1.66E-04 & 0.00\% \\
		\hline
	\end{tabular}%
	\caption{Errors for the original and extrapolated results for the down-and-in Parisian option under the regime-switching BS model. The benchmark is obtained from CTMC approximation with very fine grids, and it matches the result of Monte Carlo simulation with a large number of replications. }
	\label{tab:rs-bs-down-in}%
\end{table}%


\section{Conclusions}\label{sec:conclusions}

Our approach based on CTMC approximation is an efficient, accurate and flexible way to deal with Parisian stopping times and the related option pricing problems. It is applicable to general one-dimensional Markovian models and its convergence can be proved in a general setting. Moreover, extensions to regime-switching and stochastic volatility models as well as more complicated problems involving Parisian stopping times can be readily developed. We are also able to derive sharp estimate of the convergence rate for diffusion models. In general, convergence of CTMC approximation is first order, but we can improve it to second order with a proper grid design which places the Parisian barrier on the grid and all the discontinuities of the payoff function midway between two adjacent grid points. Such a grid can be easily constructed using a piecewise uniform structure. Our grid design also attains second order convergence in jump models as shown by the numerical examples. Furthermore, it ensures smooth convergence for both diffusion and jump models, which enables us to apply Richardson's extrapolation to achieve significant reduction in the error. In future research, it would be interesting to develop effective hedging strategies for Parisian options based on our method. 
 
\section*{Acknowledgements}
The research of Lingfei Li was supported by Hong Kong Research Grant Council General Research Fund Grant 14202117. 
The research of Gongqiu Zhang was supported by National Natural Science Foundation of China Grant 11801423 and Shenzhen Basic Research Program Project JCYJ20190813165407555.

\appendix

\section{Several Extensions}\label{sec:extensions}

In this section, we generalize our method for some more complicated situations to demonstrate the versality of our approach. 

\subsection{Multi-Sided Parisian Options}

Define a Parisian stopping time for an arbitrary set as
\begin{align}
	\tau_A^D = \inf\{ t \ge 0 : 1_{\{ Y_t \in A \}} (t - g_{A, t}) \ge D \},\ g_{A, t} = \sup\{s\le t : Y_s \notin A\}.
\end{align}
If $A = (-\infty, L)$, then $\tau_A^D = \tau_{L, D}^-$ and if $A = (U, \infty)$, $\tau_A^D = \tau_{U, D}^+$. The multi-sided Parisian stopping time can be defined as
\begin{align}
	\tau_\mathcal{A}^D = \min\{ \tau_A^D: A \in \mathcal{A} \},
\end{align}
where $\mathcal{A}$ is a collection of subsets of $\mathbb{R}$. The double-barrier case corresponds to $\mathcal{A} = \{ (-\infty, L), (U, \infty) \}$ with $L<U$. We also let
\begin{align}
	T^-_A = \inf\{ t \ge 0: Y_t \notin A \},\ T^+_A = \inf\{ t \ge 0: Y_t \in A \}.
\end{align}

Consider a finite state CTMC $Y$. For any $x, y$ in its state space $\mathbb{S}$, let $h(q, x; y) = \mathbb{E}_x \left[ e^{-q \tau_{\mathcal{A}}^D} 1_{\{ Y_{\tau_{\mathcal{A}}^D} = y \}} \right]$. Define $B=\cup_{A\in\mathcal{A}}A$. 
Then, by the conditioning argument, we obtain
\begin{align}
h(q, x; y) &= \sum_{A \in \mathcal{A}} e^{-qD} \mathbb{E}_x \left[ 1_{\{ T_A^- \ge D, Y_D = y \}} \right] \times 1_{\{ x \in A \}} \\
	&\quad + \sum_{A \in \mathcal{A}} \sum_{z \notin A} \mathbb{E}_x\left[ e^{-q T_A^-} 1_{\{ T_A^- < D, Y_{T_A^-} = z \}} \right] h(q, z; y)  \times 1_{\{ x \in A \}} \\
	&\quad + \sum_{z \in B}\mathbb{E}_x \left[ e^{-q T_{B}^+} 1_{\{ Y_{T_{B}^+} = z \}} \right] h(q, z; y) \times 1_{\{ x \notin B \}}.
\end{align}
Let $v_A(D, x; y) = \mathbb{E}_x \left[ 1_{\{ T_A^- \ge D, Y_D = y \}} \right]$. It solves
\begin{align}
	\begin{cases}
		\frac{\partial v_A}{\partial D}(D, x; y) = {\pmb G} v_A(D, x; y) ,\ x \in A, D > 0\\
		v_A(D, x; y) = 0,\ x \notin A, D > 0,\\
		v_A(0, x; y) = 1_{\{ x = y \}}.
	\end{cases}
\end{align}
Let $u^-_A(q, D, x; z) = \mathbb{E}_x\left[ e^{-q T_A^-} 1_{\{ T_A^- < D, Y_{T_A^-} = z \}} \right]$, then it satisfies,
\begin{align}
	\begin{cases}
		\frac{\partial u^-_A}{\partial D}(q, D, x; z) = {\pmb G} u^-_A(q, D, x; z) - q u^-_A(q, D, x; z),\ x \in A, D > 0,\\
		u^-_A(q, D, x; z) = 1_{\{ x = z \}},\ x \notin A, D > 0,\\
		u^-_A(q, 0, x; z) = 0.
	\end{cases}
\end{align}
We can rewrite it as $u^-_A(q, D, x; z) = u^-_{1,A}(q, x; z) - u^-_{2,A}(q, D, x; z)$, and these two parts satisfy
\begin{align}
	\begin{cases}
		{\pmb G} u^-_{1,A}(q, x; z) - q u^-_{1,A}(q, x; z) = 0,\ x \in A,\\
		u^-_{1,A}(q, x; z) = 1_{\{ x = z \}},\ x \notin A,
	\end{cases}
\end{align}
and
\begin{align}
	\begin{cases}
		\frac{\partial u^-_{2,A}}{\partial{D}}(q, D, x; z) = {\pmb G} u^-_{2,A}(q, D, x; z) - q u^-_{2,A}(q, D, x; z),\ x \in A,\\
		u^-_{2,A}(q, D, x; z) = 0,\ x \notin A, D > 0,\\
		u^-_{2,A}(q, 0, x; z) = u^-_{1,A}(q, x; z).
	\end{cases}
\end{align}
Let $u^+(q, x; z) = \mathbb{E}_x \left[ e^{-q T_{B}^+} 1_{\{ Y_{T_{B}^+} = z \}} \right]$. It is the solution to
\begin{align}
	\begin{cases}
		{\pmb G} u^+(q, x; z) - q u^+(q, x; z) = 0,\ x \notin B,\\
		u^+(q, x; z) = 1_{\{ x = z \}},\ x \in B.
	\end{cases}
\end{align}
Let ${\pmb I}_A = \diag((1_{\{x \in A\}})_{x \in \mathbb{S}})$. The solutions to the above equations are given as follows:
\begin{align}
	& {\pmb V}_A = \left( v_A(D, x; y) \right)_{x,y \in \mathbb{S}}= \exp\left( {\pmb I}_A {\pmb G}D \right) {\pmb I}_A,\\
	& {\pmb U}_{1,A}^-(q) = \left( u^-_{1,A}(q, x; z) \right)_{x, z\in \mathbb{S}} = \left( {\pmb I} - {\pmb I}_A - {\pmb I}_A ({\pmb G} - q{\pmb I} ) \right)^{-1} ({\pmb I} - {\pmb I}_A),\\
	& {\pmb U}_{2,A}^-(q) = \left( u^-_{2,A}(q,D, x; z) \right)_{x, z\in \mathbb{S}} = \exp\left( {\pmb I}_A {\pmb G}D \right) {\pmb I}_A {\pmb U}_{1,A}^-(q),\\
	& {\pmb U}^-_A(q) = \left( u^-_{A}(q,D, x; z) \right)_{x, z\in \mathbb{S}} = {\pmb U}_{1,A}^-(q) - {\pmb U}_{2,A}^-(q),\\
	& {\pmb U}^+(q) = \left( u^+(q, x; z) \right)_{x, z\in \mathbb{S}} = \left( {\pmb I}_B - ({\pmb I} - {\pmb I}_{B}) ({\pmb G } - q{\pmb I}) \right)^{-1} {\pmb I}_{B}.
\end{align}
Then, we obtain
\begin{align}
	{\pmb H}(q) = \left( h(q, x; y) \right)_{x, y \in \mathbb{S}} = e^{-qD} \left( {\pmb I} - 
	\sum_{A \in \mathcal{A}} {\pmb I}_A {\pmb U}^-_A(q) - ({\pmb I} - {\pmb I}_{B}) {\pmb U}^+(q) \right)^{-1}  \sum_{A \in \mathcal{A}} {\pmb I}_A {\pmb V}_A.
\end{align}

Now consider a multi-sided Parisian option with price $u(t, x) = \mathbb{E}_x\left[ 1_{\{ \tau_\mathcal{A}^D \le t \}} f(Y_t) \right]$. Using the arguments in the single-sided case, we can derive that the Laplace transform $\tilde{u}(q, x) = \int_{0}^{\infty} e^{-qt} u(t, x) dt$ is given by
\begin{align}
	\widetilde{\pmb u}(q) = \left( \tilde{u}(q, x) \right)_{x \in \mathbb{S}} = {\pmb H}(q) (q{\pmb I} - {\pmb G})^{-1}{\pmb f}.
\end{align}

\subsection{Mixed Barrier and Parisian Options}\label{subsec:other-types}

CTMC approximation can also be generalized to derive the joint distribution of Parisian stopping times and first passage times. For example, \cite{dassios2016joint} introduce the so-called MinParisianHit option which is triggered either when the age of an excursion above $L$ reaches time $D$ or a barrier $B > L$ is crossed by the underlying asset price process $S_t$. The MinParisianHit option price can be approximated as
\begin{align}
	\text{minPHC}_i^u(t, x; L, D, B) = e^{-r_ft} \mathbb{E}_x \left[ 1_{\{ \tau_{L, D}^+ \wedge \tau_B^+ \le t \}} f(Y_t) \right].
\end{align}
where $Y_t$ is a CTMC with state space $\mathbb{S}$ and transition rate matrix ${\pmb G}$ to approximate the underlying price process. To price the option, it suffices to substitute $\tau_{L, D}^-$ by $\tau_{L, D}^+ \wedge \tau_B^+$ in the proof of Theorem \ref{thm:opt-price-lt} and find the Laplace transform of $\tau_{L, D}^+ \wedge \tau_B^+$ under the model $Y_t$. Let $h_1(q, x; y) = \mathbb{E}_x[e^{-q\tau_{L, D}^+ \wedge \tau_B^+} 1_{\{ \tau_{L, D}^+ \wedge \tau_B^+ = y \}}]$. Using the conditioning argument, we can show that $h_1(q, x, y)$ satisfies a linear system
\begin{align}
	h_1(q, x; y) &= 1_{\{ x\ge B, x = y \}}  + 1_{\{ x \le L \}} \sum_{z > L} \mathbb{E}_{x}\left[ e^{-q \overline{\tau}_L^+} 1_{\{ Y_{\overline{\tau}_L^+} = z  \}} \right] h_1(q, z; y) \label{eq:h1-ls-1} \\
	&\quad + 1_{\{ L < x < B \}} \mathbb{E}_x\left[ e^{-q\tau_B^+} 1_{\{ \tau_B^+ < D, Y_{\tau_B^+} = y \}} \right] \label{eq:h1-ls-2} \\
	&\quad + 1_{\{ L < x < B \}} \sum_{z \le L} \mathbb{E}_x\left[ e^{-q\overline{\tau}_L^-} 1_{\{ \overline{\tau}_L^- < D \le \tau_B^+,  Y_{\overline{\tau}_L^-} = z \}} \right] h_1(q, z; y) \label{eq:h1-ls-3} \\
	&\quad + 1_{\{ L < x < B \}} e^{-qD} \mathbb{E}_x\left[ 1_{\{ \overline{\tau}_L^- \wedge \tau_B^+ \ge D, Y_{D} = y \}} \right], \label{eq:h1-ls-4}
\end{align}
where $\overline{\tau}_L^+ = \inf\{ t \ge 0: Y_t > L \}$ and $\overline{\tau}_L^- = \inf\{ t \ge 0: Y_t \le L \}$ (they are slightly different from $\tau_L^+$ and $\tau_L^-$ defined in Section \ref{sec:Parisian-CTMC}). 

We next analyze each term. For \eqref{eq:h1-ls-1}, let $\overline{u}(q, x; z) = \mathbb{E}_{x}\left[ e^{-q \overline{\tau}_L^+} 1_{\{ Y_{\overline{\tau}_L^+} = z  \}} \right]$. It satisfies
\begin{align}
	\begin{cases}
		{\pmb G} \overline{u}(q, x; z) - q \overline{u}(q, x; z) = 0,\ x \in (-\infty, L] \cap \mathbb{S},\\
		\overline{u}(q, x; z) = 1_{\{ x = z \}},\ x \in (L, \infty) \cap \mathbb{S},\\
	\end{cases}
\end{align}
For \eqref{eq:h1-ls-2}, let $v(q, D, x; y) = \mathbb{E}_x\left[ e^{-q\tau_B^+} 1_{\{ \tau_B^+ < D, Y_{\tau_B^+} = y \}} \right]$. It is the solution to
\begin{align}
	\begin{cases}
		\frac{\partial v}{\partial D}(q, D, x; y) = {\pmb G} v(q, D, x; y) - q v(q, D, x; y),\ D >0,\ x \in (-\infty, B) \cap \mathbb{S},\\
		v(q, D, x; y) = 1_{\{ x = y \}},\ D>0,\ x \in [B, \infty) \cap \mathbb{S},\\
		v(q, 0, x; y) = 0,\ x \in \mathbb{S}.
	\end{cases}
\end{align}
$v(q, D, x; y)$ can be split as $v_1(q, x; y) - v_2(q, D, x; y)$ with $v_1(q, x ; y)$ satisfying
\begin{align}
	\begin{cases}
		{\pmb G} v_1(q,x; y) - qv_1(q, x; y) = 0,\ x \in (-\infty, B) \cap \mathbb{S},\\
		v_1(q, x; y) = 1_{\{ x = y \}},\ x \in [B, \infty) \cap \mathbb{S},\\
	\end{cases}
\end{align}
and $v_2(q, D, x; y)$ satisfying
\begin{align}
	\begin{cases}
		\frac{\partial v_2}{\partial D}(q, D, x; y) = {\pmb G} v_2(q, D, x; y) - q v_2(q, D, x; y),\ D >0,\ x \in (-\infty, B) \cap \mathbb{S},\\
		v_2(q, D, x; y) = 0,\ D>0,\ x \in [B, \infty) \cap \mathbb{S},\\
		v_2(q, 0, x; y) = v_1(q, x; y),\ x \in \mathbb{S}.
	\end{cases}
\end{align}
For \eqref{eq:h1-ls-3}, let $u(q, D, x; z) = \mathbb{E}_x\left[ e^{-q\overline{\tau}_L^-} 1_{\{ \overline{\tau}_L^- < D \le \tau_B^+,  Y_{\overline{\tau}_L^-} = z \}} \right]$. It solves
\begin{align}
	\begin{cases}
		\frac{\partial u}{\partial D}(q, D, x; z) = {\pmb G} u(q, D, x; z) - q u(q, D, x; z),\ D > 0,\ x \in (L, B) \cap \mathbb{S},\\
		u(q, D, x; z) = 1_{\{ x = z \}} \mathbb{E}_x \left[ 1_{\{ \tau_B^+ \ge D \}} \right],\ D>0,\ x \in (-\infty, L] \cap \mathbb{S},\\
		u(q, D, x; z) = 0,\ D > 0,\ x \in [B, \infty) \cap \mathbb{S},\\
		u(q, 0, x; z) = 0,\ x \in \mathbb{S}. 
	\end{cases}
\end{align}
$u(q, D, x; z)$ can be expressed as $u_1(q, x; z) - u_2(q, D, x; z)$ with $u_1(q, x; z)$ satisfying
\begin{align}
	\begin{cases}
		{\pmb G} u_1(q, x; z) - q u_1(q, x; z) = 0,\ x \in (L, B) \cap \mathbb{S},\\
		u_1(q, x; z) = 1_{\{ x = z \}} \mathbb{E}_x \left[ 1_{\{ \tau_B^+ \ge D \}} \right],\ x \in (-\infty, L] \cap \mathbb{S},\\
		u_1(q, x; z) = 0,\ \ x \in [B, \infty) \cap \mathbb{S},
	\end{cases}
\end{align}
and $u_2(q, D, x; z)$ satisfying
\begin{align}
	\begin{cases}
		\frac{\partial u_2}{\partial D}(q, D, x; z) = {\pmb G} u_2(q, D, x; z) - q u_2(q, D, x; z),\ D > 0,\ x \in (L, B) \cap \mathbb{S},\\
		u_2(q, D, x; z) = 0,\ D>0,\ x \in (-\infty, L] \cap \mathbb{S},\\
		u_2(q, D, x; z) = 0,\ D > 0,\ x \in [B, \infty) \cap \mathbb{S},\\
		u_2(q, 0, x; z) = u_1(q, x; z),\ x \in \mathbb{S}. 
	\end{cases}
\end{align}
For \eqref{eq:h1-ls-4}, let $w(D, x; y)  = \mathbb{E}_x\left[ 1_{\{ \overline{\tau}_L^- \wedge \tau_B^+ \ge D, Y_{D} = y \}} \right]$. It solves
\begin{align}
	\begin{cases}
		\frac{\partial w}{\partial D}(D, x; y) = {\pmb G} w(D, x; y),\ D>0,\ x \in (L, B) \cap \mathbb{S},\\
		w(D, x; y) = 0,\ D>0,\ x \in \mathbb{S} \backslash (L, B),\\
		w(0, x; y) = 1_{\{ x = y \}},\ x \in \mathbb{S}.
	\end{cases}
\end{align}
Let  $\overline{\pmb I}_L^+ = \operatorname{diag}\left((1_{\{ x > L \}})_{x \in \mathbb{S}}\right)$, $\overline{\pmb I}_L^- = \operatorname{diag}\left((1_{\{ x \le L \}})_{x \in \mathbb{S}}\right)$, ${\pmb I}_{L, B} = \overline{\pmb I}_L^+ {\pmb I}_B^-$ (they are slightly different from ${\pmb I}_L^+$ and ${\pmb I}_L^-$ defined in Section \ref{sec:Parisian-CTMC}). The solutions to the above equations are given by 
\begin{align}
	&\overline{\pmb U}(q) = \left( \overline{u}(q, x; z) \right)_{x, z \in \mathbb{S}} = \left( q\overline{\pmb I}_L^- - \overline{\pmb I}_L^-{\pmb G} + \overline{\pmb I}_L^+\right)^{-1} \overline{\pmb I}_L^+,\\
	&{\pmb V}_1(q)  = \left( v_1(q, x; y) \right)_{x, y \in \mathbb{S}}= \left( q {\pmb I}_B^- - {\pmb I}_B^- {\pmb G} + {\pmb I}_B^+ \right)^{-1} {\pmb I}_B^+,\\
	&{\pmb V}_2(q)  = \left( v_2(q, D, x; y) \right)_{x, y \in \mathbb{S}} = \exp\left( {\pmb I}_B^- {\pmb G} - q{\pmb I}_B^- \right) {\pmb I}_B^- {\pmb V}_1(q),\\
	&{\pmb U}_1(q) = \left( u_1(q, x; z) \right)_{x, z \in \mathbb{S}} = \left( q{\pmb I}_{L, B} - {\pmb I}_{L, B} {\pmb G} + {\pmb I} - {\pmb I}_{L, B} \right)^{-1} \overline{\pmb I}_L^- \operatorname{diag} \left(\exp\left( {\pmb I}_B^- {\pmb G} D \right) {\pmb e}_B^-\right) ,\\
	&{\pmb U}_2(q) = \left( u_2(q, D, x; z) \right)_{x, z \in \mathbb{S}} = \exp\left( {\pmb I}_{L, B} {\pmb G} - q {\pmb I}_{L, B} \right) {\pmb I}_{L, B} {\pmb U}_1(q),\\
	&{\pmb W} = \left( w(D, x; y) \right)_{x, y \in \mathbb{S}} = \exp\left( {\pmb I}_{L, B} {\pmb G} D \right) {\pmb I}_{L, B},
\end{align}
where ${\pmb e}_B^- = (1_{\{ x < B \}})_{x \in \mathbb{S}}$. Let ${\pmb H}_1(q) = \left( h_1(q,x; y) \right)_{x, y \in \mathbb{S}}$, which satisfies
\begin{align}
	{\pmb H}_1(q) = {\pmb I}_B^+ + {\pmb I}_{L, B} \left({\pmb V}_1(q) - {\pmb V}_2(q) \right) + {\pmb I}_{L, B} \left({\pmb U}_1(q) - {\pmb U}_2(q) \right) {\pmb H}_1(q) + \overline{\pmb I}_L^- \overline{\pmb U}(q) {\pmb H}_1(q) + e^{-qD} {\pmb I}_{L, B} {\pmb W}.
\end{align}
The solution is
\begin{align}
	{\pmb H}_1(q) = \left( {\pmb I} - {\pmb U}(q) \right)^{-1} {\pmb V}(q),
\end{align}
where ${\pmb U}(q) =  {\pmb I}_{L, B} \left({\pmb U}_1(q) - {\pmb U}_2(q) \right) + \overline{\pmb I}_L^- \overline{\pmb U}(q)$ and ${\pmb V}(q) = {\pmb I}_B^+ + {\pmb I}_{L, B} \left({\pmb V}_1(q) - {\pmb V}_2(q) \right) + e^{-qD} {\pmb I}_{L, B} {\pmb W}$. Consider the Laplace transform
\begin{align}
	\widetilde{u}_i(q, x) = \int_{0}^{\infty} e^{-qt} \text{minPHC}_i^u(t, x; L, D, B) dt,\ \Re(q) > 0. 
\end{align}
It can be obtained as
\begin{align}
	\widetilde{\pmb u}_i(q) = \left( \widetilde{u}_i(q, x) \right)_{x \in \mathbb{S}} = {\pmb H}_1(q + r_f) \left( (q + r_f) {\pmb I} - {\pmb G} \right)^{-1} {\pmb f}.
\end{align}

\subsection{Pricing Parisian Bonds}

Recently, \cite{Dassios2020} propose a Parisian type of bonds, whose payoff depends on whether the excursion of the interest rate above some level $L$ exceeds a given duration before maturity, i.e., $h(R_\tau) 1_{\{\tau < T \}}$, where $T$ is the bond maturity, $R$ is the short rate process, $h(\cdot)$ is the payoff function, and
\begin{align}
	\tau = \inf\{ t > 0: U_t = D \},\ U_t = t - \sup\{ s < t: R_s \le L \}.
\end{align}
Then its price can be written as,
\begin{align}
	P(T, x) = \mathbb{E}_x\left[ e^{-\int_{0}^{\tau} R_t dt} f(R_{\tau}) 1_{\{ \tau < T \}} \right],
\end{align}
where $\mathbb{E}_x[\cdot] = \mathbb{E}[\cdot|R_0 = x]$.
We can calculate its Laplace transform w.r.t. $T$ as,
\begin{align}
	\widetilde{P}(q, x) = \int_{0}^{\infty} e^{-qT} P(T, x) dT = \frac{1}{q} \mathbb{E}_x\left[ e^{-\int_{0}^{\tau} (q + R_t) dt} f(R_{\tau})  \right].
\end{align}
Suppose $R$ is a CTMC with state space $\mathbb{S}_R$ to approximate the original short rate model (e.g., the CIR model considered in \cite{Dassios2020}). Let $h(q, x) = \mathbb{E}_x\left[ e^{-\int_{0}^{\tau} (q + R_t) dt} f(R_{\tau})  \right]$ and $\overline{\tau}_L^+$ and $\overline{\tau}_L^-$ are defined as in Section \ref{subsec:other-types}. Then, using the conditioning argument, we obtain
\begin{align}
	h(q, x) &= \mathbb{E}_x\left[ e^{-\int_{0}^{D} (q + R_t) dt} f(R_{D}) 1_{\{ \overline{\tau}_L^- \ge D \}} \right] 1_{\{ x > L \}} \\
	&\quad + \sum_{z \le L} \mathbb{E}_x\left[ e^{-\int_{0}^{\overline{\tau}_L^-} (q + R_t) dt} 1_{\{ \overline{\tau}_L^- < D,  R_{\overline{\tau}_L^-} = z \}} \right] 1_{\{ x > L \}} h(q, z) \\
	&\quad + \sum_{z > L} \mathbb{E}_x\left[ e^{-\int_{0}^{\overline{\tau}_L^+} (q + R_t) dt} 1_{\{ R_{\overline{\tau}_L^+} = z \}} \right] 1_{\{ x \le L \}} h(q, z).
\end{align}
Let ${\pmb G}$ be the generator matrix of $R$.  Consider 
$v(q, D, x) =  \mathbb{E}_x\left[ e^{-\int_{0}^{D} (q + R_t) dt} f(R_{D}) 1_{\{ \overline{\tau}_L^- \ge D \}} \right]$.  It satisfies
\begin{align}
	\begin{cases}
		\frac{\partial v}{\partial D}(q, D, x) = ({\pmb G} - {\pmb R}_q) v(q, D, x),\ D > 0, x \in (L, \infty) \cap \mathbb{S}_R,\\
		v(q, D, x) = 0,\ D > 0, x \in (-\infty, L] \cap \mathbb{S}_R, \\
		v(q, 0, x) = f(x),\ x \in \mathbb{S}_R,
	\end{cases}
\end{align}
where ${\pmb R}_q = \text{diag}\left( (q + x)_{x \in \mathbb{S}_R} \right)$.
Let $u^-(q, D, x; z) = \mathbb{E}_x\left[ e^{-\int_{0}^{\overline{\tau}_L^-} (q + R_t) dt} 1_{\{ \overline{\tau}_L^- < D,  R_{\overline{\tau}_L^-} = z \}} \right]$. It solves
\begin{align}
	\begin{cases}
		\frac{\partial u^-}{\partial D}(q, D, x; z) = ({\pmb G} - {\pmb R}_q) u^-(q, D, x; z),\ D > 0, x \in (L, \infty) \cap \mathbb{S}_R, \\
		u^-(q, D, x; z) = 1_{\{  x = z \}},\ D>0, x \in (-\infty, L] \cap \mathbb{S}_R,\\
		u^-(q, 0, x; z) = 0.  
	\end{cases}
\end{align}
We can decomposed $u^-(q, D, x; z)$ as $u^-(q, D, x; z) = u_1^-(q, x; z) - u_2^-(q, D, x; z)$ with them satisfying
\begin{align}
	\begin{cases}
		({\pmb G} - {\pmb R}_q) u_1^-(q, x; z) = 0,\ x \in (L, \infty) \cap \mathbb{S}_R,\\
		u_1^-(q, x; z) = 1_{\{  x = z \}},\ x \in (-\infty, L] \cap \mathbb{S}_R,
	\end{cases}
\end{align}
and
\begin{align}
	\begin{cases}
		\frac{\partial u_2^-}{\partial D}(q, D, x; z) = ({\pmb G} - {\pmb R}_q) u_2^-(q, D, x; z),\ D > 0, x \in (L, \infty) \cap \mathbb{S}_R, \\
		u_2^-(q, D, x; z) = 0,\ D>0, x \in (-\infty, L] \cap \mathbb{S}_R,\\
		u_2^-(q, 0, x; z) = u_1^-(q, x; z),\ x \in \mathbb{S}_R.
	\end{cases}
\end{align}
Let $u^+(q, x; z) = \mathbb{E}_x\left[ e^{-\int_{0}^{\overline{\tau}_L^+} (q + R_t) dt} 1_{\{ R_{\overline{\tau}_L^+} = z \}} \right]$. It satisfies
\begin{align}
	\begin{cases}
		({\pmb G} - {\pmb R}_q) u^+(q, x; z) = 0,\ x \in (-\infty, L] \cap \mathbb{S}_R,\\
		u^+(q, x; z) = 1_{\{ x = z \}},\ x \in (L, \infty) \cap \mathbb{S}_R.
	\end{cases}
\end{align}
Let ${\pmb h}(q) = \left( h(q, x) \right)_{x \in \mathbb{S}_R}$, ${\pmb v}(q) = \left( v(q, D, x) \right)_{x \in \mathbb{S}_R}$, ${\pmb U}^-(q) = \left( u^-(q, D, x; z) \right)_{x, z \in \mathbb{S}_R}$, ${\pmb U}_1^-(q) = \left( u_1^-(q, x; z) \right)_{x, z \in \mathbb{S}_R}$, ${\pmb U}_2^-(q) = \left( u_2^-(q, D, x; z) \right)_{x, z \in \mathbb{S}_R}$, ${\pmb U}^+(q) = \left( u^+(q, x; z) \right)_{x, z \in \mathbb{S}_R}$ and ${\pmb f} = (f(x))_{x \in \mathbb{S}_R}$. They can be calculated as follows:
\begin{align}
	& {\pmb v}(q) = \exp\left( \overline{\pmb I}^+_L ({\pmb G} - {\pmb R}_q)  D \right) \overline{\pmb I}^+_L {\pmb f}, \\
	& {\pmb U}_1^-(q) = \left( \overline{{\pmb I}}^-_L - \overline{\pmb I}^+_L ({\pmb G} - {\pmb R}_q) \right)^{-1} \overline{{\pmb I}}^-_L,\\
	& {\pmb U}_2^-(q) = \exp\left( \overline{\pmb I}^+_L ({\pmb G} - {\pmb R}_q)  D \right) \overline{\pmb I}^+_L {\pmb U}_1^-(q), \\
	& {\pmb U}^-(q) = {\pmb U}_1^-(q) - {\pmb U}_2^-(q), \\
	& {\pmb U}^+(q) = \left( \overline{\pmb I}^+_L - \overline{\pmb I}^-_L ({\pmb G} - {\pmb R}_q)  \right)^{-1}  \overline{\pmb I}^+_L,\\
	& {\pmb h}(q) = \left( {\pmb I} - \overline{\pmb I}^+_L {\pmb U}^-(q) \overline{\pmb I}^- - \overline{\pmb I}^-_L {\pmb U}^+(q) \overline{\pmb I}^+_L \right)^{-1} {\pmb v}(q).
\end{align}
We then calculate $\widetilde{P}(q, x)$ by dividing $h(q,x)$ by $q$ and obtain the bond price $P(T, x)$ by Laplace inversion.

\subsection{Regime-Switching and Stochastic Volatility Models}\label{app:RS-SV}

Suppose the asset price $S_t = \zeta(X_t, \widetilde{v}_t)$ for some function $\zeta(\cdot, \cdot)$ and a regime-switching process $X_t$. The regime process $\widetilde{v}_t \in \mathbb{S}_v = \{ v_1, v_2, \cdots, v_m \}$, and in each regime $X_t$ is a general jump-diffusion. We approximate the dynamics of $X_t$ in each regime by a CTMC $\widetilde{X}_t$ with state space $\mathbb{S}_X = \{ x_1, x_2, \cdots, x_n \}$. Hence, $S_t$ can be approximated by $\widetilde{S}_t = \zeta(\widetilde{X}_t, \widetilde{v}_t)$. The analysis of single-barrier Parisian stopping times for this type of models can be done similarly as in Section \ref{sec:Parisian-CTMC}. Let 
\begin{align}
	\widetilde{{\pmb x}} = ((x_1, v_1), (x_2, v_1), \cdots, (x_n, v_1), \cdots, (x_1, v_m), (x_2, v_m), \cdots, (x_n, v_m))^\top \in \mathbb{R}^{nm}.
\end{align}
Let $\widetilde{{\pmb G}}$ be the generator matrix of $(\widetilde{X}_t, \widetilde{v}_t)$ and it can be constructed as
\begin{align}
	\widetilde{\pmb G} = \operatorname{diag}(\pmb G_1, \pmb G_2, \cdots, \pmb G_m) + \Lambda \otimes \pmb I \in \mathbb{R}^{nm \times nm},
\end{align}
where $\pmb G_i$ is the generator matrix of $\widetilde{X}_t$ in regime $v_i$, $\Lambda$ is the generator matrix of $\widetilde{v}_t$, 
$\otimes$ stands for Kronecker product and $\pmb I$ is the identity matrix in $\mathbb{R}^{n \times n}$. Let 
\begin{align}
	{\pmb H}(q) = (h(q, x, v; y, u))_{x, y \in \mathbb{S}_X, u, v \in \mathbb{S}_v},
\end{align} 
with 
\begin{align}
	h(q, x, v; y, u) = \mathbb{E}_{x, v}\left[ e^{-q \widetilde{\tau}_{L, D}^-} 1_{\{ \widetilde{X}_{\widetilde{\tau}_{L, D}^-} = y, \widetilde{v}_{\widetilde{\tau}_{L, D}^-} = u \}} \right], 
\end{align}
where $\widetilde{\tau}_{L, D}^-  = \inf\{ t \ge 0: 1_{\{ \widetilde{S}_t < L \}} (t - \widetilde{g}^-_{L, t}) \ge D \}$ with $\widetilde{g}^-_{L, t} = \sup \{ s \le t: \widetilde{S}_s \ge L  \}$. We can solve the Parisian problem in the same way as for 1D CTMCs. Let
\begin{align}
	&\widetilde{{\pmb V}} = \exp\left( \widetilde{\pmb I}_L^- \widetilde{{\pmb G}} D \right) \widetilde{\pmb I}_L^-, \label{eq:v-matrix-sv}\\
	&\widetilde{\pmb U}^+_1(q) = \widetilde{\pmb I}_L^- (q\widetilde{\pmb I}_L^- - \widetilde{\pmb I}_L^- \widetilde{\pmb G} + \widetilde{\pmb I}_L^+ )^{-1} \widetilde{\pmb I}_L^+, \label{eq:u1-matrix-sv}\\
	&\widetilde{\pmb U}^+_2(q)  = \widetilde{\pmb I}_L^- \exp\left( \widetilde{\pmb I}_L^- \widetilde{\pmb G} D \right) \widetilde{\pmb I}_L^- \widetilde{\pmb U}_1(q), \label{eq:u2-matrix-sv} \\
	&\widetilde{\pmb U}^-(q)  = \widetilde{\pmb I}_L^+(q\widetilde{\pmb I}_L^+ - \widetilde{\pmb I}_L^+ \widetilde{\pmb G} + \widetilde{\pmb I}_L^- )^{-1} \widetilde{\pmb I}_L^-,\label{eq:um-matrix-sv} \\
	& \widetilde{\pmb U}(q) = \widetilde{\pmb U}_1^+(q) - \widetilde{\pmb U}_2^+(q) + \widetilde{\pmb U}^-(q), 
\end{align}
where $\widetilde{\pmb I}_L^+ = \diag(1_{\{ \zeta(\pmb x) \ge L  \}})$ and $\widetilde{\pmb I}_L^- = \diag(1_{\{ \zeta(\pmb x) < L  \}})$.
Then, we have
\begin{align}
	\widetilde{\pmb H}(q) = e^{-qD} \widetilde{\pmb V} + \widetilde{\pmb U}(q) \widetilde{\pmb H}(q),
\end{align}
and $\widetilde{\pmb H}(q)$ can be found by solving this linear system,
\begin{align}
	\widetilde{\pmb H}(q) = e^{-qD} (\widetilde{\pmb I} - \widetilde{\pmb U}(q))^{-1} \widetilde{\pmb V}. \label{eq:H-matrix-sv}
\end{align}
For option pricing, let $\widetilde{{\pmb u}}(q) = \left( \widetilde{u}(q, x, v) \right)_{x\in \mathbb{S}_X, v \in \mathbb{S}_v}$ with $\widetilde{u}(q, x, v) $ being the Laplace transform of option prices 
\begin{align}
	\widetilde{u}(q, x, v) = \int_{0}^{\infty} e^{-qt} \mathbb{E}_{x, v}\left[ f(\widetilde{X}_t, \widetilde{v}_t) 1_{\{ \widetilde{\tau}_{L, D}^- \le t \}} \right] dt.
\end{align} 
And then $\widetilde{{\pmb u}}(q)$ can be found as follows,
\begin{align}
	\widetilde{{\pmb u}}(q) = e^{-qD} (\widetilde{\pmb I} - \widetilde{\pmb U}(q))^{-1} \widetilde{\pmb V} \left(q \widetilde{\pmb I} - \widetilde{\pmb G} \right)^{-1} \widetilde{\pmb f} .
\end{align}
where $\widetilde{{\pmb f}} = {\pmb e}_m \otimes (f(x_1), f(x_2), \cdots, f(x_n))$ with ${\pmb e}_m$ being an all-one vector in $\mathbb{R}^m$.

\begin{remark}[Stochastic volatility models]
	\cite{cui2018general} show that general stochastic volatility models can be approximated by a regime-switching CTMC. Consider
	\begin{align}
		\left\{\begin{array}{l}{d S_{t}=\omega\left(S_{t}, v_{t}\right) d t+m\left(v_{t}\right) \Gamma\left(S_{t}\right) d W_{t}^{(1)}}, \\ {d v_{t}=\mu\left(v_{t}\right) d t+\sigma\left(v_{t}\right) d W_{t}^{(2)}},\end{array}\right.
	\end{align}
	where $[W^{(1)}, W^{(2)}]_t = \rho t$ with $\rho \in [-1, 1]$. As in \cite{cui2018general}, consider $X_t = g(S_t) - \rho f(v_t)$ with $g(x):=\int_{0}^{x} \frac{1}{\Gamma(u)} d u$ and $f(x):=\int_{0}^{x} \frac{m(u)}{\sigma(u)} d u$. It follows that
	\begin{align}
		dX_t =\theta(X_t, v_t) d t+\sqrt{1-\rho^{2}} m\left(v_{t}\right) d W_{t}^{*},
	\end{align}
	where $W^\ast$ is a standard Brownian motion independent of $W^{(2)}$ and,
	\begin{align}
		\theta(x, v) = \left(\frac{\omega\left(\zeta(x, v), v\right)}{\Gamma\left(\zeta(x, v)\right)}-\frac{\Gamma^{\prime}\left(\zeta(x, v)\right)}{2} m^{2}\left(v\right)-\rho h\left(v\right)\right),
	\end{align}
	with $\zeta(x, v) := g^{-1}(x + \rho f(v))$ and $h(x): =\mu(x) \frac{m(x)}{\sigma(x)}+\frac{1}{2}\left(\sigma(x) m^{\prime}(x)-\sigma^{\prime}(x) m(x)\right)$.
	Then we can employ a two-layer approximation to $(X_t, v_t)$: first construct a CTMC $\widetilde{v}_t$ with state space $\mathbb{S}_v = \{ v_1, \cdots, v_m \}$ and generator matrix $\Lambda \in \mathbb{R}^{m \times m}$ to approximate $v_t$, and then for each $v_l \in \mathbb{S}_v$, construct a CTMC with state space $\mathbb{S}_X = \{ x_1, \cdots, x_n \}$ and generator matrix $\mathcal{G}_v$ to approximate the dynamics of $X_t$ conditioning on $v_t=v_l$. Then, $(X_t, v_t)$ is approximated by a regime-switching CTMC $(\widetilde{X}_t, \widetilde{v}_t)$ where $\widetilde{X}_t$ transitions on $\mathbb{S}_X$ following $\mathcal{G}_v$ when $\widetilde{v}_t = v$ for each $v \in \mathbb{S}_v$, and $\widetilde{v}_t$ evolves according to its transition rate matrix $\Lambda$.
\end{remark}

\section{Proofs}\label{sec:proofs}

\begin{proof}[Proof of Lemma \ref{lmm:priors-eigenfunctions}]
	(1) This result can be found in Proposition 1 and Corollary 1 of \cite{zhang2018analysis}.
	
	(2) We first apply Liouville transform to the eigenvalue problem. Let $y = \int_{l}^{x} \frac{1}{\sigma(z)} dz$, $B = \int_{l}^{b} \frac{1}{\sigma(z)} dz$,  and $Q(y) = \frac{((m(x(y))/s(x(y)))^{1/4})''}{(m(x(y))/s(x(y)))^{1/4}}$ and $\mu_k^+(B) = \lambda^+_k(b)$. Then the eigenvalue problem is cast in the Liouville normal form as follows (\cite{fulton1994eigenvalue}, Eqs.(2.1-2.5))
	\begin{align}
		\begin{cases}
			-\partial_{yy} \phi_k^+(y, B) + Q(y)  \phi_k^+(y, B) = -\mu_k^+(B) \phi_k^+(y, B),\ y \in (0, B),\\
			\phi_k^+( 0, B) = \phi_k^+(B, B) = 0.
		\end{cases}
	\end{align}
	As shown in Eq.(2.13) of \cite{fulton1994eigenvalue},
	\begin{align}
		\left\| \varphi_k^+( \cdot, b) \right\|_2^2 = \int_{l}^{b} \varphi_k^+( x, b)^2 m(x) dx = \int_{0}^B \phi_k^+(y, B)^2 dy = \left\| \phi_k^+(\cdot, B) \right\|_2^2.
	\end{align}
	Hence the normalized eigenfunction can be recovered as 
	\begin{align}
		\varphi_k^+( x, b) = \left(\frac{s(x(y))}{m(x(y))}\right)^{1/4} \frac{\phi_k^+(y, B)}{\left\| \phi_k^+(\cdot, B) \right\|}.
	\end{align}
	Then we start to study the sensitivities of $\phi_k^+( y, B)$ and $\mu_k^+( B)$ w.r.t. $y$ and $B$ from which we can obtain those of $\varphi_k^+( x, b)$ and $\lambda_k^+( b)$ w.r.t. $x$ and $b$. Let $s_k(B) = \sqrt{\mu_k^+( B)}$. By \cite{fulton1994eigenvalue}, Eq.(3.7), we have
	\begin{align}
		\phi_k^+( y, B) = \frac{\sin(s_k(B) y) }{s_k(B) } + \frac{1}{s_k(B) } \int_{0}^{y}  \sin(s_k(B) (y - z)) Q(z) \phi_k^+( z, B) dz.   \label{eq:phi-equation} 
	\end{align}
	By Gronwall's inequality, $\phi_k^+( y, B) = \mathcal{O}(1/k)$. Taking differentiation on both sides yields
	\begin{align}
		\partial_y \phi_k^+( y, B) &= \sin(s_k(B) y)+\int_{0}^{y}  \cos(s_k(B) (y - z)) Q(z) \phi_k^+( z, B) dz, \label{eq:phi-y-equation} \\
		\partial_y^2 \phi_k^+( y, B) &= s_k(B) \cos(s_k(B) y) + Q(y) \phi_k^+( y, B) \\
		&\quad -s_k(B) \int_{0}^{y}  \sin(s_k(B) (y - z)) Q(z) \phi_k^+( z, B) dz, \label{eq:phi-yy-equation}\\
		\partial_y^3 \phi_k^+( y, B) &= -s_k^2(B) \sin(s_k(B) y) + Q'(y) \phi_k^+( y, B) + Q(y) \partial_y \phi_k^+( y, B) \\
		&\quad -s_k^2(B) \int_{0}^{y}  \cos(s_k(B) (y - z)) Q(z) \phi_k^+( z, B) dz. \label{eq:phi-yyy-equation}
	\end{align}
	Thus,
	\begin{align}
		\left| \partial_y \phi_k^+( y, B) \right| = \mathcal{O}(1),\ \left| \partial_y^2 \phi_k^+( y, B) \right| = \mathcal{O}(k),\ \left| \partial_y^3 \phi_k^+( y, B) \right| = \mathcal{O}(k^2).
	\end{align}
	By Theorem 3.2 in \cite{kong1996dependence}, we have
	\begin{align}
		s_k'(B) = -\frac{(\partial_y \phi_k^+(B, B))^2}{ \left\| \phi_k^+(\cdot, B) \right\|^2 }. \label{eq:sk-equation}
	\end{align}
	As in Eq.(6.11) of \cite{fulton1994eigenvalue}, $1/\left\| \phi_k^+(\cdot, B) \right\| = O(k)$.  \eqref{eq:phi-y-equation} implies $\partial_y \phi_k^+( B, B) = \mathcal{O}(1)$. Therefore, $s_k'(B)  = \mathcal{O}(k^2)$.
	
	Differentiating on both sides of \eqref{eq:phi-equation}, we obtain
	\begin{align}
		&\partial_B \phi^+_{k}( y, B)\\
		&= \frac{y \cos \left(s_{k}(B) y\right) s_{k}^{\prime}(B) s_{k}(B)-\sin \left(s_{k}(B) y\right) s_{k}^{\prime}(B)}{s_{k}(B)^{2}} \\ 
		&\quad +\int_{0}^{y} \frac{(x-z) \cos \left(s_{k}(B)(y-z)\right) s_{k}^{\prime}(B) s_{k}(B)-\sin \left(s_{k}(B)(y-z)\right) s_{k}^{\prime}(B)}{s_{k}(B)^{2}} Q(z) \phi^+_{k}( z, B) d z \\ 
		&\quad +\int_{0}^{y} \sin \left(s_{k}(B)(y-z)\right) Q(z) \partial_{B} \phi^+_{k}( z, B) d z.
	\end{align}
	Then there exist constants $c_2, c_3 > 0$ such that $\left|\partial_B \phi^+_{k}( y, B)\right| \le c_2 k + c_3 \int_{0}^{y} \left|\partial_B \phi^+_{k}( z, B)\right| dz$. Applying Gronwall's inequality again shows
	\begin{align}
		\left|\partial_B \phi^+_{k}( y, B)\right| \le c_2k \exp\left( c_3 B \right) = \mathcal{O}(k).
	\end{align}
	Taking differentiation w.r.t. $B$ on both sides of \eqref{eq:phi-y-equation}, \eqref{eq:phi-yy-equation}, \eqref{eq:phi-yyy-equation}, and applying the above estimate, we can obtain
	\begin{align}
		\left| \partial_B \partial_y \phi_k^+( y, B) \right| = \mathcal{O}(k^2),\ \left| \partial_B \partial_y^2 \phi_k^+( y, B) \right|  = \mathcal{O}(k^3),\ \left| \partial_B \partial_y^3 \phi_k^+( y, B) \right|  = \mathcal{O}(k^4).
	\end{align}
	We can also derive that
	\begin{align}
		\partial_B  \left\| \phi_k^+( \cdot, B) \right\|^2 = (\phi_k^+(\cdot, B))^2 + 2\int_{0}^{B} \partial_B \phi_k^+( z, B) \phi_k^+( z, B) dz = O(1/k).
	\end{align}
	Differentiating on both sides of \eqref{eq:sk-equation}, we obtain
	\begin{align}
		s_k''(B) = \frac{\partial_B \left\|  \phi_k^+(\cdot, B) \right\|^2 (\partial_y \phi_k^+(B, B))^2 - \left\|  \phi_k^+(\cdot, B) \right\|^2 \partial_B (\partial_y \phi_k^+(B, B))^2}{ \left\| \phi_k^+(\cdot, B) \right\|^4} = \mathcal{O}(k^4).
	\end{align}
	Subsequently, we get 
	\begin{align}
		\partial_B \mu_k( B) = 2s_k(B) s_k'(B) = \mathcal{O}(k^3),\ \partial_{BB} \mu_k( B) = 2s_k(B) s_k''(B) + 2 (s_k'(B))^2 = \mathcal{O}(k^5).
	\end{align}
	
	After some calculations based on the relationship between $(\lambda_k^+( b) , \varphi_k^+ ( x, b))$ and $(\mu_k^+( B), \phi_k^+(y, B))$, we obtain
	\begin{align}
		\partial_b \lambda^+_k( b) = \mathcal{O}(k^3),\ \partial_{bb} \lambda^+_k( b) = \mathcal{O}(k^5),
	\end{align}
	and
	\begin{align}
		\left| \partial_b  \varphi_k^+ ( x, b) \right| = \mathcal{O}(k^2),\ \left| \partial_b \partial_x \varphi_k^+ ( x, b) \right| = \mathcal{O}(k^3),\ \left| \partial_b \partial_x^2 \varphi_k^+ ( x, b) \right|  = O(k^4),\ \left| \partial_b \partial_x^3 \varphi_k^+ ( x, b) \right|  = O(k^5).
	\end{align}
	By Proposition 2 in \cite{zhang2018analysis}, we have $\lambda_{n, k}^+  = \lambda_k( L^+) + \mathcal{O}(k^4 \delta_n^2)$, and hence it holds that
	\begin{align}
		\lambda_{n, k}^+  = \lambda_k( L) + k^3\mathcal{O}(L^+ - L) + \mathcal{O}(k^4 \delta_n^2),
	\end{align}
	For \eqref{eq:eigenfunction-error-boundary}, by Proposition 3 in \cite{zhang2018analysis}, we have
	\begin{align}
		\varphi_{n, k}^+( x) = \varphi_k^+( x, L^+) + \mathcal{O}(k^4 \delta_n^2) = \varphi_k^+( x, L) + k \mathcal{O}(L^+ - L) + \mathcal{O}(k^4 \delta_n^2).
	\end{align}
	The same proposition also shows $\nabla^+ \varphi^+_{n, k}( L^-) = \nabla^+ \varphi^+_k( L^-; L^+) + \mathcal{O}(k^6 \delta_n^2)$. Thus, we obtain
	\begin{align}
		&\varphi^+_{n, k}( L^-) = \varphi^+_k( L^-, L^+) +  \mathcal{O}(k^6 \delta_n^3) \\
		&= \varphi^+_k( L^-, L^+) - \varphi^+_k( L, L^+) + \varphi^+_k( L, L^+)- \varphi^+_k( L^+, L^+) +  \mathcal{O}(k^6 \delta_n^3) \\
		&= \partial_x \varphi^+_k( L, L^+) (L^- - L) + \frac{1}{2} \partial_{xx} \varphi^+_k( L, L^+) (L^- - L)^2 \\
		&\quad - \partial_x \varphi^+_k( L, L^+) (L^+ - L) - \frac{1}{2} \partial_{xx} \varphi^+_k( L, L^+) (L^+ - L)^2 +   \mathcal{O}(k^6 \delta_n^3) \\
		&= -\partial_x \varphi^+_k( L, L^+) \delta^+ L^- + \frac{1}{2} \partial_{xx} \varphi^+_k( L, L^+) (\delta^+ L^-)^2 + \mathcal{O}(k^6 \delta_n^3) \\
		&\quad + \partial_{xx} \varphi^+_k( L, L^+) \delta^+ L^-(L^+ - L) +  \mathcal{O}(k^6 \delta_n^3) \\
		&= -\partial_x \varphi^+_k( L, L) \delta^+ L^- + \frac{1}{2} \partial_{xx} \varphi^+_k( L, L) (\delta^+ L^-)^2 + k^4 \delta^+ L^- \mathcal{O}(L^+ - L) +  \mathcal{O}(k^6 \delta_n^3). \label{eq:eigenfunction-boudary-error-deriv}
	\end{align}
	The rest of the results in the second part follows from Lemma 3, Lemma 5 and Lemma 6 in \cite{zhang2018analysis}. 
	
	(3) These results can be proved using arguments similar to those for Lemma \ref{lmm:prior-u1p} and Proposition \ref{prop:conv-unp}.
\end{proof}

\begin{proof}[Proof of Lemma \ref{lmm:prior-u1p}]
	Let $\psi^+(\cdot)$ and $\psi^-(\cdot)$ be two independent solutions to the Sturm-Liouville problem $\mu(x) \psi'(x) + \frac{1}{2} \sigma^2(x) \psi''(x) - q\psi(x) = 0$, where $\psi^+(\cdot)$ is strictly increasing and $\psi^-(\cdot)$ is strictly decreasing and they are $C^4$ by Theorem 6.19 in \cite{gilbarg2015elliptic}. Then we can construct $u_1^+(q,  x, b)$ as
	\begin{align}
		&u_1^+(q, x, b) = \frac{\psi^+(l) \psi^-(x) - \psi^+(x) \psi^-(l)}{\psi^+(l) \psi^-(b) - \psi^+(b) \psi^-(l)},
	\end{align}
	from which it is easy to see $u_1^+(q, x, b)\in C^3$ in $b$.
	
	Similar to Theorem 3.22 in \cite{li2018error}, we can show that $u_{1, n}^+(q, x; L^+) = u_1^+(q, x, L^+) + \mathcal{O}(h_n^2)$ and hence we have the first equation in the lemma due to the smoothness of $u_1^+(q, x; b)$ in $b$. For the second one, let $e_n(x) = u_{1, n}^+(q, x; L^+) - u^+_1(q, x; L^+)$.  For $x \in \mathbb{S}_n \cap (l, L^+)$, the following holds:
	\begin{align}
		{\pmb G}_n e_n(x) &= {\pmb G}_n u_{1, n}^+(q, x; L^+) - \left({\pmb G}_n u^+_1(q, x; L^+) - \mathcal{G} u^+_1(q, x, L^+) \right) + \mathcal{G} u^+_1(q, x, L^+) \\
		&= q \left( u_{1, n}^+(q, x; L^+) -  u_1^+(q, x; L^+) \right) - \left({\pmb G}_n u_1^+(q, x, L^+) - \mathcal{G} u_1^+(q, x, L^+) \right) = \mathcal{O}(\delta_n^2).
	\end{align}
	Therefore, for any $x, y \in \mathbb{S}_n \cap [l, L^-]$, we get
	\begin{align}
		\frac{1}{s_n(y)} \nabla^+ e_n(y) - \frac{1}{s_n(x)} \nabla^+ e_n(x) = \sum_{z \in \mathbb{S}_n \cap [x^+, y]} m_n(z) \delta z {\pmb G}_n e_n(z) = \mathcal{O}(\delta_n^2).
	\end{align}
	We also have $\sum_{z \in \mathbb{S}_n \cap [l, L^-]} \delta^+ z \nabla^+ e_n(z) = e_n(L^+) - e_n(l) = 0$, from which we conclude that there must exist a pair of $x, y \in \mathbb{S}_n \cap [l, L^-]$ such that $\frac{1}{s_n(y)} \nabla^+ e_n(y)  \frac{1}{s_n(x)} \nabla^+ e_n(x) \le 0$. In this case, we obtain 
	\begin{align}
		\left| \frac{1}{s_n(x)} \nabla^+ e_n(x) \right| \le \left| \frac{1}{s_n(y)} \nabla^+ e_n(y) - \frac{1}{s_n(x)} \nabla^+ e_n(x)  \right| = \mathcal{O}(\delta_n^2).
	\end{align}
	It follows that 
	\begin{equation}
		\left| \frac{1}{s_n(L^-)} \nabla^+ e_n(L^-) \right| \le \left|  \frac{1}{s_n(x)} \nabla^+ e_n(x)  \right| + \mathcal{O}(\delta_n^2)  = \mathcal{O}(\delta_n^2), 
	\end{equation}
	and hence $\nabla^+ e_n(L^-) = \mathcal{O}(\delta_n^2)$ holds. Therefore, $u_{1, n}^+(q, L^-; L^+) - u_1^+(q, L^-, L^+) = \delta^+ L^- \nabla^+ e_n(L^-) = \mathcal{O}(\delta_n^3)$. Using the arguments for obtaining \eqref{eq:eigenfunction-boudary-error-deriv}, we arrive at the second equation.
\end{proof}

\begin{proof}[Proof of Lemma \ref{prop:conv-unn}]
 We can construct $u^-(q, x, b)$ as
\begin{align}
	&u^-(q, x, b) = \frac{\psi^+(r) \psi^-(x) - \psi^+(x) \psi^-(r)}{\psi^+(r) \psi^-(b) - \psi^+(b) \psi^-(r)},
\end{align}
where $\psi^+$  and $\psi^-$ are given in the proof of Lemma \ref{lmm:prior-u1p}. From this expression, it is easy to deduce $u^-(q, x, b)$ is $C^3$ in $b$.  
Direct calculations show that
	\begin{align}
		\partial_b u^-(q, b, b) &= -\frac{\psi^+(r) (\psi^-)'(b) - (\psi^+)'(b) \psi^-(r)}{\psi^+(r) \psi^-(b) - \psi^+(b) \psi^-(r)},\\
		\partial_{bb} u^-(q, x, b) &= -\frac{\psi^+(r) (\psi^-)''(b) - (\psi^+)''(b) \psi^-(r)}{\psi^+(r) \psi^-(b) - \psi^+(b) \psi^-(r)} + 2\frac{(\psi^+(r) (\psi^-)'(b) - (\psi^+)'(b) \psi^-(r))^2}{(\psi^+(r) \psi^-(b) - \psi^+(b) \psi^-(r))^2}.
	\end{align}
Using these equations, we can verify 
	\begin{align}
		&\mu(L) \partial_b u^-(q, L, L) - \sigma^2(L)  (\partial_b u^-(q, L, L))^2 + \frac{1}{2} \sigma^2(L)  \partial_{bb} u^-(q, L, L) + q \\
		&= -\frac{\psi^+(r) \left( \mu(b) (\psi^-)'(b) + \frac{1}{2} \sigma^2(b) (\psi^-)''(b) - q\psi^-(b) \right)  }{\psi^+(r) \psi^-(b) - \psi^+(b) \psi^-(r)} \\
		& \quad + \frac{\psi^-(r) \left( \mu(b) (\psi^+)'(b) + \frac{1}{2} \sigma^2(b) (\psi^+)''(b) - q\psi^+(b) \right) }{\psi^+(r) \psi^-(b) - \psi^+(b) \psi^-(r)} = 0.
	\end{align}
	
	We have the following facterization for $u_n^-(q, x; L^-)$:
	\begin{align}
		u_n^-(q, x; L^-) &= \mathbb{E}_x\left[ e^{q\tau_L^-} 1_{\{ Y_{\tau_L^-} = L^- \}} \right] \\
		&= \mathbb{E}_x\left[ e^{q\tau_{L^{++}}^-} 1_{\{ Y_{\tau_{L^{++}}^-} = L^+ \}} \right] \mathbb{E}_{L^+}\left[ e^{q\tau_L^-} 1_{\{ Y_{\tau_L^-} = L^- \}} \right] = \widetilde{u}_n^-(q, x; L^+) u_n^-(q, L^+; L^-).
	\end{align}
	To derive it, we use the fact that for $Y^{(n)}$ to arrive at $L^-$ from the above, it should first touch $L^+$, because it is a birth-and-death process.
	
	Using the smoothness of $u^-(q, x; b)$ in $b$ and $u_n^-(q, x; L^+) = u^-(q, x; L^+) + \mathcal{O}(\delta_n^2)$, we get
	\begin{align}
		u_n^-(q, x; L^+) = u^-(q, x, L) + \mathcal{O}(L^+ - L) + \mathcal{O}(\delta_n^2).
	\end{align}
	Similar to the proof of Lemma \ref{lmm:prior-u1p}, we can show $u_n^-(q, L^+; L^-) = u^-(q, L^+, L^-) + \mathcal{O}(\delta_n^3)$. Thus, we can develop the following estimate for $u^-(q, L^+, L^-)$:
	\begin{align}
		&u^-(q, L^+, L^-) - 1 \\
		&= u^-(q, L^+, L^-) - u^-(q, L^+, L) + u^-(q, L^+, L)  - u^-(q, L^+, L^+) \\
		&= \partial_b u^-(q, L^+, L)(L^- - L) + \frac{1}{2} \partial_{bb} u^-(q, L^+, L)(L^- - L)^2 \\
		&\quad - \partial_b u^-(q, L^+, L)(L^+ - L) - \frac{1}{2} \partial_{bb} u^-(q, L^+, L)(L^+ - L)^2 + \mathcal{O}(\delta_n^3) \\
		&= -\partial_b u^-(q, L^+, L) \delta^+ L^- + \frac{1}{2} \partial_{bb} u^-(q, L^+, L)(\delta^+ L^-)^2 \\
		&\quad - \frac{1}{2} \partial_{bb} u^-(q, L^+, L) \delta^+ L^- (L^+ - L) + \mathcal{O}(\delta_n^3) \\
		&= -\partial_b u^-(q, L^+, L) \delta^+ L^- + \frac{1}{2} \partial_{bb} w^-(q, L^+, L)(\delta^+ L^-)^2 + \mathcal{O}(\delta^+ L^-) (L^+ - L) + \mathcal{O}(\delta_n^3).
	\end{align}
	This concludes the proof.
\end{proof}

\begin{proof}[Proof of Proposition \ref{prop:conv-vn}]
	By Lemma \ref{lmm:priors-eigenfunctions}, there exist constants $c_1, c_2 > 0$ such that
	\begin{align}
		&v_n(D, x; y) / (m_n(y) \delta y) \\
		&= \sum_{k = 1}^{n_e} e^{-\lambda_{n, k}^+ D} \varphi_{n, k}^+(q, x) \varphi_{n, k}^+(y) \\
		&= \sum_{k = 1}^{n_e} e^{-\lambda_k^+(L)D} \varphi_k^+(x, L) \varphi_k^+(y, L) + \left(  \mathcal{O}(L^+ - L) + \mathcal{O}(\delta_n^2) \right)\sum_{k = 1}^{n_e} k^{11} e^{-c_1k^2 D} \\
		&= \sum_{k = 1}^{\infty} e^{-\lambda_k^+(L)} \varphi_k^+(x, L) \varphi_k^+(y, L) + \mathcal{O}\left(\sum_{k = n_e + 1}^{\infty} e^{-c_2 k^2 D}\right) \\
		&\quad + \left(  \mathcal{O}(L^+ - L) + \mathcal{O}(\delta_n^2) \right)\sum_{k = 1}^{\infty} k^{11} e^{-c_1k^2 D} \\
		&= \bar{v}(D, x, y) + \mathcal{O}(L^+ - L) + \mathcal{O}(\delta_n^2).
	\end{align}
	In the last equality, we use the fact that $\sum_{k = n_e + 1}^{\infty} e^{-c_2 k^2 D} \le c_3 e^{-c_4/\delta_n^2} \le c_5 \delta_n^2$ for some constants $c_3, c_4, c_5 > 0$. Later on, we use  this result directly. Using \eqref{eq:eigenfunction-error-boundary}, we also obtain 
	\begin{align}
		&v_n(D, L^-; y) / (m_n(y) \delta y)\\
		&= \sum_{k = 1}^{n_e} e^{-\lambda_{n, k}^+ D} \varphi_{n, k}^+(L^-) \varphi_{n, k}^+(y) \\
		&= \sum_{k = 1}^{n_e} e^{-\lambda_k^+(L) D} \varphi_{n, k}^+(L^-) \varphi_k^+(y, L) + \mathcal{O}((L^+ - L)\delta^+L^-) + \mathcal{O}(\delta_n^3) \\
		&= -\delta^+ L^-\sum_{k = 1}^{n_e} e^{-\lambda_k^+(L) D} \partial_x \varphi^+_k(L, L) \varphi_k^+(y, L) \\
		&\quad + \frac{1}{2} (\delta^+ L^-)^2 \sum_{k = 1}^{n_e} e^{-\lambda_k^+(L) D} \partial_{xx} \varphi^+_k(L, L) \varphi_k^+(y, L) + \mathcal{O}((L^+ - L)\delta^+L^-) + \mathcal{O}(\delta_n^3) \\
		&=  -\delta^+ L^-\sum_{k = 1}^{\infty} e^{-\lambda_k^+(L) D} \partial_x \varphi^+_k(L, L) \varphi_k^+(q, y; L) \\
		&\quad + \frac{1}{2} (\delta^+ L^-)^2 \sum_{k = 1}^{\infty} e^{-\lambda_k^+(L) D} \partial_{xx} \varphi^+_k(L, L) \varphi_k^+(y, L) + \mathcal{O}((L^+ - L)\delta^+L^-) + \mathcal{O}(\delta_n^3) \\
		&= - \bar{v}_x(D, L; y) \delta^+ L^- + \frac{1}{2} \partial_{xx} \bar{v}(D, L; y) ( \delta^+ L^-)^2 + \mathcal{O}((L^+ - L)\delta^+L^-) + \mathcal{O}(\delta_n^3).
	\end{align}
	
	For the last claim, we note that,
	\begin{align}
		\mu(L)\partial_x \varphi_k^+(L, L) + \frac{1}{2} \sigma^2(L) \partial_{xx} \varphi_k^+(L, L) - q \varphi_k^+(L, L) = -\lambda_k(L) \varphi_k^+(L, L). 
	\end{align}
	Noting that $\varphi_k^+(L,L) = 0$, we have $\mu(L)\partial_x \varphi_k^+(L, L) + \frac{1}{2} \sigma^2(L) \partial_{xx} \varphi_k^+(L, L) = 0$. Then
	\begin{align}
		&\mu(L) \partial_x \bar{v}(D, L; y) + \frac{1}{2} \sigma^2(L) \partial_{xx} \bar{v}(D, L; y) \\
		&= \sum_{k = 1}^\infty e^{-\lambda_k^+(L) D} \left( \mu(L)\partial_x \varphi_k^+(L, L) + \frac{1}{2} \sigma^2(L) \partial_{xx} \varphi_k^+(L, L) \right) = 0,
	\end{align}
	where the interchange of summation and differentiation can be verified with the estimates of $\lambda_k^+(L)$ and $\varphi_k^+(L, L)$ in Lemma \ref{lmm:priors-eigenfunctions}.
	
	The results for $v_n(D, x; l)$ and $v(D, L, L)$ can be proved by arguments similar to those for the third part of Lemma \ref{lmm:prior-u1p} and the equation $v(D, x, L) = v_1(x, L) - v_2(D, x, L)$ along with the differential equation for $v_1(x, L)$ and eigenfunction expansion for $v_2(D, x, L)$.
\end{proof}

\begin{proof}[Proof of Proposition \ref{prop:conv-unp}]
	For the first and second claims, we only prove that $c_{n, k}(q) = c_k(q, L) + O(k^4 \delta_n^2)$. The other steps are almost identical to those in the proof of Proposition \ref{prop:conv-vn}. 
	\begin{align}
		&c_{n, k}(q) - c_k(q, L) \\
		&= \sum_{y \in \mathbb{S}_n^o \cap (-\infty, L^+)} \varphi_{n, k}^+(y) u_{1, n}^+(q, y; L^+) m_n(y) \delta y -  \int_{l}^{L} \varphi_k^+(y, L) u_1^+(q, y, L) m(y) dy \\
		&=  \sum_{y \in \mathbb{S}_n^o \cap (-\infty, L^+)} \varphi_k^+(y, L) u_1^+(q, y, L) m_n(y) \delta y -  \int_{l}^{L} \varphi_k^+(y, L) u_1^+(q, y, L) m(y) dy \\
		&\quad + \mathcal{O}(k^4 \delta_n^2) + \mathcal{O}(k^2 (L^+ - L))\\
		&= \sum_{y \in \mathbb{S}_n^o \cap (-\infty, L^+)} \varphi_k^+(y, L) u_1^+(q, y, L) m(y) \delta y -  \int_{l}^{L} \varphi_k^+(y, L) u_1^+(q, y, L) m(y) dy \\
		&\quad + \mathcal{O}(k^4 \delta_n^2) + \mathcal{O}(k^2 (L^+ - L)) \\
		&= \sum_{z \in \mathbb{S}_n^- \cap (-\infty, L^-)} \left(\frac{1}{2} \left( \varphi_k^+(z, L) u_1^+(q, z, L) m(z) + \varphi_k^+(z^+, L) u_1^+(q, z^+, L)) m(z^+) \right) \delta^+ z \right. \\
		&\quad \quad \quad \left. - \int_{z}^{z^+} \varphi_k^+(y, L) u_1^+(q, y, L) m(y) dy\right) + \int_{L^-}^{L} \varphi_k^+(y, L) u_1^+(q, y, L) m(y) dy \\
		&\quad + \mathcal{O}(k^4 \delta_n^2) + \mathcal{O}(k^2 (L^+ - L)) \\
		&= \mathcal{O}(k^4 \delta_n^2) + \mathcal{O}(k^2 (L^+ - L)).
	\end{align}
	
	For the last claim, using the arguments in the proof of Proposition \ref{prop:conv-vn},  we obtain $\mu(L) \partial_x u_2^+(q, D, L; L) + \frac{1}{2} \sigma^2(L) \partial_{xx} u^+_2(q, D, L; L) = 0$.  Hence, there holds that
	\begin{align}
		&\mu(L) \partial_x u^+(q, D, L, L) + \frac{1}{2} \sigma^2(L) \partial_{xx} u^+(q, D, L, L) - q \\
		&= \mu(L) \partial_x u_1^+(q, L, L) + \frac{1}{2} \sigma^2(L) \partial_{xx} u_1^+(q, L, L) - qu_1^+(q, L, L) = 0,
	\end{align}
where we use the equation of $u_1^+(q, \cdot, L)$ at the boundary and $u_1^+(q, L, L) = 1$.
\end{proof}

\begin{proof}[Proof of Proposition \ref{prop:conv-wn}]
	
	By Theorem 1.4 in \cite{athanasiadis1996some}, \eqref{eq:w-tilde-equation} admits a unique solution $w(\cdot)$ that belongs to $C^1([l, r]) \cap C^2([l, r] \backslash \mathcal{D})$. The equation for $\widetilde{w}(q, z)$ can be written in a self-adjoint form as
	\begin{align}
		\frac{1}{m(x)} \partial_x\left( \frac{1}{s(x)} \partial_x \widetilde{w}(q, x)\right) - q \widetilde{w}(q, x) = f(x).
	\end{align}
	Multiplying both sides with $m(x)$ and integrating from $l^+_{1/2} = l + \delta^+ l/2$ to $z \in \mathbb{S}_n$ yield
	\begin{align}
		\frac{1}{s(z)} \partial_x \widetilde{w}(q, z) - \frac{1}{s(l^+_{1/2})} \partial_x \widetilde{w}(q, l^+_{1/2}) - q \int_{l^+_{1/2}}^{z} m(y) \widetilde{w}(q, y) dy dz = \int_{l^+_{1/2}}^{z} m(y) f(y) dy.
	\end{align}
	Further multiplying both sides with $s(z)$ and integrating from $x$ to $x^+$ give us
	\begin{align}
		&\widetilde{w}(q, x^+) - \widetilde{w}(q, x) - \frac{\int_{x}^{x^+} s(z) dz}{s(l^+_{1/2})} \partial_x \widetilde{w}(q, l^+_{1/2}) - q \int_{x}^{x^+} s(z) \int_{l^+_{1/2}}^{z} m(y) \widetilde{w}(q, y) dy dz \\
		&= \int_{x}^{x^+} s(z) \int_{l^+_{1/2}}^{z} m(y) f(y) dy dz.
	\end{align}
	Moreover, it is clear that $\widetilde{w}_n(q, z)$ satisfies
	\begin{align}
		\frac{\delta^- x}{m_n(x) \delta x} \nabla^-\left( \frac{1}{s_n(x)} \nabla^+ \widetilde{w}_n(q, x) \right) - q \widetilde{w}_n(q, x) = f(x).
	\end{align}
	Multiplying both sides with $m_n(x) \delta x$ and summing from $l^+$ to $x$, we obtain
	\begin{align}
		\frac{1}{s_n(x)} \nabla^+ \widetilde{w}_n(q, x) - \frac{1}{s_n(l)} \nabla^+ \widetilde{w}_n(q, l) - q \sum_{l < y \le x} \widetilde{w}_n(q, y) m_n(y) \delta y = \sum_{l < y \le x} f(y) m_n(y) \delta y.
	\end{align}
	It follows that
	\begin{align}
		&\widetilde{w}_n(q, x^+) - \widetilde{w}_n(q, x) - s_n(x)  \frac{\delta^+ x}{s_n(l)} \nabla^+ \widetilde{w}_n(q, l) - q s_n(x) \delta^+ x\sum_{l < y \le x} \widetilde{w}_n(q, y) m_n(y) \delta y \\
		&= s_n(x) \delta^+ x \sum_{l < y \le x} f(y) m_n(y) \delta y. 
	\end{align}
	Let $e(x) = \widetilde{w}_n(q, x) - \widetilde{w}(q, x)$. We have
	\begin{align}
		e(x^+) - e(x) &= s_n(x) \delta^+ x \left( \frac{1}{s_n(l^+_{1/2})} \nabla^+ \widetilde{w}_n(q, l) - \frac{1}{s(l^+_{1/2})} \partial_x \widetilde{w}(q, l) \right) \label{eq:exp-ex-1}\\
		&\quad + \left( s_n(x) \delta^+ x - \int_{x}^{x^+} s(z) dz  \right) \frac{1}{s(l^+_{1/2})} \partial_x \widetilde{w}(q, l) \label{eq:exp-ex-2} \\
		&\quad + q s_n(x) \delta^+ x \sum_{l < y \le x} e(y) m_n(y) \delta y \label{eq:exp-ex-3} \\
		&\quad + q \left(s_n(x) \delta^+ x - \int_{x}^{x^+} s(z) dz\right) \sum_{l < y \le x} \widetilde{w}(q, y) m_n(y) \delta y \label{eq:exp-ex-4} \\
		&\quad + q \int_{x}^{x^+} s(z) \left( \sum_{l < y \le x} \widetilde{w}(q, y) m_n(y) \delta y - \int_{l^+_{1/2}}^{z} \widetilde{w}(q, y) m(y) dy \right) dz \label{eq:exp-ex-5}\\
		&\quad + \left(s_n(x) \delta^+ x - \int_{x}^{x^+} s(z) dz\right) \sum_{l < y \le x} f(y) m_n(y) \delta y \label{eq:exp-ex-6} \\
		&\quad + \int_{x}^{x^+} s(z) \left( \sum_{l < y \le x} f(y) m_n(y) \delta y - \int_{l^+_{1/2}}^{z} f(y) m(y) dy \right) dz. \label{eq:exp-ex-7}
	\end{align}
	The quantities in \eqref{eq:exp-ex-2}, \eqref{eq:exp-ex-4} and \eqref{eq:exp-ex-6} are all $\mathcal{O}(\delta_n^3)$. For the quantity in \eqref{eq:exp-ex-7}, note that $s(z) = s(x) + s'(x)(z - x) + \mathcal{O}(\delta_n^2)$, $\int_{l^+_{1/2}}^{z} f(y) m(y) dy = \int_{l^+_{1/2}}^{x^+_{1/2}} f(y) m(y) dy + f(x) m(x) (z - x^+_{1/2}) + \mathcal{O}(\delta_n^\gamma 1_{\{ x \in \mathcal{D}_N \}} + \delta_n^2)$ with $x^+_{1/2} = x + \delta^+x/2$ for $z \in [x, x^+]$, $\mathcal{D}_N = \{ y \in \mathbb{S}_n^o: [y^-, y^+] \cap \mathcal{D} \ne \emptyset \}$, and $\int_{l^+_{1/2}}^{x^+_{1/2}} f(y) m(y) dy = \sum_{l < y \le x} f(y) m_n(y) dy + \mathcal{O}(\delta_n^\gamma)$. Therefore,
	\begin{align}
		\eqref{eq:exp-ex-7} &= s(x) \delta^+ x \left( \sum_{l < y \le x} f(y) m_n(y) dy - \int_{l^+_{1/2}}^{x^+_{1/2}} f(y) m(y) dy \right) \\
		&\quad + s(x) f(x) m(x) \int_{x}^{x^+} (x^+_{1/2} - z) dz + \mathcal{O}(\delta_n^{1 + \gamma} 1_{\{ x\in \mathcal{D}_N \}} + \delta_n^3 + \delta_n^{1 + \gamma}) \\
		&= \mathcal{O}(\delta_n^{1 + \gamma} 1_{\{ x\in \mathcal{D}_N \}} + \delta_n^3).
	\end{align}
	By the same argument, we can show that $\eqref{eq:exp-ex-5} = \mathcal{O}(\delta_n^3)$. Putting these estimates back and letting $e^\ast(x) = -e(x)$, we deduce that there exists a constant $c_1 > 0$ independent of $\delta_n$ such that
	\begin{align}
		e(x^+) &\le s_n(x) \delta^+ x \left( \frac{1}{s_n(l^+_{1/2})} \nabla^+ \widetilde{w}_n(q, l) - \frac{1}{s(l^+_{1/2})} \partial_x \widetilde{w}(q, l) \right) \\
		&\quad + e(x) + q s_n(x) \delta^+ x \sum_{l < y \le x} e(y) m_n(y) \delta y  + c_1 (\delta_n^{1 + \gamma} 1_{\{ x\in \mathcal{D}_N \}} + \delta_n^3),\\
		e^\ast(x^+) &\le -s_n(x) \delta^+ x \left( \frac{1}{s_n(l^+_{1/2})} \nabla^+ \widetilde{w}_n(q, l) - \frac{1}{s(l^+_{1/2})} \partial_x \widetilde{w}(q, l) \right) \\
		&\quad + e^\ast(x) + q s_n(x) \delta^+ x \sum_{l < y \le x} e^\ast(y) m_n(y) \delta y  + c_1 (\delta_n^{1 + \gamma} 1_{\{ x\in \mathcal{D}_N \}} + \delta_n^3).  
	\end{align}
	Noting the positive lower and upper bounds for $s_n(x)$, $m_n(x)$ which are independent of $\delta_n$ and using the discrete Gronwall's inequality, we conclude that there exist constants $c_2, c_3, c_4, c_5 > 0$ independent of $\delta_n$ such that
	\begin{align}
		e(x) &\le c_2\left( c_1 h_n^\gamma + c_3 \left( \frac{1}{s_n(l^+_{1/2})} \nabla^+ \widetilde{w}_n(q, l) - \frac{1}{s(l^+_{1/2})} \partial_x \widetilde{w}(q, l) \right) \right), \label{eq:ex-bound}\\
		e^\ast(x) &\le c_4\left( c_1 h_n^\gamma - c_5 \left( \frac{1}{s_n(l^+_{1/2})} \nabla^+ \widetilde{w}_n(q, l) - \frac{1}{s(l^+_{1/2})} \partial_x \widetilde{w}(q, l) \right) \right). \label{eq:exs-bound}
	\end{align}
	Note that $e(r) = e^\ast(r) = 0$. Then, there exist constants $c_6, c_7 > 0$ independent of $\delta_n$ such that
	\begin{align}
		&\left( \frac{1}{s_n(l^+_{1/2})} \nabla^+ \widetilde{w}_n(q, l) - \frac{1}{s(l^+_{1/2})} \partial_x \widetilde{w}(q, l) \right) \le c_6 \delta_n^\gamma,\ \\
		&-\left( \frac{1}{s_n(l^+_{1/2})} \nabla^+ \widetilde{w}_n(q, l) - \frac{1}{s(l^+_{1/2})} \partial_x \widetilde{w}(q, l) \right) \le c_7 \delta_n^\gamma.
	\end{align}
	Hence $\left| \frac{1}{s_n(l^+_{1/2})} \nabla^+ \widetilde{w}_n(q, l) - \frac{1}{s(l^+_{1/2})} \partial_x \widetilde{w}(q, l) \right| = \mathcal{O}(\delta_n^\gamma)$. Putting them back to \eqref{eq:ex-bound} and \eqref{eq:exs-bound}, we have $e(x) = \mathcal{O}(\delta_n^\gamma)$ and $-e(x) = \mathcal{O}(\delta_n^\gamma)$. Therefore, $e(x) = \mathcal{O}(\delta_n^\gamma)$ holds and the proof is completed.
\end{proof}

\begin{proof}[Proof of Theorem \ref{thm:conv-h}]
	The smoothness of $h(q, x; y)$ and its limit at $y = l$ and value at $y = L$ are direct consequences of the properties of $v(D, x; y)$. By \eqref{eq:vn-L-conv}, \eqref{eq:unp-L-conv} and \eqref{eq:conv-umn-L}, we have
	\begin{align}
		&u_n^-(q, L^+; L^-) v_n(D, L^-; y) / (m_n(y) \delta y) \\
		&= \left(- \bar{v}_x \delta^+ L^- + \frac{1}{2} \bar{v}_{xx} (\delta^+ L^-)^2  + \mathcal{O}(\delta^+ L^-)(L^+ - L) \right) \\
		&\quad \times \left(1 - u^-_b \delta^+ L^- + \frac{1}{2} u^-_{bb} (\delta^+ L^-)^2 + \mathcal{O}(\delta^+ L^-)(L^+ - L) \right) + \mathcal{O}(\delta_n^3) \\
		&= -\bar{v}_x \delta^+L^- + (\delta^+L^-)^2 \left( \bar{v}_x u_b^- + \frac{1}{2} \bar{v}_{xx}  \right) + \mathcal{O}(\delta^+ L^-) (L^+ - L) + \mathcal{O}(\delta_n^3).
	\end{align}
	Here we neglect the arguments $(D, L; y)$ for $\bar{v}$ and $(q, L, L)$ for $u^-$ to make the equations shorter. Moreover, using \eqref{eq:unp-L-conv} and \eqref{eq:conv-umn-L}, we obtain
	\begin{align}
		&1 - u_n^+(q,D, L^-; L^+) u_n^-(q, L^+; L^-) \\
		&= 1 - \left(1 - u^+_x \delta^+ L^- + \frac{1}{2} u^+_{xx} (\delta^+ L^-)^2 + \mathcal{O}(\delta^+ L^-)(L^+ - L)  \right) \\
		&\quad \quad \quad \times \left(1 - u^-_b \delta^+ L^- + \frac{1}{2} u^-_{bb} (\delta^+ L^-)^2 + \mathcal{O}(\delta^+ L^-)(L^+ - L) \right) + \mathcal{O}(\delta_n^3) \\
		&= \delta^+ L^- \left( u_b^- + u_x^+ \right) - (\delta^+ L^-)^2 \left( \frac{1}{2} u_{bb}^- + u_x^+ u_b^- + \frac{1}{2} u_{xx}^+  \right) + \mathcal{O}(\delta^+ L^-) (L^+ - L) + \mathcal{O}(\delta_n^3).
	\end{align}
	Here we neglect the arguments $(q, D, x, L)$ for $u^+$ for the same reason. It follows that
	\begin{align}
		&h_n(q, L^+; y)/(m_n(y)\delta y) \\
		&= e^{-qD} \frac{-v_x + \delta^+L^- \left( \bar{v}_x u_b^- + \frac{1}{2} \bar{v}_{xx}  \right) + \mathcal{O}(L^+ - L) + \mathcal{O}(\delta_n^2)}{  u_b^- + u_x^+  - \delta^+ L^- \left( \frac{1}{2} u_{bb}^- + u_x^+ u_b^- + \frac{1}{2} u_{xx}^+  \right) + \mathcal{O} (L^+ - L) + \mathcal{O}(\delta_n^2)} \\
		&= h(q, L; y) / m(y) + e^{-qD} \delta^+ L^-\frac{-\frac{1}{2} \bar{v}_x u_{bb}^- - \frac{1}{2} \bar{v}_x u_{xx}^+ + v_x^- (u_b^- )^2 + \frac{1}{2} \bar{v}_{xx} u_b^- + \frac{1}{2} u_x^+ \bar{v}_{xx} }{\left( u_b^- + u_x^+ \right)^2 + \mathcal{O}(\delta_n)} \\
		&\quad + \mathcal{O}(L^+ - L) + \mathcal{O}(\delta_n^2).
	\end{align}
	By \eqref{eq:up-boundary} and \eqref{eq:v-bar-boundary}, we get $\bar{v}_{xx} = -\frac{2\mu(L)}{\sigma^2(L)} \bar{v}_x$ and $u^+_{xx} = -\frac{2\mu(L)}{\sigma^2(L)} u^+_x + \frac{2q}{\sigma^2(L)}$. Subsequently, we obtain
	\begin{align}
		&-\frac{1}{2} \bar{v}_x u_{bb}^- - \frac{1}{2} \bar{v}_x u_{xx}^+ + v_x^- (u_b^- )^2 + \frac{1}{2} \bar{v}_{xx} u_b^- + \frac{1}{2} u_x^+ \bar{v}_{xx} \\
		&= -\frac{1}{2} \bar{v}_x u_{bb}^-- \frac{1}{2} \bar{v}_x \left( -\frac{2\mu(L)}{\sigma^2(L)} u^+_x + \frac{2q}{\sigma^2(L)} \right) + v_x^- (u_b^- )^2 - \frac{1}{2} \frac{2\mu(L)}{\sigma^2(L)} \bar{v}_x u_b^- - \frac{1}{2} \frac{2\mu(L)}{\sigma^2(L)} \bar{v}_x u_x^+ \\
		&= -\frac{1}{\sigma^2(L) }\left( \frac{1}{2} \sigma^2(L) u_{bb}^- + q - \sigma^2(L) (u_b^-)^2 + \mu(L) u_b^- \right) = 0.
	\end{align}
	The last equality follows from \eqref{eq:wm-b-boundary}. Therefore, there holds that
	\begin{align}
		h_n(q, L^+; y)/(m_n(y)\delta y) = h(q, L; y) / m(y) + \mathcal{O}(L^+ - L) + \mathcal{O}(\delta_n^2).
	\end{align}
	Then, by \eqref{eq:conv-umn-x}, we obtain
	\begin{align}
		h_n(q, L^-; y) u_n^-(q, x; L^-) / (m_n(y) \delta y)  &= h_n(q, L^-; y) u_n^-(q, x; L^+) u_n^-(q, L^+; L^-) / (m_n(y) \delta y) \\
		&= h_n(q, L^+; y) u_n^-(q, x; L^+) / (m_n(y) \delta y) \\
		&= h(q, L; y) u_n^-(q, x; L^+) / m(y) + \mathcal{O}(L^+ - L) + \mathcal{O}(\delta_n^2) \\
		&= h(q, L; y) u^-(q, x, L) / m(y) + \mathcal{O}(L^+ - L) + \mathcal{O}(\delta_n^2).
	\end{align}
	Based on the previous estimates, we deduce
	\begin{align}
		&h_n(q, x; y) / (m_n(y) \delta y) \\
		&= 1_{\{ x< L \}} e^{-qD} v_n(D, x; y) / (m_n(y) \delta y) + 1_{\{ x < L \}} u_n^+(q, D, x; L^+) h_n(q, L^+; y) / (m_n(y) \delta y) \\
		& \quad + 1_{\{ x \ge L \}} u^-_n(q, x; L^-) h_n(q, L^-; y)  / (m_n(y) \delta y) \\
		&= 1_{\{ x< L \}} e^{-qD} \bar{v}(D, x; y) + 1_{\{ x < L \}} u^+(q, D, x, L) h(q, L; y) / m(y) \\
		&\quad + 1_{\{ x \ge L \}} h(q, L; y) u^-(q, x, L) / m(y) + \mathcal{O}(L^+ - L) + \mathcal{O}(\delta_n^2).
	\end{align}	
\end{proof}

\begin{proof}[Proof of Theorem \ref{thm:conv-u}] 
	Recall that $h_n(q, x; l) = h(q, x; l) + \mathcal{O}(L^+ - L) + \mathcal{O}(\delta_n^2)$ and $\widetilde{w}(q, l) = \widetilde{w}_n(q, l) = f(l)/q$, which will be used below. We have
	\begin{align}
		&\widetilde{u}_n(q, x) - \widetilde{u}(q, x) \\
		&= \sum_{z \in \mathbb{S}_n \cap (l, L)} h_n(q, x; z) \widetilde{w}_n(q, z) - \int_{l}^{L} h(q, x; z) \widetilde{w}(q, z) dz + h_n(q, x; l) \widetilde{w}_n(q, l) - h(q, x; l) \widetilde{w}(q, l)\\
		&= \sum_{z \in \mathbb{S}_n \cap (l, L)} h_n(q, x; z) \widetilde{w}_n(q, z) - \int_{l}^{L} h(q, x; z) \widetilde{w}(q, z) dz + \mathcal{O}(L^+ - L) + \mathcal{O}(\delta_n^2)\\
		&= \sum_{z \in \mathbb{S}_n \cap (l, L)} h_n(q, x; z) /(m_n(z) \delta z) \widetilde{w}_n(q, z) m_n(z) \delta z - \int_{l}^{L} h(q, x; z) \widetilde{w}(q, z) dz + \mathcal{O}(L^+ - L) + \mathcal{O}(\delta_n^2) \\
		&= \sum_{z \in \mathbb{S}_n \cap (l, L)} h(q, x; z) \widetilde{w}(q, z) m_n(z) / m(z) \delta z - \int_{l}^{L} h(q, x; z) \widetilde{w}(q, z) dz + \mathcal{O}(L^+ - L) + \mathcal{O}(\delta_n^\gamma) \\
		&= \sum_{z \in \mathbb{S}_n \cap (l, L)} h(q, x; z) \widetilde{w}(q, z) \delta z - \int_{l}^{L} h(q, x; z) \widetilde{w}(q, z) dz \\
		&\quad + \frac{1}{2}\sum_{z \in \mathbb{S}_n \cap (l, L)} h(q, x; z) \widetilde{w}(q, z) \frac{\mu(z)}{\sigma^2(z) } \left( (\delta^+ z)^2 - (\delta^- z)^2 \right) + \mathcal{O}(L^+ - L) + \mathcal{O}(\delta_n^\gamma) \\
		&= \sum_{z \in \mathbb{S}_n^- \cap (l, L^-)} \left(\frac{1}{2} (h(q, x; z) + h(q, x; z^+)) \delta^+ z - \int_{z}^{z^+} h(q, x; y) \widetilde{w}(q, y) dy \right) \\
		&\quad + h(q, x; l^+) \widetilde{w}(q, l^+) \delta^- l^+/2 + h(q, x; L^-) \widetilde{w}(q, L^-) \delta^+ L^-/2 - \int_{y \in (l, l^+) \cup (L^-, L)} h(q, x; y) \widetilde{w}(q, y) dy \\
		&\quad + \frac{1}{2} \sum_{z \in \mathbb{S}_n^- \cap (l, L^-)} \left( h(q, x; z) \widetilde{w}(q, z) \frac{\mu(z)}{\sigma^2(z) } - h(q, x; z^+) \widetilde{w}(q, z^+) \frac{\mu(z^+)}{\sigma^2(z^+) } \right) (\delta^+ z)^2 \\
		&\quad - \frac{1}{2} h(q, x; l^+) \widetilde{w}(q, l^+) \frac{\mu(l^+)}{\sigma^2(l^+) } (\delta^-l^+)^2 + \frac{1}{2} h(q, x; L^-) \widetilde{w}(q, L^-) \frac{\mu(L^-)}{\sigma^2(L^-) } (\delta^+L^-)^2 \\
		&\quad + \mathcal{O}(L^+ - L) + \mathcal{O}(\delta_n^\gamma) \\
		&= \sum_{z \in \mathbb{S}_n \cap (l, L^-), [z, z^+] \mathcal{D} = \emptyset} \max_{y \in [z, z^+]} \partial_{zz} \left( h(q, x; \cdot) \widetilde{w}(q, \cdot)  \right) (y) \mathcal{O}\left( \delta_n^3 \right) \\
		&\quad + \sum_{z \in \mathbb{S}_n \cap (l, L^-), [z, z^+] \mathcal{D} \ne \emptyset} \max_{y \in [z, z^+]} \partial_z \left( h(q, x; \cdot) \widetilde{w}(q, \cdot)  \right) (y) \mathcal{O}\left( \delta_n^2 \right) + \mathcal{O}(L^+ - L) + \mathcal{O}(\delta_n^\gamma) \\
		&=\mathcal{O}(L^+ - L) + \mathcal{O}(\delta_n^\gamma).
	\end{align}
	In the last to the second equality, we use the error estimate of the trapezoid rule and the smoothness of $h(q, x; z)$ and $\widetilde{w}(q, z)$.
\end{proof}

\bibliographystyle{chicagoa}
\bibliography{Parisian,mdd,EEFD,lookback_lqmc}

\end{document}